\documentclass[aps,prx,twocolumn,showpacs,superscriptaddress,floatfix,10pt,longbibliography]{revtex4-2}
\pdfoutput=1
\usepackage{xpatch}
\usepackage{bm,bbm}
\usepackage{mathrsfs}
\usepackage{graphicx}
\usepackage{color}
 \usepackage{hyperref}
\usepackage{amsmath}
\usepackage{amssymb}
\usepackage{mathtools}
\usepackage{setspace}
\usepackage{sidecap}
\usepackage{booktabs}
\usepackage{amsthm}
\usepackage{url}
\usepackage{textcomp}
\usepackage{amsfonts}
\usepackage{amsthm}
\include{cite}
\usepackage{bbm}
\usepackage{mathtools}
\usepackage[english]{babel}
\usepackage[toc,page]{appendix}

\newcommand{\AI}{\Upsilon}
\newcommand{\eqqref}[1]{Eq.~\eqref{#1}}

\newtheorem{proposition}{Proposition}
\newtheorem{definition}{Definition}
\newtheorem{theorem}{Theorem}
\newtheorem{observation}{Observation}
\newtheorem{corollary}{Corollary}[theorem]
\newtheorem{remark}{Remark}[proposition]
\newtheorem{lemma}[theorem]{Lemma}

 \makeatletter
 \patchcmd{\@ssect@ltx}
     {\addcontentsline{toc}{#1}{\protect\numberline{}#8}}
     {}
     {}
     {}
     \makeatother
     
\begin{document}
\title{Ubiquitous Dynamical Time Asymmetry in Measurements on Materials and Biological Systems}

\author{Alessio Lapolla}
\affiliation{Mathematical bioPhysics Group, Max Planck Institute for
  Biophysical Chemistry, G\"{o}ttingen 37077, Germany}

\author{Jeremy~C. Smith}
\affiliation{Center for Molecular Biophysics, Oak Ridge National
  Laboratory, Oak Ridge, Tennessee 37830, USA}
\affiliation{
Department of Biochemistry and Cellular and Molecular Biology, University of Tennessee, Knoxville, Tennessee 37996, USA}

\author{Alja\v{z} Godec}
\email{agodec@mpibpc.mpg.de}
\affiliation{Mathematical bioPhysics Group, Max Planck Institute for Biophysical Chemistry, G\"{o}ttingen 37077, Germany}

\begin{abstract}
Many measurements on soft condensed matter (e.g., biological and
materials) systems track low-dimensional observables
projected from the full system phase space as a function of
time. Examples are dynamic structure factors, spectroscopic and
rheological response functions, and time series of distances derived
from optical tweezers, single-molecule spectroscopy and molecular
dynamics simulations. In many such systems the projection renders the
reduced dynamics non-Markovian and the observable is not prepared in,
or initially sampled from and averaged over,
a stationary distribution. We prove that such systems always exhibit
non-equilibrium, time asymmetric dynamics. That is, they evolve in time with a
broken time-translation invariance in a manner closely resembling aging dynamics. We
  identify the entropy associated with the breaking of time-translation
  symmetry that is a
  measure of the instantaneous thermodynamic displacement of latent, hidden
degrees of freedom from their stationary state. Dynamical time asymmetry is
 a general phenomenon, 
independent of the underlying energy surface, and is frequently even
visible in measurements on systems that have fully reached
equilibrium. This finding has
fundamental implications for the interpretation of many experiments
on, and simulations of, biological and materials systems.
\end{abstract}

\maketitle
\section*{Introduction}
Relaxation refers to the dynamics of approaching a stationary state
(e.g. thermodynamic equilibrium) and is a hallmark of non-equilibrium physics,
from condensed matter \cite{farhan,dattagupta_relaxation_2012,Kubo_1991} to single-molecule systems
\cite{Chen_2007} initially perturbed near
\cite{Onsager1,Onsager2,Kubo_1957,Kubo_1991,Ralf_close,Maes_2011,Baiesi_2013,polettini_nonconvexity_2013,Maes_2017,Maes_2019} or far
\cite{Kurchan,Kurchan2,lu_nonequilibrium_2017,klich_mpemba_2019,Saito_PRL_2019,Lapolla_PRL}
from equilibrium. 
{In extreme cases
the non-stationary behavior of a system extends over all
experimentally accessible time-scales -- a phenomenon often
referred to as  ``aging'' \cite{Bouchaud_92,Bouchaud,Mezard,Monthus,Dean_wow}. 
Aging is typically assumed to occur in systems whose energy landscapes
contain a large number (scaling exponentially
with the system size) of meta-stable states  \cite{Bouchaud_92,Bouchaud,Mezard,Monthus,Dean_wow,KMC}. It has been observed in polymeric \cite{Hodge,polymer}, spin \cite{spinG1,endof} and
colloidal glasses \cite{Weitz1,Weitz}, supercooled liquids
\cite{PRL,Wolynes,local,deBe} and recently in protein
internal dynamics \cite{Frauen,Brujic,Saleh,Jeremy}, where it may also
affect biological function \cite{Xue,Xie1,Xie2,Xie3}.

Typical manifestations of aging are a
complex, non-exponential relaxation spectrum and non-stationary
correlation and response functions
\cite{Hodge,polymer,RMP,PRL,Wolynes,local,spinG1,endof,Weitz1,Weitz,Frauen,Brujic,Saleh,Jeremy,Ralf} that depend strongly and systematically on the time elapsed since the system was
prepared \cite{Hodge,Kurchan2,Franz,Ritort,leDoussal} or, when derived from time-series measurements, on the duration of the
observation \cite{Ralf,Jeremy}. 
The temporal extent of apparent aging dynamics in experimental systems (e.g. spin glass materials),
although very long, may be finite
\cite{endof}. Throughout we will refer to aging
  systems with experimentally observable equilibration as
  ``transiently aging'' irrespective of the precise manner in which the relaxation time
  depends on the system size.

Theoretical studies on aging have focused mainly on
  non-stationary correlations and responses
\cite{Franz,Ritort,leDoussal,Ralf,PNAS,Dean_wow,Kurchan_fin,Kurchan_1994,Franz_PRX} as well as generalizations to
aging systems of the
fluctuation-dissipation relation
\cite{Kurchan,Kurchan2,Peliti,Peter,Eli1}. 
Aging dynamics has frequently
been associated with the existence of deep
traps with unbounded depth in the potential energy function
\cite{Bouchaud,Monthus}, fractal properties of the underlying free energy landscape
\cite{Frauen,Brujic,NComm}, the presence of disorder
\cite{leDoussal,Peliti}, and other effects \cite{Franz,Ritort,Ritort2,Eli2,KMC}.

Recent efforts in understanding relaxation dynamics that are not
limited to systems with unobservable stationary states focus on
diverse aspects of the thermodynamics of relaxation, e.g. the r\^ole
of initial conditions in the context of the so-called ``Mpemba effect'' (i.e. the phenomenon where
 a system can cool down faster when initiated at a higher temperature)
 \cite{lu_nonequilibrium_2017,klich_mpemba_2019}, asymmetries in the
 kinetics of relaxation from thermodynamically equidistant temperature
 quenches \cite{Lapolla_PRL}, a spectral duality between relaxation and
 first-passage processes \cite{David,Hartich_2019}, so-called
 ``frenetic'' concepts \cite{Maes_2017,Maes_2019}, and the
statistics of the 'house-keeping' heat \cite{Speck,Roldan} and
 entropy production \cite{Noh}. Important advances in understanding transients of
 relaxation also include information-theoretic bounds on the entropy production
during relaxation far from equilibrium \cite{Saito_PRL_2019} and the
so-called ``thermodynamic uncertainty relation'' for non-stationary
initial conditions that bounds transient currents by means of the total entropy production \cite{TUR}. 

Here, we look at non-stationary physical observables from a more general, ``first
principles'' perspective. 
By directly analyzing the mathematical structure of the
underlying multi-point probability density functions we reveal 
the universality of a broken time-translation invariance that we coin as
\emph{dynamical time asymmetry} (DTA).
We prove the established linear aging correlation
  functions to be ambiguous indicators of broken time-translation
  invariance.  DTA has many of the properties
commonly associated with aging but, unlike theoretical models of aging
\cite{KMC,Bouchaud,Bouchaud_92,Monthus,Mezard,Dean_wow,Bovier},  does not require any particular functional form of the dependence on
the aging time
nor that the relaxation time increases exponentially with
system size and is therefore experimentally unobservable. 
Moreover, we here show that specific properties, such
  as deep traps in the potential energy function \cite{Bouchaud,Monthus}, fractal properties of the underlying free energy landscape
\cite{Frauen,Brujic,NComm}, or the presence of disorder
\cite{leDoussal,Peliti} that are often required for aging to occur,
are not required for DTA dynamics, although they can
amplify the breaking of time-translation invariance. In fact, DTA typically
implies (transient) aging but the converse is not true.
Instead, we prove DTA to emerge whenever (i) a
physical observable corresponds to a
lower-dimensional projection in configuration space that renders the
reduced dynamics non-Markovian, and (ii) the projected physical observable is not
prepared in, or initially sampled from and averaged over, a stationary
distribution i.e., a distribution that does not change in
time.

Most measurements on condensed matter correspond to projections of type (i), examples being structure factors in
scattering experiments \cite{NComm,Wolynes,local,Weitz1,Weitz}, spectroscopic response
functions (e.g. magnetization \cite{spinG1,endof,Peliti} and dielectric responses
\cite{PRL,PNAS,polymer}), the rheology of soft materials \cite{rheo,rheo2},
diverse empirical order parameters \cite{Franz,Ritort,Ralf} and
measurements of mechanical responses \cite{Hodge}. These projections
also inevitably arise in
single-particle tracking \cite{local,NComm,Ralf} and measurements of
various reaction coordinates in all single-molecule experiments
(e.g. internal distances) and
simulations (e.g. projections onto dominant principal modes in
Principal Component Analysis)
\cite{Brujic,Xie1,Xie2,Xie3,Jeremy,Woodside1,Woodside2,Saleh,Woddside_rough,dynamical_disorder}.

In these
measurements (i) applies as soon as the latent degrees of freedom
(DOF) (those being
effectively integrated out) evolve on a time-scale similar to the
monitored observable. In contrast, (i) does not apply when the latent
DOF relax much faster than the observable, for example when
neglecting inertia and integrating out solvent degrees of freedom of a colloidal particle in a low Reynolds number
environment. Condition (ii) applies whenever the observable evolves
from a non-stationary initial condition. This includes all experiments
involving an instantaneous perturbation of the observable in
equilibrium (e.g. magnetization or dielectric, rheological and mechanical response),
and all experiments involving evolution from a quench, such as in temperature,
pressure, or volume (which inter alia includes scattering experiments on
supercooled liquids). Condition (ii) also holds
in situations where the observable is
neither perturbed nor quenched but is initially under-sampled
from equilibrium, that is, when it is sampled from equilibrium with a
limited number of repetitions (say $1-10^3$) such as 
in single-molecule FRET, AFM or
optical tweezers experiment, as well as particle-based computer
simulations. This yields a distribution that does not converge to the invariant measure.
  In fact, as regards DTA we prove quenching and
the under-sampling of equilibrium to be
qualitatively equivalent.             
Whenever both conditions (i) and (ii) are
fulfilled, DTA emerges irrespective of the details of the dynamics.

In the main text and in the examples we focus on systems whose
dynamics obey detailed balance and, as a whole, are initially prepared at
equilibrium. The monitored lower-dimensional  observable is assumed to
evolve from some non-equilibrium initial distribution (i.e. not the
marginalized equilibrium distribution \footnote{From a
thermodynamic point of view such
systems are characterized by a transiently positive entropy production
that vanishes upon reaching equilibrium}). Generalizations to
a non-equilibrium preparation of the full system (e.g. by a
temperature quench) are discussed in detail the Appendix.

\section*{Theory}
\addtocounter{section}{1}
We consider a mechanical system at least
weakly coupled to a thermal reservoir, such that the full system's
dynamics (i.e. all
degrees of freedom; Fig.~\ref{Fg1}a, red trajectory) obeys a
time-homogeneous Markovian
stochastic equation of motion \cite{Freidlin} (for details see
Appendix), which generates ergodic dynamics in phase space.
That is, starting from any initial condition
the system is assumed to evolve to a
unique stationary distribution
in a finite, but potentially extremely long,
time that may or may not be reached during an observation. This
assumption is true for a vast majority of soft matter and biological systems
of interest and also includes glassy materials.
To impose only the mildest of
assumptions we consider that the full system is prepared in an equilibrium state at $t=0$, i.e. the full
system was created at a time $t=-\infty$ and 
the initiation of an experiment or phenomenon imposes
a time origin at $t=0$, whereas the actual observation starts after some
time $t_a\ge 0$ (see
Fig.~\ref{Fg1}b), where $t_a$ is the so-called aging (or waiting) time and the measurement time-window is the time delay
$\tau=t-t_a$.  The more restrictive assumption of a non-stationary preparation (e.g. a
temperature quench \cite{Kurchan,Lapolla_PRL}) is treated in the
Appendix \ref{2B}. In practice, a stationary preparation means that
at $t=0$ the full system's configuration is distributed according to
a stationary, invariant probability measure. This refers either to the
initial statistical ensemble of configurations in a bulk system or to
the repeated sampling of individual initial configurations (say in a
single molecule experiment), which are
drawn randomly from the invariant probability measure.  
We assume that only the projected observable is being monitored at all
times $t\ge 0$. The assumptions stated above suffice to prove our
claims (for details see Appendix). 
\begin{SCfigure*}[\sidecaptionrelwidth][t]
\centering
\includegraphics[width=11.cm]{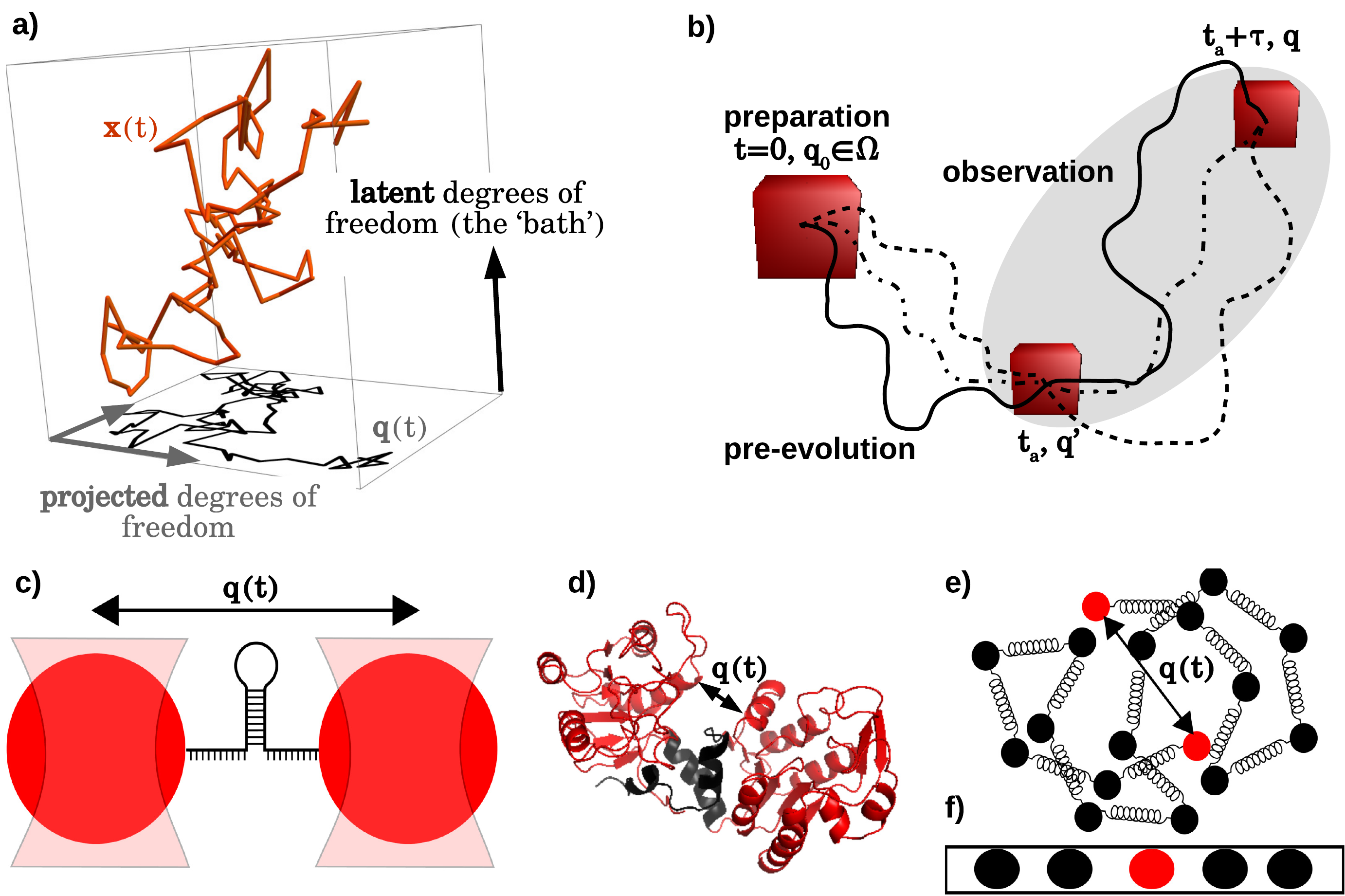}
\caption{\textbf{Schematics of projected observables, multi-point propagation and
    model systems.} a) A physical observable corresponding to a
  simple lower-dimensional projection (shadow trajectory) of the full
  system's trajectory (red line), defining projected and latent (hidden)
  degrees of freedom. b) Trajectories of length $t$ evolving from
  preparation, through an aging (or pre-evolution) period of length
  $t_a$, followed by the observation of duration $\tau=t-t_a$. c)
  Optical tweezers experimental set-up  probing DNA-hairpin
  dynamics; d) Structure of the yeast PGK protein with the reduced
  coordinate represented by the arrow. 
  e) Rouse model of a polymer chain, comprising Hookean
  springs with a zero rest-length immersed in a heat bath. The reduced
  coordinate corresponds to the end-to-end distance. f)
 Single file model with the tracer particle depicted
 in red.}
\label{Fg1}
\end{SCfigure*}

For simplicity we use $\langle \cdot\rangle$ interchangeably to
denote the average over an ensemble of trajectories at a given time
and over time along a given trajectory,
respectively, keeping in mind that they are identical only when the
trajectory is much longer than the longest relaxation time $t_{\mathrm{rel}}$. The state of the
observable is denoted by $q(t)\in\Xi$ (Fig.~\ref{Fg1}a, black
trajectory), which we assume, 
without loss of generality, to be one dimensional (for the
general case see the Appendix). Theoretically, each repetition of the experiment/process 
leads to an initial condition for $q(t)$ drawn randomly from the
reduced stationary probability
density $p_{\mathrm{inv}}(q_0)$. In practice, however, this is not
necessarily the case.
For example, supercooled liquids \cite{PRL,Wolynes,local} as well as polymeric \cite{Hodge,polymer}, spin \cite{spinG1,endof}, and
colloidal \cite{Weitz1,Weitz} glasses are
prepared by a quench in an external parameter (typically temperature)
\cite{Hodge,polymer,spinG1,endof,Weitz1,Weitz}, such that the
observable $q(t)$ nominally attains a non-stationary initial condition.    
A process may also start with the observable internally constrained to
a subdomain of $p_{\mathrm{inv}}(q_0)$, e.g. a chaperone stabilizing a
particular configuration of a folded protein, with the biological process starting
upon unbinding of the chaperone \cite{chaperone}. In another example
single-molecule enzyme experiments may monitor the statistics of substrate
turnover, where $q(t)$ reflects the geometry of the catalytic
site of an enzyme that is reactive only for a specific sub-ensemble of
configurations \cite{Xie1,Xie2,Xie3}. Binding of a substrate
molecule enforces an initial constraint on $q(t)$ thereby imposing
non-stationary initial conditions on the chemical reaction. Alternatively, we may simply choose to initialize the
experiment (i.e. reset our clock) a posteriori, such that $q(0)$ has a
preset value, or we are dealing with a single, or a
limited number of time-series  \cite{Jeremy} which do not
sample $p_{\mathrm{inv}}(q_0)$ sufficiently. In all these cases the observable
is effectively not prepared in a stationary state, i.e. $p_{0}(q_0)\ne p_{\mathrm{inv}}(q_0)$. 

The dynamics in aging systems is conventionally analyzed via
the normalized two-time correlation function \cite{spinG1,Franz,Ritort,leDoussal,Ralf,Kurchan2}
\begin{equation}
C_{t_a}(\tau)=\frac{\langle q(\tau+t_a)q(t_a)\rangle-\langle q(\tau+t_a)\rangle\langle q(t_a)\rangle}{\langle q(t_a)^2\rangle-\langle q(t_a)\rangle^2}.
  \label{corr_m}
\end{equation}  
A system is often said to be aging if $C_{t_a}(\tau)$ strongly depends
on $t_a$ in the sense that the relaxation of a system
takes place on time-scales that grow with the age of the system $t_a$, and continue to do so beyond the
largest times accessible within an experiment or simulation
\cite{Bouchaud,Monthus,Dean_wow,Bovier,KMC}. 

However,
the analysis and interpretation of time-series of physical observables
that show DTA require a fundamentally
different approach irrespective of whether equilibrium is attainable in an
experiment or not. We prove below that $C_{t_a}(\tau)$ cannot conclusively indicate
whether time-translation invariance is broken
(see Appendix \ref{Oh}, Lemma \ref{nonconc}); in particular it
  cannot disentangle broken  time-translation invariance from
  ``trivial'' correlations with a non-stationary initial
  condition (i.e. from ``weak'' or ``second order'' non-stationarity \cite{Pavliotis}).
  This is particularly problematic if one uses \eqqref{corr_m}
  as a ``definition of DTA'' 
  to infer whether
a complex experimental system, such as an individual protein molecule
\cite{Jeremy,Saleh}, evolves with broken time-translation invariance or not. \eqqref{corr_m} is nevertheless reasonable,
albeit sub-optimal,
for quantifying DTA in 
materials that are known to posses a broken time-translation invariance.

Our aim is to conclusively and unambiguously infer whether 
relaxation evolves with a broken
  time-translation invariance that is encoded in
  $G(q,t_a+\tau|q',t_a,q_0\!\in\!\Omega_0)$, the probability density for
the observable to be found in an infinitesimal volume
element centered at $q$ at time $\tau+t_a$ given that it was at $q'$
at time $t_a$ and started at $t=0$ somewhere
in a subdomain
$q_0\in\Omega_0\subset \Xi$ (Fig.~\ref{Fg1}b) with probability
$p_{0}(q_0)$. $\Omega_0$ is strictly non-empty and may be a point, an
interval or a union of intervals.

The dynamics of an observable $q(t)$ is generally said to be
time-translation invariant (mathematically referred to as ``strictly stationary''
\cite{Reichl,Pavliotis} or ``well-aged'' \cite{Keizer}) if the
underlying (effective) equations of motion that govern the evolution of $q(t)$ do \emph{not}
explicitly depend on time. That is, the probability of a path 
$\{q(t)\}$ for $t\in[t_a,t_a+\tau]$ does not depend on $t_a$.
This is the case, e.g. in Newtonian dynamics or Langevin 
  dynamics driven by Gaussian white noise \cite{Pavliotis} as well as generalized Langevin 
  dynamics driven by \emph{stationary} Gaussian colored noise
  \cite{Haenggi,Haenggi_2007,Fox}.
 Here $q(t)$ is said to be
time-translation invariant
  if and only if (see also Definition \ref{Def1} in Appendix \ref{Oh})
  \begin{equation}
    G(q,t_a+\tau|q',t_a,q_0\!\in\!\Omega_0)=G(q,t'+\tau|q',t',q_0\!\in\!\Omega_0),
\label{ttinv}    
  \end{equation}
  holds for any $\tau$ and $t'$ \footnote{Generally speaking ``strict
  stationarity'' and Definition \ref{Def1} in Appendix \ref{Oh} are
  equivalent only if $q(t)$ is a Gaussian process. If this is not the
  case Definition \ref{Def1} imposes a milder condition on time-translation symmetry.}.
Conversely, if time-translation invariance is broken we say that the
  system is 
  dynamically time asymmetric. That is, time-translation invariance is broken
 if and only if the two-point conditioned Green’s function
 $G(q,t_a+\tau|q',t_a,q_0\!\in\!\Omega_0)$ 
 depends on $t_a$ (see also Definition \ref{Def2} in Appendix \ref{Oh}).  The
 two-point conditioned Green’s function
is
 defined as
 \begin{equation}
G(q,t_a+\tau|q',t_a,q_0\!\in\!\Omega_0)\equiv\frac{P(q,t_a+\tau,q',t_a,q_0\!\in\!\Omega_0)}{P(q',t_a,q_0\!\in\!\Omega_0)},
   \label{2point}
 \end{equation}
 where $P(q,t_a+\tau,q',t_a,q_0\!\in\!\Omega_0)$ denotes the joint density
 of $q(t)$ to be found initially within $\Omega_0$ and to pass $q'$ at
 time $t_a$ and to end up in $q$ at time $t_a+\tau$, and
 $P(q',t_a,q_0\!\in\!\Omega_0)$ the  joint density
 of $q(t)$ to be found initially within $\Omega_0$ and to pass $q'$ at
 time $t_a$.

Note that there seems to be some relation between DTA
  and aging. A system is typically said to be aging
if $C_{t_a}(\tau)$ in Eq.~(\ref{corr_m})
   depends on $t_a$ (i.e. that $q(t)$ is weakly non-stationary) but in a specific
   manner, e.g. the so-called ``slow'', non-stationary component of
   $C_{t_a}(\tau)$ must scale for all large $t_a$ as some power of $\tau/t_a$
   \cite{Bouchaud,Monthus,Dean_wow} (for a rigorous discussion see
   \cite{Bovier}). However, this does not require that
   time-translation invariance (i.e. Eq.~(\ref{ttinv})) is
   broken \cite{Bouchaud,Monthus,Dean_wow}. So-called kinetically constrained
   models \cite{KMC} and the spherical p-spin model
   \cite{Peliti,Franz_PRX,Arous}, for example, have
   correlation functions \eqqref{corr_m} that show aging, but, when
   fully observed and \emph{not} averaged over disorder (and only
   then), satisfy
   Eq.~(\ref{ttinv}.
Clearly, if time-translation invariance is broken (see also
   Definition \ref{Def1} in the Appendix \ref{Oh}) then $C_{t_a}(\tau)$
   automatically depends on $t_a$. 
If the dynamics is furthermore such that $C_{t_a}(\tau)$
   depends on $t_a$ as some power of $\tau/t_a$ (see Eq.~\eqref{powerlaw}
   below as well as Eqs.~\eqref{full_a} and \eqref{full_c} as well as
   \cite{Ralf,Eli_aging,Eli_JCP}) and, in addition,
   equilibrium
   cannot be attained during an observation then DTA also implies aging dynamics. However, the converse
   is not true.

To connect  the aging correlation function in \eqqref{corr_m} with \eqqref{ttinv} we note
   that the numerator in
 \eqqref{corr_m} involves averages
\begin{eqnarray}
&& \langle q(t)\rangle\equiv\int_\Xi q G(q,t|q_0\!\in\!\Omega_0)dq\\ 
&& \langle q(\tau+t_a)q(t_a)\rangle\equiv\!\!\int_\Xi\int_\Xi qq'
 G(q,\tau+t_a,q',t_a|q_0\!\in\!\Omega_0)dqdq'\nonumber
\end{eqnarray}
 where the conditional 
 density of the projected observable $G(q,t|q_0\!\in\!\Omega_0)$ is
 discussed in \cite{Lapolla_f,Lapolla_PRL} and in Appendix \ref{2B} (see Eq.~\eqref{projected}).
The three-point conditional probability density
$G(q,\tau+t_a,q',t_a|q_0\!\in\!\Omega_0)$ -- the probability density for
the observable to pass through an infinitesimal volume
element centered at $q'$
at time $t_a$ and end up in $q$ at time $\tau+t_a$ having started at $t=0$ 
in a subdomain
$q_0\in\Omega_0\subset \Xi$ with probability $p_{0}(q_0)$, 
 is defined as (for details see Appendix \ref{2B}, Eq.~\eqref{Greens3})
 \begin{equation}
G(q,t_a+\tau, q',t_a|q_0\!\in\!\Xi_0)\equiv\frac{P(q,t_a+\tau,q',t_a,q_0\!\in\!\Omega_0)}{P(q_0\!\in\!\Omega_0)}.
\label{3p}
 \end{equation}
Based on the mathematical properties of
$G(q,t_a+\tau|q',t_a,q_0\!\in\!\Omega_0)$ and $G(q,t_a+\tau,
q',t_a|q_0\!\in\!\Xi_0)$ we prove in the Appendix \ref{Oh} (see Theorem \ref{invariant}, Corollary \ref{binvariant} and, Lemma \ref{nonconc})
that $C_{t_a}(\tau)$ in \eqqref{corr_m} can show a $t_a$-dependence
even if \eqqref{ttinv} is satisfied, i.e. when the system is
time-translation invariant. 
That is, if the system is dynamically time asymmetric then
$C_{t_a}(\tau)$ depends on $t_a$, whereas the converse is not
necessarily true. In turn this implies that one cannot determine on
the basis of $C_{t_a}(\tau)$ derived from a time-series $q(t)$ whether 
time-translation invariance is broken,
and a definitive and unambiguous indicator
must be sought for.

We demonstrate this using the
cleanest and most elementary example of a time-translation invariant
system -- a Brownian
particle confined to a box of unit length (i.e. $L=1$) evolving from a a point $\Omega_0=x_0$
and from a uniform
distribution within an interval $\Omega_0=[a,b]$ for some $0<a<b<1$. For this
example the denominator in \eqqref{2point} is defined as
$P(q',t,q_0\!\in\!\Omega_0)\equiv\int_a^bQ(q',t|q_0)dq_0$ and the
numerator as $P(q,t_a+\tau,q',t_a,q_0\!\in\!\Omega_0)\equiv
Q(q,\tau+t_a|q')P(q',t_a,q_0\!\in\!\Omega_0)$, where $Q(x,t|x_0)$ denotes
the propagator of the confined Brownian particle. Plugging into
\eqqref{2point} confirms the validity of \eqqref{ttinv} and hence
time-translation invariance. Nevertheless, the very same system
exhibits a $t_a$-dependence of the aging autocorrelation function
defined in
\eqqref{corr_m} over more than two orders of magnitude in time
measured in units of the relaxation time $t_{\rm rel}=L^2/D\pi^2$ as
depicted explicitly in Fig.~\ref{boxed}. Note that by allowing
  the box to become macroscopic in size (i.e. $L\to\infty$) the
  relaxation time and thereby the extent
of the $t_a$-dependence can become arbitrarily large when expressed in
absolute units.
\begin{figure}
\centering
\includegraphics[width=1.\linewidth]{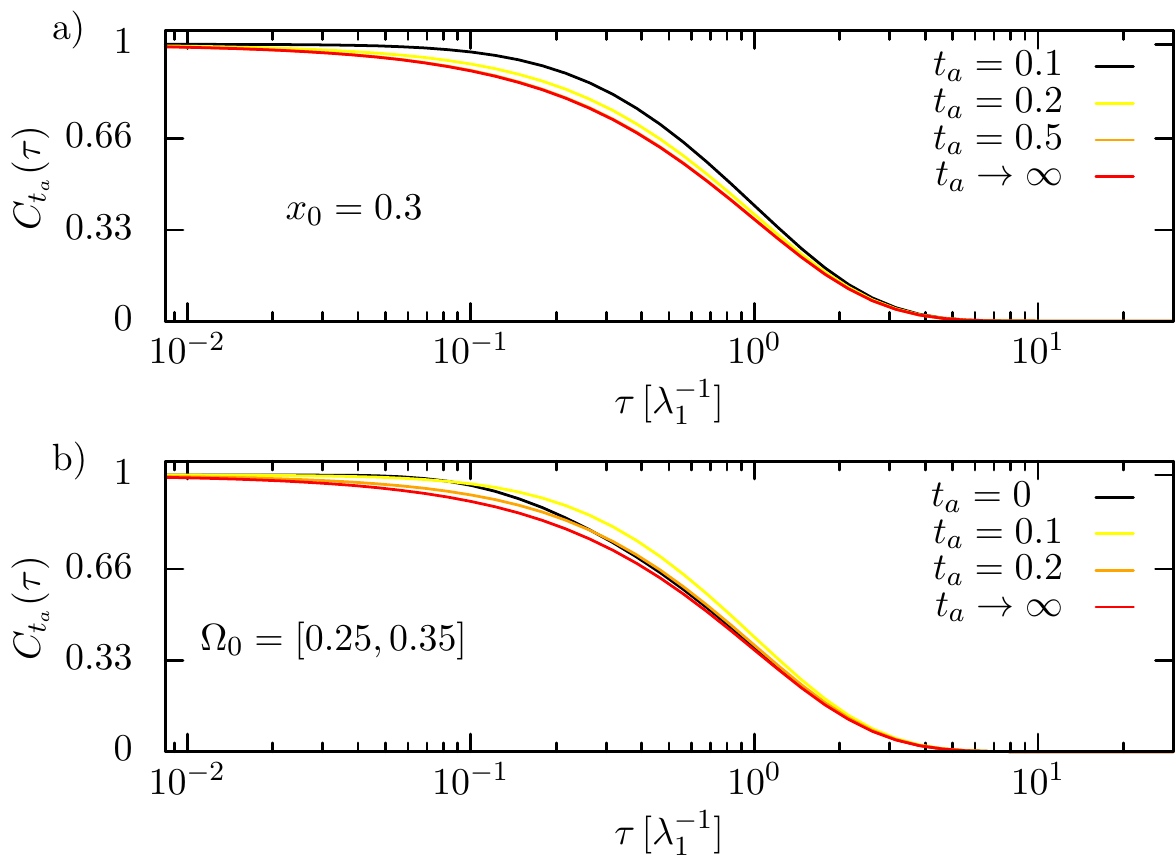}
\caption{\textbf{Aging correlation functions display fictitious
    dynamical time asymmetry  in
    time-translation invariant systems.} Analytical results for the aging correlation
  function $C_{t_a}(\tau)$ defined in \eqqref{corr_m}
  for a Brownian particle confined to a unit box evolving from a) the
  point $\Omega_0=0.3$ and b) the interval
  $\Omega_0=[0.25,0.35]$ for several values of the aging time
  $t_a$. Time $\tau$ is expressed in units of the relaxation time
  $\lambda_1^{-1}$.}
\label{boxed}
\end{figure}

A general mathematical analysis (see Appendix \ref{2B}) therefore necessarily ties
DTA to the
three-point (non-Markovian) conditional probability density,
$G(q,\tau+t_a,q',t_a|q_0\in\Omega_0)$. 
If the projected dynamics is Markovian it is in turn fully described by two-point conditional densities 
$G_{\mathrm{Markov}}(q,\tau+t_a,q',t_a|q_0\in\Omega_0)=G(q,\tau|q',0)G(q',t_a|q_0\in\Omega_0)$. If, on the
other hand, the reduced dynamics is non-Markovian but the initial
condition is sampled from the full (invariant) stationary density
$p_0(q_0)\to p_{\mathrm{inv}}(q_0)$ (or equivalently, $\Omega_0=\Xi$), we have $G(q,\tau+t_a,q',t_a|q_0\in\Xi)=G(q,\tau|q',0)p_{\mathrm{inv}}(q')$.   
In both cases there is no DTA (see Appendix~\ref{Oh}). To
quantify broken time-translation invariance on the level of reduced phase
space probability densities we therefore define the \emph{time asymmetry index} as
\begin{align}
\AI_{\Omega_0}(t_a,\tau)\equiv \int_{\Omega}dq&\int_\Omega dq'
\bigg[ G(q,\tau+t_a,q',t_a|q_0\!\in\!\Omega_0)\times\notag\\
&  \ln\frac{G(q,\tau+t_a,q',t_a|q_0\!\in\!\Omega_0)}{G(q,\tau|q')G(q',t_a|q_0\!\in\!\Omega_0)}\bigg]
\label{AI}
\end{align}
where for notational convenience we
henceforth drop the explicit dependence on $\Omega_0$,
i.e. $\AI_{\Omega_0}(t_a,\tau)\equiv\AI(t_a,\tau)$.
The time asymmetry index measures the relative entropy between
  the actual evolution of the observable 
and a corresponding ``fictitious'' dynamics that has the same probability
 density of the intermediate point $q$ at time $t_a$ but where at time
 $t_a$ the latent degrees of freedom
 are instantaneously quenched to equilibrium. Broken
   time translation invariance reflects that the effective equations
   of motion that govern the evolution of $q(t)$ change in time as a
   result of the relaxation of the hidden DOF the observable is
   coupled to. That is, if one were e.g. to derive an effective generalized Langevin
   equation for $q(t)$ the latter would contain a memory kernel and
   noise that depend explicitly on the time elapsed since the
   preparation of the system (see e.g. \cite{Robertson}).

 A broken time-translation invariance is evidently a clear signature
of non-equilibrium dynamics and therefore intimately related to
entropy production.  $\AI$ may thus also be given
 a thermodynamic interpretation as an \emph{entropy associated with
 the breaking of time-translation invariance} in analogy to the ``instantaneous
 excess free energy'' -- the relative entropy between
 $G(q,t|q_0\in\Omega_0)$ and $p_{\mathrm{inv}}(q)$
 \cite{Lebowitz,Mackey,Qian,Lapolla_PRL}. 
Therefore it appears that the entropy of breaking
time-translation invariance measures the instantaneous thermodynamic displacement of latent
degrees of freedom at time $t_a$ from their stationary state. Note
that $\AI(t_a,\tau)>0$ also implies a violation of the
fluctuation-dissipation theorem for non-Markovian system because it
implies that the ``bath'' is non-stationary \cite{GFDT}. In 
   general $\AI$ is experimentally measurable simply by monitoring the
   time-series of the observable $q(t)$ (for details see Appendix~\ref{experimental}).

The relative entropy is a pseudo-metric and therefore the absolute
value of the time asymmetry index (other than $\AI(t_a,\tau)=0$ implying
time-translation invariance and $\AI(t_a,\tau)>0$ its violation) does not necessarily immediately
allow for a quantitative comparison of DTA in different systems with disparate
dimensionality. It is always meaningful when one considers a
comparison of the  same system and observable under different conditions
(e.g. initial conditions, values of control parameters etc.). If one
aims at comparing quantitatively DTA in different systems and/or observables one should instead
consider a symmetrized version of the relative entropy (see e.g. \cite{Nielsen}).

The time asymmetry index is constructed to \emph{detect and quantify conclusively} broken
time-translation invariance according to \eqqref{ttinv}. It effectively
measures
the instantaneous relaxation of the latent degrees of freedom and is
unaffected by spurious
non-stationarity due to correlations between the value of the
observable at time $t_a+\tau$ and the particular ``initial'' value at
time $t_a$. These
correlations are spurious because they exist for any $t_a$ and relax
as a function of $\tau$ irrespective of whether a system is
time-translation invariant or not.

By construction
$\AI(t_a,\tau)\ge 0$ and is identically zero for any $t_a$ and
$\tau$ if and only if $q(t)$ is time-translation invariant. In
turn, the observable $q(t)$ is time-translation invariant if and only if it
is Markovian and/or $q(t=0)$ is sampled
from a distribution converging in law to the invariant measure (the proof is presented in
the Appendix~\ref{Oh}, Theorem \ref{ttinv} and Corollary \ref{binvariant}).
As a result $\AI(t_a,\tau)$ is identically zero for all $\tau$
and $t_a$ for the
time-translation invariant dynamics of a confined Brownian
particle evolving from a non-equilibrium initial condition (see,
however, the fictitious DTA due to weak non-stationarity that is implied by the aging autocorrelation
function in  Fig.~\ref{boxed}).
Moreover, the extent of DTA is limited by the relaxation time
$t_{\mathrm{rel}}$ such that $\AI(t_a,\tau)\to 0$ whenever
$t_a\gg t_{\mathrm{rel}}$ or $\tau\gg t_{\mathrm{rel}}$. Obviously, if the full
system is initially quenched into any non-stationary initial condition
(see e.g. \cite{Lapolla_PRL}),
then $\AI(t_a,\tau)> 0$ as long as the projection renders the
reduced dynamics non-Markovian. 
Therefore, as soon as $\AI(t_a,\tau)\ne 0$ for some values $t_a$ and
$\tau$ smaller than $t_{\mathrm{rel}}$, the dynamics is time asymmetric,
in specific cases with a self-similar scaling (see Appendix~\ref{Oh}, Propositions \ref{similar}
\& \ref{full}). In addition
the following generic structure emerges:
\begin{equation}
C_{t_a}(\tau)=(1-\varphi)g_1(\tau)+\varphi g_2(\tau,t_a),
\end{equation}
with
$0<\varphi<1$ and $g_{1,2}$ depending on the
details of the dynamics (see Appendix~\ref{Oh}, Theorem \ref{repr}) in agreement with
the properties of aging systems \cite{Franz,Ritort,leDoussal,Ralf,PNAS,Kurchan2,spinG1,Peter}. These
results are universal -- they are independent of
details of the dynamics, and, in particular, the underlying energy
landscape.

Microscopically reversible dynamics in general allows
  for a spectral expansion of propagators and thus correlation and response
  functions (see e.g. Appendix~\ref{lower}). Moreover, in specific cases the
  projection renders the observed dynamics self-similar with parameter
  $\alpha$, that is, a change of
  time-scale merely effects an $\alpha$-dependent renormalizion of the spectrum (for
  details see Definition \ref{selfsim} in the Appendix~\ref{2B}). This arises,
    for example, when the observable corresponds to an internal
    distance within a single polymer molecule \cite{Rouse_self_sim}
    (studied here in
    Figs.~\ref{Fg2}a and ~\ref{Fg3}) or within individual protein
    molecules \cite{Granek,proteins_self_sim}, as well as in diffusion on
    fractal objects \cite{fractals}.
The aging correlation function in
  \eqqref{corr_m} then displays a power-law scaling for $\alpha>0$
  (as in Fig.~\ref{Fg2}d and Eq.~\eqref{powr} in the Appendix~\ref{Oh}) or, when $\alpha=0$
  a logarithmic behavior (as observed in \cite{Saleh};
  see also Eq.~\eqref{logr} 
  in the Appendix~\ref{Oh}). The latter is
    mathematically equivalent to the logarithmic relaxation found in
  \cite{Amir}.
  For more details see
  Propositions \ref{similar} and \ref{full} in the Appendix~\ref{Oh},
  respectively. In particular for $\tau/t_a\gg 1$,  in the
  glassy literature referred to as the ``full aging''
  \cite{PNAS,Bouchaud_92,Rodriguez_2003,Amir} regime, we find (see Appendix~\ref{Oh}, 
  Eqs.~\eqref{full_a} and \eqref{full_c})
  \begin{equation}
    C_{t_a}(\tau)\simeq A + 
    \begin{cases}
      B_\alpha\left(\frac{t_a}{\tau}\right)^\alpha &, \alpha>0,\\
      B_\alpha\left(\frac{t_a}{\tau}\right) &, \alpha=0.
\end{cases}
    \label{powerlaw}
\end{equation}
with constants $A$ and $B_\alpha$ that depend on the details of the
dynamics. 
On a transient time-scale the asymptotic results in \eqqref{powerlaw} agree with
predictions of minimalistic ``trap'' models
\cite{Bouchaud,Monthus,Dean_wow} as well as fractional dynamics and random walks with diverging waiting times
\cite{Ralf,Johannes,Eli2} (for more details see also
  Remark~\ref{exmpl} in the Appendix~\ref{Oh}).
Fractional dynamics and random walks with long waiting times (that
as well display DTA \cite{Eli_aging,Eli_JCP,Johannes})
were in fact explicitly shown to arise as
transients in projected dynamics when the latent degrees of freedom
are orthogonal to $q(t)$ \cite{Lapolla_f} and in the spatial
coarse-graining of continuous dynamics on networks
\cite{David_networks}. The phenomenology of systems displaying an
algebraic scaling of $C_{t_a}(\tau)$  as in
\eqqref{powerlaw} is therefore by no means unique, and represents
only a specific class of dynamical systems with a broken
time-translation invariance. Dynamical time asymmetry is much more general.

\section*{Examples}
It is not difficult to verify the above
claims in practice as all corresponding quantities can readily be obtained from experimental or simulation-derived
time-series. To that end we analyze DTA in four very different systems (see
Fig.~\ref{Fg1}c-e): DNA hairpin dynamics measured by dual optical tweezers
experiments, where $q(t)$ reflects the end-to-end distance
(Fig.~\ref{Fg1}c and Appendix~\ref{DNA}) \cite{Woodside1,Woodside2},
extensive MD simulations of internal motions of
 yeast PGK, where $q(t)$ corresponds to the inter-domain
 distance  (Fig.~\ref{Fg1}d and Appendix~\ref{PGK})
 \cite{Jeremy}, as well as two theoretical examples:
 the end-to-end distance fluctuations of a Rouse polymer chain
 \cite{Fixman} (Fig.~\ref{Fg1}e  and Appendix~\ref{Rouse}) and
 tracer particle dynamics in a single file of impenetrable
 diffusing particles, where $q(t)$
 reflects the position of the tracer particle
 \cite{Lapolla_f,Lapolla_CPC,Lapolla_PRL} (Fig.~\ref{Fg1}f  and
 Appendix~\ref{sfile}). The underlying energy
 landscapes of these four systems are fundamentally very different; the DNA-hairpin
 exhibits two well-defined metastable conformational states/ensembles
 \cite{Woodside1,Woodside2}, the yeast PGK has a very rugged and
 apparently fractal energy landscape \cite{Jeremy}, that of the Rouse polymer is perfectly smooth and exactly parabolic,
 and that of the single file is flat with
 the tracer motion confined to a hyper-cone as a result of the
 non-crossing condition between particles. Yet, despite these striking
 differences, all systems display
 the same qualitative time asymmetric behavior, consistent with the
 proven universality of DTA. 

 \begin{figure*}[ht!]
\centering
\includegraphics[width=12.6cm]{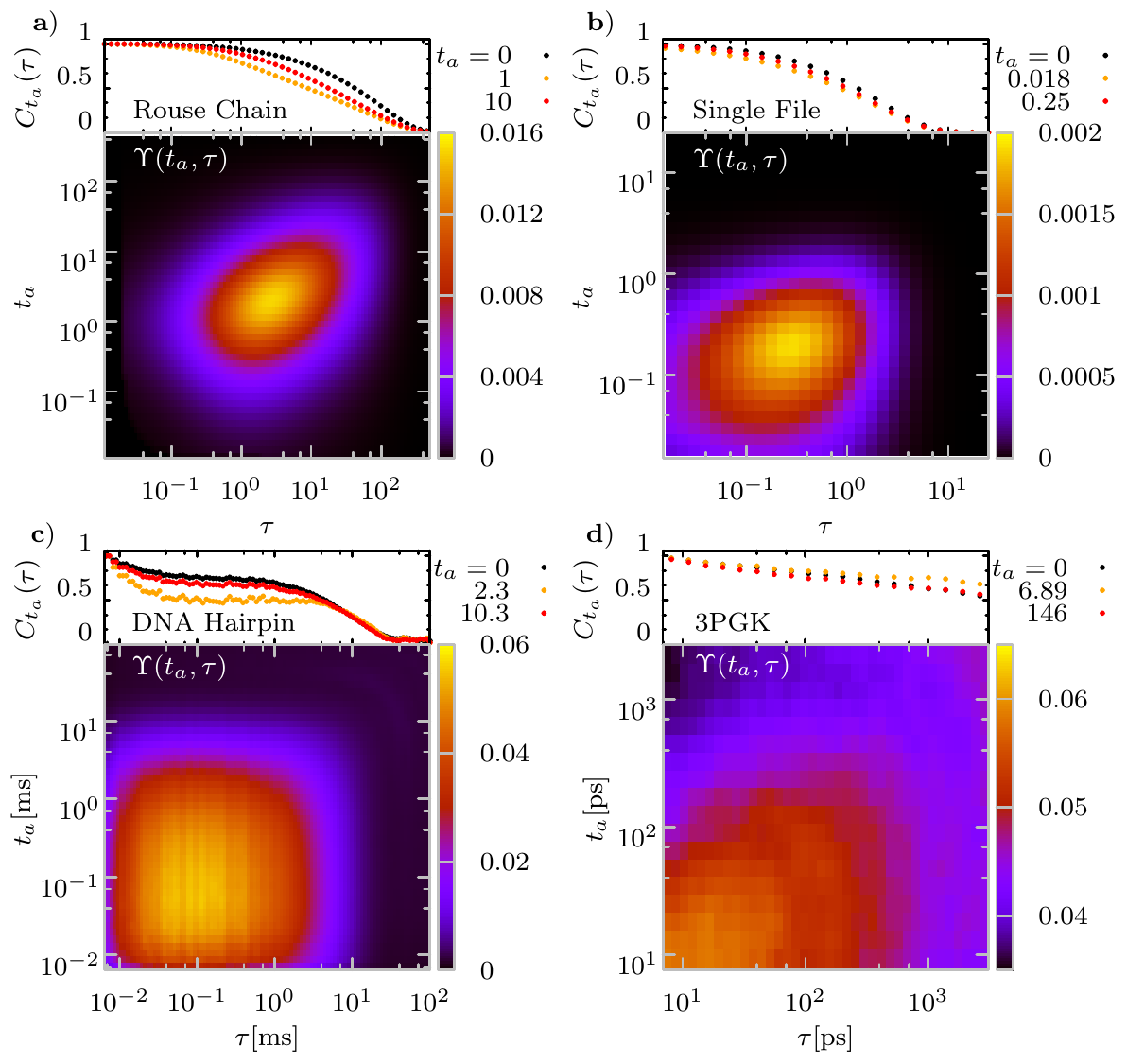}
\caption{\textbf{Aging two-point correlation function and time asymmetry
    index.} $C_{t_a}(\tau)$ for different
  values of aging time $t_a$ and corresponding $\AI(t_a,\tau)$
  for: a) the Rouse polymer chain with 50 beads with
  initial end-to-end distance in dimensionless units equal to
  $q_0=9.85$, which corresponds to the most likely end-to-end distance (the
  dimensionless relaxation time here corresponds to $t_{\mathrm{rel}}\simeq 253.38$); b) tracer-particle
  dynamics in a single file with $N=5$ confined to a
  box of unit length, tagging the
  central particle particle with initial condition $q_0=0.5$ (the
  relaxation time measured in natural units of the ``collision time'' of is  $t_{\mathrm{rel}}\simeq 2.5$), c)
  the DNA-hairpin extension determined from a trajectory of length of
  $2.75\cdot10^4$ ms sampled at $400
  \mathrm{kHz}$. The initial condition was taken at the absolute maximum of equilibrium probability density
  $q_0=2.0\pm 1$ nm, and $q$ refers to deviations from the mean
  distance $\langle d\rangle$, i.e. $q(t)=d(t)-\langle d\rangle$ (the relaxation time is  $t_{\mathrm{rel}}\approx
  15$ ms); The statistical error in
  determining $\AI(t_a,\tau)$ from the hairpin data is less than 1\%\ (see Fig.~\ref{err_an}
  in the Appendix~\ref{DNA});  d) inter-domain motion between the centers of mass of
  the N-terminal (residues 1-185) and C-terminal domains (residues 200-389) in yeast
  PGK determined from a $200$ ns atomistic MD simulation sampled
  every $150$ ps. The initial condition was $q_0=0.01\pm0.2$ nm
  relative to the average inter-domain distance $\langle d\rangle$, i.e. $q(t)=d(t)-\langle d\rangle$. c) was
  obtained from experimental data of Refs. \cite{Woodside1,Woodside2} and d) was determined from
  molecular dynamics simulations in Ref. \cite{Jeremy}. Further
  details can be found in Appendix~\ref{systems}. ``Transient aging'' in $C_{t_a}(\tau)$ arises
  whenever there is a region $(t_a,\tau)$ where
  $\AI (t_a,\tau)>0$. In the case of PGK (panel d)
  $t_{\mathrm{rel}}$ is not reached within the simulation time, which
  renders the system virtually eternally time asymmetric and ``forever aging''
  \cite{Jeremy,Ralf_forever}.}
\label{Fg2}
\end{figure*}
\begin{figure*}
\centering
\includegraphics[width=14.cm]{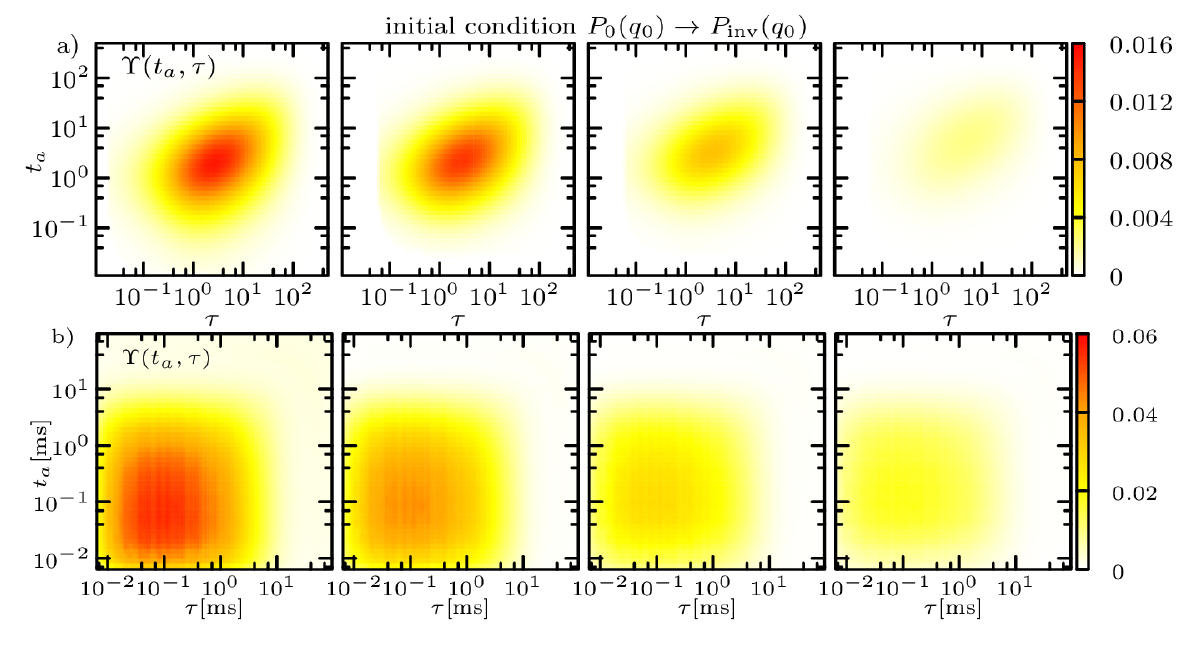}
\caption{\textbf{Attenuation and disappearance of dynamical time asymmetry upon approaching stationary initial conditions.} Gradual vanishing of
  the time asymmetry index $\AI (t_a,\tau)$ when the initial
  distribution of the projected coordinate $p_0(q_0)$ is sampled from a
  distribution being closer and closer to the density of the invariant measure,
  $p_\mathrm{inv}(q_0)$ for: a) Rouse
  model of a polymer chain (the parameters are the same as in
  Fig.~\ref{Fg2}); The initial end-to-end distance is sampled from
  intervals (from left to right): $q_0=$9.85,  $q_0\in [9-11]$,
  $q_0\in [7-13]$ and $q_0\in [4-16]$, respectively. b) experimental
  data for the DNA-hairpin. The initial conditions (relative to the
  mean value $\langle d\rangle$, that is, $d(t)=q(t)+\langle d\rangle$) are sampled (from
  left to right) from the following intervals: $q_0\in [1,3]$ nm,
  $q_0\in [-3,7]$ nm, $q_0\in [-6,10]$ and $q_0\in [-8,12]$ nm,
  respectively. When the initial condition is sampled from a
  distribution closer to the invariant measure, DTA vanishes
  confirming the claims of our theory.}
\label{Fg3}
\end{figure*}
 
 The aging correlation functions $C_{t_a}(\tau)$ and time asymmetry indices
 $\AI(t_a,\tau)$ are shown in Fig.~\ref{Fg2}. With the
 exception of the PGK protein, which does not equilibrate within the
 duration of the trajectory, in agreement with previous findings
 \cite{Jeremy}, DTA is manifested as a transient
 phenomenon. 
 The
 precise form of $C_{t_a}(\tau)$ depends on the details of the
 dynamics, which naturally vary between the systems. Moreover, the dependence of $C_{t_a}(\tau)$ on $t_a$ is
non-monotonic. 
 The
 generic form of $\AI (t_a,\tau)$ displays an initial increase towards
 a plateau, followed by a long-time decay to zero, which can be understood as
 follows. Irrespective of the details a finite time is required in
 order to allow for a build-up of memory, that is,
 of correlations between the instantaneous state of the projected observable
 and the initial condition of the latent variables. The memory at some point reaches a
 maximum. Afterwards, the memory of the preparation of the system is
 progressively lost as a result of the mixing of trajectories in full phase
 space during relaxation.
 Due to a relatively higher sampling frequency and sufficiently
 long sampling times that extend beyond the
   relaxation time all these effects are resolved in the
 experimental DNA-hairpin data but not in the case of the PGK
 simulation. 

Moreover, a hallmark of aging is that at least part of the relaxation of a system
takes place on time-scales that grow with the age of the system $t_a$, and continue to do so up to the
largest times accessible within an experiment or simulation. Interestingly,
Figs.~\ref{Fg2}  and \ref{Fg3} show that the relaxation time increases
(at least transiently) with the aging time, i.e. $\AI (t_a,\tau)$ decays with $t$ more slowly as
$t_a$ grows at least up to a threshold time. 
If an experiment or simulation does not reach this
threshold time the breaking of time-translation invariance would
seemingly take place on timescales that grow indefinitely, somewhat
similar to the aging phenomenon. 
Note that the threshold time may
become arbitrarily large in large systems (e.g. the relaxation time
and thus the threshold time in natural units for the Rouse polymer and
single file grow with the number of particles as $\propto N^2$ (see
e.g. Fig.~\ref{size} in the Appendix~\ref{Rouse}); for
any duration of an observation one may find a $N$ that makes DTA appear as everlasting). 

One appreciates that
   $\AI (t_a,\tau)$ truly quantifies the degree of broken
 time-translation invariance and \emph{not} correlations with the
 value of the observable at $t_a$. This is also the
   reason why $\AI (t_a,\tau)$ decays to zero on a time-scale
   shorter
   than $C_{t_a}(\tau)$. $C_{t_a}(\tau)$ starts at 1 and decays to zero as a result of
 ``forgetting the initial condition''. Because the probability density
 of being found at a given point always depends trivially on $t_a\ne0$ (see
 \eqqref{3p}) irrespective of whether time-translation invariance in \eqqref{ttinv} is broken or
 satisfied, $C_{t_a}(\tau)$ displays non-stationarity manifested in a $t_a$-dependence  even for time-translation
 symmetric dynamics. Conversely, $\AI (t_a,\tau)$ is
 constructed to not be affected by such spurious
 non-stationarity. Instead, it reflects how far the latent degrees
 of freedom are displaced from equilibrium at time $t_a$. In other
 words,  $\AI (t_a,\tau)$ compares
 the probability densities of the actual dynamics with
 those of fictitious dynamics that have the same probability
 density at time $t_a$ but in which at time
 $t_a$ the latent degrees of freedom are quenched to equilibrium (see
 \eqqref{AI}).

 One can look at
 $\AI (t_a,\tau)$ in two ways; as a function of $\tau$ at fixed
 $t_a$ and as a function of $t_a$ at fixed $\tau$. While the former
 intuitively reflects how the relaxation of the observable to equilibrium depends
 on the instantaneous (``initial'') state of the latent
 degrees of freedom at time $t_a$, the latter measures how the
 correlation of the value of the observable at two times separated by
 $\tau$ changes due to the relaxation of the latent degrees of freedom
 to equilibrium. The time asymmetry index
 therefore provides access to the dynamics of hidden
 degrees of freedom coupled to the observable through an analysis of time-series derived from
 measurements on the observable.

 A verification that a breaking of time-translation invariance occurs whenever the distribution of initial conditions sampled by the experiment has not converged to the equilibrium distribution 
follows from inspection of
 $\AI (t_a,\tau)$ evolving from an ensemble of initial
 conditions being closer and closer to an equilibrium distribution, i.e. $\Omega_0\to\Xi$ (see
 Fig.~\ref{Fg3} for the Rouse chain and DNA-hairpin). Indeed,
 $\AI (t_a,\tau)$ progressively
 vanishes when the initial condition becomes sampled from a
 distribution approaching the invariant measure,
 $p_{0}(q_0)\to p_{\mathrm{inv}}(q_0)$. In the Appendix~\ref{Oh} we
 prove that this is a general effect (Theorem \ref{invariant}), independent of any details of
 the dynamics.

\section*{Discussion}
Non-stationary behavior of physical observables is traditionally
considered as being important in systems with
 glassy, aging dynamics, such as polymer, spin or colloidal glasses,
 that attain glassy properties upon a quench in an
 external parameter
 \cite{Hodge,polymer,spinG1,endof,Weitz1,Weitz}. During, for example, a temperature quench, the
 system (e.g. a supercooled liquid or a set of spins) at some point cannot keep
 pace with rapid changes in the bath, and is pushed out of 
 equilibrium \cite{Kurchan}. After the quench at $t=0$ the
 observable is thus (at least weakly) non-stationary -- it is sampled from
 and averaged over a non-equilibrium ensemble, i.e. $p_{0}(q_0)\ne
 p_{\mathrm{inv}}(q_0)$. The absence of such an obvious quench
 rendered the origin of non-stationary, apparent aging behavior in
biological macromolecules somewhat mysterious
\cite{Frauen,Brujic,Jeremy,Xue,Xie1,Xie2,Xie3}. However, in biological
systems the observable can become quenched implicitly, e.g. by the
'locking in' of a protein's configuration by a chaperone
\cite{chaperone}, the configurational requirements for enzymatic catalysis
\cite{Xie1,Xie2,Xie3}, or simply by the under-sampling of equilibrium
such as in single-molecule experiments and particle-based computer
simulations \cite{Jeremy}, such that $p_{0}(q_0)\ne
p_{\mathrm{inv}}(q_0)$. In an experiment one can check for
non-stationarity of initial conditions, e.g. by inspecting whether 
histograms of the observable (also referred to as the
  ``occupation time fraction'' or ``empirical density'') at $t=0$ and at all later
times coincide \cite{Lapolla_PRR}.    

Here, we highlight a more general and wide-spread aspect of
out-of-equilibrium dynamics of physical observables -- dynamical
time asymmetry. The requirements for DTA to occur are much weaker than
for aging, and it is manifested in a very broad variety of
experimental situations, and in particular, one may also expect aging
physical observables probed in many experiments to display DTA. Even measurements on polymer, spin and colloidal
 glasses have built-in underlying
 projections. For example, in tensile creep experiments in polymeric glasses  the motion in a (cold)
 polymer is projected onto a local, effectively one-dimensional flow \cite{Hodge}. 
 In supercooled liquids and colloidal
 glasses the dynamics is typically projected onto
 local particle displacements, pair correlation functions and
 structure factors \cite{Wolynes,local,Weitz1,Weitz}. In bulk
 experiments with spin glasses and supercooled liquids
  one measures quantities such as the average single-spin
  auto-correlation function \cite{Bouchaud,Garrahan}
   , magnetization, conductance
 or the dielectric constant,  which correspond to projections of
 many-particle dynamics onto a scalar parameter \cite{endof,PRL,PNAS}. In biological
 macromolecules the projection may correspond to \cite{Brujic,Jeremy} or
 depend on \cite{Xie1,Xie2,Xie3} some internal distance within the
 macromolecule.
 These projections lead to non-Markovian observables evolving from
 non-stationary initial conditions which are in turn expected to show DTA. In fact we can
appreciate that the physical origin of DTA in both, 'traditional' glassy
systems \cite{Hodge,polymer,spinG1,endof,Weitz1,Weitz}  and biological
matter \cite{Frauen,Brujic,Jeremy,Xue,Xie1,Xie2,Xie3}, is
qualitatively the same and simply results from non-stationary initial
conditions of non-Markovian observables (see Observation \ref{obs2} in
the Appendix~\ref{Oh}). In
most of these aforementioned systems the dynamics is also aging \cite{Hodge,polymer,spinG1,endof,Weitz1,Weitz,Frauen,Brujic,Jeremy}.     

 It is important to realize that it is not possible to
   infer from a finite measurement whether the observed process is
   genuinely non-ergodic (i.e. a result of some true localization
   phenomenon in phase space) or whether the observation is made on an
   ergodic system but on a time-scale shorter the relaxation time
   \cite{Lapolla_f} (note that a comparison of the
   dynamics of PGK in Fig.~\ref{Fg2}d with a transient
   shorter than the relaxation time in any of the
   remaining examples in Fig.~\ref{Fg2}a-c shows no qualitative difference). A
 theoretical description of both scenarios on time-scales shorter than the
 relaxation time is in fact \emph{identical} (for details see 
 \cite{Lapolla_f} as well as \cite{Dean_wow} in the context of
 glasses).

Although sporting characteristics commonly associated
  with aging, DTA and aging are not quite the same thing. DTA does not
require the relaxation to take place on time-scales that grow
indefinitely with the age of the system $t_a$ beyond the
largest times accessible within an experiment or
simulation, nor does it impose requirements on the precise form of the
dependence on $t_a$. It is likely to be a ubiquitous phenomenon
that is frequently observed in measurements of projected observables. In turn, aging does not imply a broken time-translation
invariance according to \eqqref{ttinv}.

Note, however, that many
paradigmatic models of aging dynamics (e.g. continuous-time random
walks with diverging mean waiting times and fractional diffusion
 \cite{Eli_aging,Eli_JCP,Ralf}) display a (strongly) broken
 time-translation invariance. Furthermore, most \emph{experimental} observations of aging dynamics
monitor projected observables, e.g. magnetization, single-spin auto-correlation
functions averaged over the sample and potentially also over
disorder \cite{Hodge,polymer,spinG1,endof,Weitz1,Weitz,Frauen,Brujic,Jeremy}. The dynamics of these
observables is thus almost surely non-Markovian \cite{Lapolla_f} and
 expected to display DTA.

The observation of
$\AI(t_a,\tau)>0$ on a given scale of $t_a$ and $\tau$ implies that the dynamics
of the observable $q(t)$ fundamentally changes in the course of time
as a results of the relaxation of hidden DOF,
and \emph{does not} reflect correlations with the value of the
observable at zero time $q(0)$. That is, the effective equations
  of motion for $q(t)$ truly change in time.
In biological systems and in particular enzymes and other protein
nanomachines non-stationary effects are thought to influence function,
e.g. memory effects in catalysis \cite{Xie1,Xie2,Xie3}. This is
particularly important because some larger proteins potentially never relax
within their life-times, i.e. before they become degraded
(note that relaxation corresponds to attaining the spontaneous unfolding-refolding
equilibrium). This renders the dynamically time asymmetric regime virtually
'forever lasting' and implies that the system is aging
\cite{Jeremy}. As proteins are produced in the cell in an ensemble of
folded configurations under the surveillance of chaperones
\cite{chaperone}, our theory implies that DTA during
function \cite{Xie1,Xie2,Xie3} should arise naturally and generically due to the memory of
a protein's preparation.

We expect DTA to be particularly
pronounced in measurements on systems with entropy-dominated, temporally heterogeneous collective
conformational dynamics involving (transient) local
structure-formation where the background DOF evolve on the same time-scale
as the observable
\cite{Saleh}, and we suggest the breaking of time-translation
invariance to be closely related to the phenomenological notion of ``dynamical
disorder'' in biomolecular dynamics \cite{Xie1,Xie2,Xie3,dynamical_disorder}.   
 

Our results have some intriguing implications. First, a quench in an
external parameter and the mere under-sampling of equilibrium
distributions give rise to qualitatively equivalent manifestations  (but potentially with a largely different
magnitude and duration) of DTA as soon as the observable
follows a non-Markovian evolution (see Appendix~\ref{Oh}, Observation \ref{obs2}). This has important practical
consequences in fields such as single-molecule spectroscopy and computer simulations of soft and biological matter, which often suffer from
sampling constraints. Second, broken time-translation invariance is 'in the eye of the beholder',
insofar as its degree depends on the specific observable; 
there should exist a (potentially less) reduced coordinate, not
necessarily accessible to experiment (e.g. when we
  follow \emph{all} degrees of freedom), according to which
the same system will exhibit virtually time-translation
  invariant dynamics. However, auto-correlation functions will show a
  $t_a$-dependence for essentially any non-stationary initial condition in any system.

A broken time-translation invariance was shown to be linked to a
  form of entropy embodied in a \emph{time asymmetry index} that is a measure of the instantaneous thermodynamic displacement of latent, hidden
degrees of freedom from their stationary state. The time asymmetry index may
therefore be used to probe systematically the time-scale of dynamics of hidden, slowly relaxing
degrees of freedom relative to the time-scale of the evolution of the observable. In particular, it may be useful as a practical
tool to discriminate between situations where the hidden degrees of freedom evolve
through a sequence of local equilibria that would yield small
values of the time asymmetry index $\AI$ from those cases where their
evolution is transient and slow on the time-scale of the
observable thus implying a significant $\AI$. For
    example, $\AI$ may potentially provide additional
insight into the dominant folding mechanism of a protein from
single-molecule force-spectroscopy data
\cite{Woodside_rev}, in particular about the much debated heterogeneity of folding
trajectories and its functional relevance
\cite{defined,reply}.

The present theory ties dynamical time asymmetry in a general setting to
both the non-stationary preparation of
an observable and its non-Markovian time evolution. Thereby
it connects aspects of the better known phenomenology of aging of projected
observables
with the broken time-translation invariance observed in
  recent measurements on in soft and biological
materials on a common footing. Moreover, dynamical time asymmetry is
suggested to be a
ubiquitous phenomenon in biological and materials systems.

\begin{acknowledgments}
We thank Krishna Neupane and Michael T. Woodside for
providing unlimited access to their DNA-hairpin data and Peter Sollich
for clarifying discussion about physical aging and critical reading of
the manuscript. The financial support from the \emph{Deutsche Forschungsgemeinschaft} (DFG) through the
\emph{Emmy Noether Program "GO 2762/1-1"} (to AG) and
  from the \emph{Department of Eenergy}  through the grant DOE BER
FWP ERKP752 (to JCS) are gratefully
acknowledged.
\end{acknowledgments}

\let\oldaddcontentsline\addcontentsline
\renewcommand{\addcontentsline}[3]{}
\let\addcontentsline\oldaddcontentsline

\clearpage
\newpage
\onecolumngrid
\newcommand{\JJ}{\mathbf{J}}
\newcommand{\LL}{\hat{\mathcal{L}}}
\newcommand{\LLA}{\hat{\tilde{\mathcal{L}}}}
\newcommand{\LLB}{\hat{\mathcal{L}}^{\dagger}}
\newcommand{\bbf}{\mathbf{F}(\mathbf{x})}
\newcommand{\bbft}{\mathbf{\tilde{F}}(\mathbf{q})}
\newcommand{\bbftp}{\mathbf{\tilde{F}}(\mathbf{q'})}
\newcommand{\bbfq}{\langle\mathbf{\tilde{F}}(\mathbf{q})\rangle}
\newcommand{\bx}{\mathbf{x}}
\newcommand{\by}{\mathbf{y}}
\newcommand{\bR}{\mathbf{R}}
\newcommand{\ket}{\left.\right{\rangle}}
\newcommand{\bra}{\left{\langle}\right.}
\newcommand{\sss}{|\mathrm{ss}\rangle}
\newcommand{\sssl}{\langle\mathrm{ss}|}
\newcommand{\lflat}{\langle\text{--}|}
\newcommand{\rflat}{|\text{--}\rangle}

\newcommand{\leftBk}{\langle \psi^L_k|}
\newcommand{\leftKk}{|\psi^L_k\rangle}
\newcommand{\leftBki}{\langle \psi^L_{k_i}|}
\newcommand{\leftKki}{|\psi^L_{k_i}\rangle}
\newcommand{\leftBl}{\langle \psi^L_l|}
\newcommand{\leftKl}{|\psi^L_l\rangle}
\newcommand{\rightBk}{\langle \psi^R_k|}
\newcommand{\rightKk}{|\psi^R_k\rangle}
\newcommand{\rightBki}{\langle \psi^R_{k_i}|}
\newcommand{\rightKki}{|\psi^R_{k_i}\rangle}
\newcommand{\rightBl}{\langle \psi^R_l|}
\newcommand{\rightKl}{|\psi^R_l\rangle}
\newcommand{\ee}{\mathrm{e}}
\newcommand{\eeb}{\hat{\textbf{e}}}
\newcommand{\mcp}{\hat{\mathcal{P}}}
\newcommand{\BGamma}{\boldsymbol{\Gamma}}
\newcommand{\bq}{\mathbf{q}}
\newcommand{\bu}{\mathbf{u}}
\newcommand{\bv}{\mathbf{v}}
\newcommand{\bw}{\mathbf{w}}
\newcommand{\bql}{\mathbf{Q}}
\newcommand{\bss}{\mathbf{S}}
\newcommand{\mX}{\mathcal{X}}
\newcommand{\mY}{\mathcal{Y}}
\newcommand{\mZ}{\mathcal{Z}}
\newcommand{\bssu}{\underline{\mathbf{S}}}
\newcommand{\blam}{\boldsymbol{\Lambda}}
\newcommand{\bnabla}{\boldsymbol{\nabla}}
\newcommand{\eql}{Eq.~(}
\newcommand{\gradient}{\nabla\varphi(\bx)}
\newcommand{\rotation}{\boldsymbol{\vartheta}(\bx)}
\newcommand{\OO}{\mathbf{O}}
\newcommand{\OOt}{\mathbf{\tilde{O}}}
\newcommand{\DD}{\mathbf{D}}
\newcommand{\LA}{\boldsymbol{\Lambda}}
\newcommand{\nablaq}{\nabla_{\mathbf{q}}}
\newcommand{\nablaqp}{\nabla_{\mathbf{q}'}}
\newcommand{\bs}{\mathbf{s}}
\newcommand*{\QEDA}{\hfill\ensuremath{\blacksquare}}
\newcommand*{\peq}{P_{\mathrm{eq}}(\bx)}
\newcommand*{\peqo}{P_{\mathrm{eq}}(\bx_0)}
\newcommand*{\eeq}{\mathrm{eq}}
\newcommand*{\Ei}{\mathrm{Ei}}
\newcommand*{\MA}{\mathcal{A}}
\renewcommand{\thefigure}{A\arabic{figure}}
\renewcommand{\theequation}{A\arabic{equation}}
\setcounter{equation}{0}
\setcounter{figure}{0}
\setcounter{page}{1}
\setcounter{section}{0}


\begin{center}
  \textbf{APPENDIX}
\end{center}
\appendix

\begin{quotation}
In this Appendix we present the main theorems
needed for the article with the corresponding proofs. We treat the problem in a general setting, that is,
not assuming that the full system is initially prepared in equilibrium.
Further included are analytical results with details of
calculations for the Rouse polymer and single file diffusion, all
details of the numerical analyses of the DNA-hairpin and protein PGK data and further supporting
results.
\tableofcontents 
\end{quotation}

\section{Definitions, notation and preliminaries}
We consider a stable conservative mechanical system in a continuous
domain $\Omega\in \mathbb{R}^d$ that is at least weakly coupled to a
thermal bath with Gaussian statistics with the longest correlation time $\tau_b$ being
much shorter than that of the system, $\tau_s$ (i.e. $\tau_b\ll \tau_s$)
such that the bath can be considered as representing stationary white
noise on the time-scale of the system's dynamics \cite{Freidlin}. The thermal
bath is either external or the result of integrating out an additional subset of
internal degrees of freedom that relaxes much faster than the system.
At any time $t$ the state of the system is specified by a $d$-dimensional state
(column) vector $\bx_t\in\mathbb{R}^d$, whose entries are generalized
coordinates $x_{t,i}$. 
Note that the dynamics in soft matter and biological systems is
typically strongly overdamped which we also assume here. The extension
to underdamped systems is conceptually straightforward (since we 
consider microscopically reversible dynamics) \cite{Risken}, but since
a broken time-translation invariance in soft and biological matter is \emph{not} tied to momenta, we omit
these for convenience.  
We are strictly interested in the evolution of
$\bx_t$ for $t\gg\tau_b$. It is well known that under certain
technical conditions imposed on the dynamics of the bath \cite{Freidlin}, which we will not
further detail here but are strictly granted for the physical systems
relevant to the discussion, $\bx_t$ evolves according to the It$\mathrm{\hat{o}}$
equation
\begin{equation}
d\bx_t=\mathbf{F}(\bx_t)dt
+ \boldsymbol{\sigma}d\mathbf{W}_t
\label{Langevin}
\end{equation}
where $\mathbf{W}_t$ is a $d$-dimensional vector of independent
Wiener processes whose increments have a Gaussian distribution with
zero mean and variance $dt$,
i.e. $\mathbb{E}[dW_{t,i}dW_{t',j}]=\delta_{ij}\delta(t-t')dt$,
$\mathbb{E}[\cdot]$ denotes the expectation over the ensemble of
Wiener increments and where $\boldsymbol{\sigma}$ is a $d\times d$
symmetric noise matrix.
If momentum
coordinates were included $\boldsymbol{\sigma}$ would be positive semi-definite with
zeros in the sector of position variables and non-zero terms
proportional to the friction constant $\gamma$ in the momentum sector,
and is strictly positive definite with terms $\propto \gamma^{-1}$  for over-damped dynamics
(i.e. for $\gamma\gg 1$) \cite{Risken}). 
We focus on microscopically reversible
dynamics, that is, we consider $d$-dimensional Markovian diffusion
with a $d\times d$ symmetric positive-definite diffusion matrix
$\mathbf{D}=\boldsymbol{\sigma}\boldsymbol{\sigma}^T/2$ and mobility tensor
$\mathbf{M}=\mathbf{D}/k_{\mathrm{B}}T$ (with
$\beta^{-1}\equiv k_{\mathrm{B}}T$ being the thermal energy) in a drift field
$\mathbf{F}(\bx)$, such that $\mathbf{M}^{-1}\mathbf{F}(\bx)=-\nabla \varphi(\bx)$
is a gradient flow. 
The
drift field $\bbf:\mathbb{R}^d\to \mathbb{R}^d$, is either nominally confining
(in this case $\Omega$ is open) or is accompanied by corresponding reflecting boundary conditions at
$\partial \Omega$ (in this case $\Omega$ is closed) thus
guaranteeing the existence of an invariant measure and hence
ergodicity \cite{Freidlin,Risken}.

On the level of probability measures in phase space the dynamics is
governed by the (forward) Fokker-Planck operator $\LL: V\to V$, where $V$ is a complete normed
linear vector space with elements $f\in C^2(\mathbb{R}^d)$. In particular, 
\begin{equation}
  \LL=\nabla\cdot\mathbf{D}\nabla-\nabla\cdot \bbf.
\label{FPE}
\end{equation}
$\bbf$ is assumed to be sufficiently confining,
i.e. $\lim_{\bx\to\infty}P(\bx,t)=0, \forall t$ sufficiently fast to assure that
$\LL$ corresponds to a coercive and densely defined
operator on $V$ with a pure point spectrum \cite{Nier,Chupin,Reed}. 
$\LL$ propagates probability measures $\mu_t(\bx)$ in time, which will throughout
be assumed to possess well-behaved probability density functions
$P(\bx,t)$, i.e. $d\mu_t(\bx)=P(\bx,t)d\bx$. The nullspace of $\LL$
(i.e. the solution of $\LL \peq=0$)  is
the equilibrium (Maxwell-)Boltzmann-Gibbs distribution,
$\peq=Q^{-1}\ee^{-\beta\varphi(\bx)}$, with partition function
$Q=\int_{\Omega}d\bx\ee^{-\beta\varphi(\bx)}$. We define the (forward)
propagator $\hat{U}(t)=\mathrm{e}^{\LL t}$ that is the generator of a
semi-group $\hat{U}(t+t')=\hat{U}(t)\hat{U}(t')$. The formal solution
of the Fokker-Planck equation $(\partial_t-\LL) P(\bx,t)=0$ is thereby
given as $P(\bx,t)=\hat{U}(t)P(\bx,0)$.
The expectation over the ensemble of paths
$\bx_t$ will be denoted by $\langle
\cdot\rangle$ and in the case of a physical observable $\mathcal{B}(\bx_t)$ is given by
\begin{equation}
\langle \mathcal{B}(\bx_t)\rangle\equiv\int\mathcal{B}(\bx)d\mu_t(\bx)\equiv\int_{\Omega}\mathcal{B}(\bx)P(\bx,t)d\bx\equiv\int_{\Omega}\mathcal{B}(\bx)\hat{U}(t)P(\bx,0)d\bx
\label{expect}  
\end{equation}
Part of the analysis will involve the use of spectral theory in Hilbert
space, for which it is convenient to introduce the bra-ket notation;  the 'ket' $|g\rangle$
represents a vector in $V$ written in
position basis as $g(\bx)\equiv\langle\bx|g \rangle$, and the 'bra'
$\langle h |$ as the integral $\int d\bx h^{\dagger}(\bx)$. The scalar
product is defined with the Lebesgue integral $\langle h|g\rangle=\int d\bx
h^{\dagger}(\bx)g(\bx)$. In this notation we
have the following evolution equation for the
 probability density function starting from an initial
condition $|p_0\rangle$: $|p_t\rangle=\mathrm{e}^{\LL t}|p_0\rangle$.
Since the process is ergodic we have $\lim_{t\to\infty}\mathrm{e}^{\LL
  t}|p_0\rangle=|\mathrm{eq}\rangle$, where $\langle\bx|\mathrm{eq}\rangle=\peq$.
We also define the (typically non-normalizable)
'flat' state $\rflat$, such that $\langle\bx\rflat=1$ and
$\langle\text{--}|p_t\rangle=1$. Hence,
$\partial_t\langle\text{--}|p_t\rangle=0$ and $\lflat \LL=0$.

Whereas $\LL$ by itself is not self-adjoint, it is orthogonally
equivalent to a self-adjoint operator, i.e. the operator $\LL_s=\ee^{\beta
  \varphi(\bx)/2}\LL\ee^{-\beta \varphi(\bx)/2}$ is self-adjoint, and,
moreover the operator $\ee^{\beta \varphi(\bx)}\LL$ is
self-adjoint (for a proof see \cite{Risken}). Because any self-adjoint operator in Hilbert space is
diagonalizable, $\LL$ is diagonalizable as well, but with a separate
set of left and right bi-orthonormal eigenvectors $\leftBk$ and
$\rightKk$, respectively. That is, $\LL\rightKk=-\lambda_k\rightKk$
and $\leftBk\LL=-\lambda_k\leftBk$ with
real eigenvalues $\lambda_k\ge 0$ (assured by detailed balance) and where $\lambda_0=0$,
$|\psi^R_0\rangle=|\mathrm{eq}\rangle$, $\langle \psi^L_0|=\lflat$,
and $\leftBk\psi^R_l\rangle=\delta_{kl}$. Moreover, since $\ee^{\beta \varphi(\bx)}\LL$ is
self-adjoint it follows that that $\leftKk=\ee^{\beta \varphi(\bx)}\rightKk$. The resolution
of identity is given by $\mathbf{1}=\sum_k\rightKk\leftBk$ and the
propagator by $\hat{U}(t)=\sum_k\rightKk\leftBk\ee^{-\lambda_kt}$.

The Markovian Green's function of the process $\bx_t$ corresponds to the conditional probability density function for a localized initial
condition $\langle\bx|p_0\rangle=\delta(\bx-\bx_0)$ and is defined as $Q(\bx,t|\bx_0,0) = \langle \bx|\hat{U}(t)|\bx_0\rangle$,
such that the probability density starting from a general
initial condition $|p_0\rangle$ becomes $P(\bx,t,p_0)=\langle \bx|\hat{U}(t)|p_0\rangle\equiv\int d\bx_0 p_0(\bx_0) Q(\bx,t|\bx_0,0)$.
In the spectral representation the Green's function reads
\begin{equation}
Q(\bx,t|\bx_0,0)=\sum_k\psi^R_k(\bx)\psi_k^{L}(\bx_0)\ee^{-\lambda_kt},
\label{SGreen}  
\end{equation}
where the semi-group property means that
$Q(\bx,\tau|\bx_0,0)=Q(\bx,t+\tau|\bx_0,t)$ is independent of $t$ as is
easily verified via
\begin{equation}
\int_{\Omega} d\bx'
Q(\bx,t|\bx',t')Q(\bx',t'|\bx_0,0)=\sum_{k,l}\psi^R_k(\bx) \leftBk\psi^R_l\rangle\psi^{L}_l(\bx_0)\ee^{-\lambda_k(t-t')-\lambda_lt'}\equiv
Q(\bx,t|\bx_0,0),
\label{CKE}  
\end{equation}
where we have used that $\leftBk\psi^R_l\rangle=\delta_{k,l}$.

In the
presence of a time-scale separation giving rise to local equilibrium
the system's dynamics may be coarse-grained further into a
discrete-state Markov jump master equation (see
e.g. \cite{Kampen,Weber_2017}). In this case the configuration space
would be discrete and $d-$dimensional, $\LL$ would be replaced
by a $d\times d$ symmetric stochastic matrix $\mathbf{M}$, 
and the Fokker-Planck
equation by the master equation $\frac{d}{dt} Q=\mathbf{M} Q$. Since
this situation corresponds to an approximate, lower-resolution
dynamics of the system that is mathematically simpler and the mapping
between the Fokker-Planck equation and Markov-state jump dynamics is
well-known \cite{Kampen,Risken,Lapolla_PRR} and does not
introduce any further conceptual changes (the complete spectral-theoretic approach in
particular remains unchanged), we will without any loss of generality
focus on the continuous scenario.   

\setcounter{equation}{0}
\setcounter{figure}{0}
\renewcommand{\thefigure}{B\arabic{figure}}
\renewcommand{\theequation}{B\arabic{equation}}
\section{Dynamics of the projected lower-dimensional observable}
\label{lower}
In order to describe the dynamics of the $r$-dimensional projected observable
$\bq=\BGamma(\bx):\mathbb{R}^d\to\mathbb{R}^r$ with $r<d$ and $\bq$ lying in some
orthogonal system in Euclidean space
$\bq\in\Xi(\mathbb{R}^r)\subset\Omega(\mathbb{R}^d)$, we define
the operator $\mcp_{\mathbf{x}}(\BGamma;\bq)$, such that, when applied
to some function $Z(\bx)\in V$, $\mcp_{\mathbf{x}}(\BGamma;\bq)$ gives (see \cite{Lapolla_f})
\begin{equation}
 \mcp_{\bx}(\BGamma;\bq) Z(\bx)\equiv\int_{\Omega}d\bx\delta(\BGamma(\bx)-\bq)Z(\bx),
  \label{projection}
\end{equation}
where $\delta(\mathbf{y})$ is to be understood in the distributional sense.
We can now define the (in general) non-Markovian
two-point conditional probability density of projected dynamics
starting from $\bq_0\in\Xi_0$, where the subdomain  $\Xi_0$ is not
necessarily simply connected, with the extended operator $\mcp_{\bx}(\BGamma;\bq\in\Xi_0)=\int_{\Xi_0}
  d\bq\mcp_{\bx}(\BGamma;\bq)$ in terms of the single-point and joint
  two-point density $P_{p_0}^0(\bq_0\in\Xi_0)$ and $P_{p_0}(\bq,t,\bq_0\in\Xi_0)$,
  respectively, as 
\begin{equation}
G_{p_0}(\bq,t|\bq_0\in\Xi_0)=\frac{P_{p_0}(\bq,t,\bq_0\in\Xi_0)}{P_{p_0}^0(\bq_0\in\Xi_0)}\equiv\frac{\mcp_{\bx}(\Gamma;\bq)\mcp_{\bx_0}(\BGamma;\bq_0\in\Xi_0)Q(\bx,t|\bx_0,0)p_0(\bx_0)}{\mcp_{\bx_0}(\BGamma;\bq_0\in\Xi_0)p_0(\bx_0)}
\label{projected}
\end{equation}
with the convention that $P_{p_0}(\bq,t,\bq_0)$ and
$G_{p_0}(\bq,t|\bq_0)$ stand for $\Xi_0$ corresponding to a single point $\bq_0$. 
\emph{The full system is said to have a stationary preparation} if and
only if $p_0(\bx_0)=\peq$, whereas \emph{the projected observable is said to have a stationary preparation} if and
only if $\Xi_0=\Xi$. Note
that $\lim_{t\to\infty}P_{p_0}(\bq,t,\bq_0\in\Xi_0)=P_{\eeq}(\bq)\int_{\Xi_0}d\bq_0P_{p_0}(\bq_0)$,
where we have defined
$P_{\eeq}(\bq)\equiv\mcp_{\bx}(\Gamma;\bq)\peq$ as well as
$P_{p_0}(\bq_0)\equiv\mcp_{\bx_0}(\BGamma;\bq_0)p_0(\bx_0)$. In turn
it follows that $\lim_{t\to\infty}G_{p_0}(\bq,t|\bq_0\in\Xi_0)=P_{\eeq}(\bq)$.
Eq.~(\ref{projected}) demonstrates that the entire time evolution of projected dynamics
starting from a fixed condition $\bq_0$ depends on the initial
preparation of the full system $p_0(\bx_0)$ as denoted by the
subscript, which is the first signature of the non-stationary nature
of projected dynamics. In addition, the dynamics described by Eq.~(\ref{projected})
is, except for quite exotic projections $\BGamma(\bx)$, non-Markovian
(see \cite{Lapolla_f}).\\
\indent We can now define averages and two-point correlation functions of
$\bq(t)$. The $n$-th moment of the position averaged over an ensemble of all projected non-Markovian
evolutions prepared in the point $\bq_0$ while the full
system at $t=0$ is prepared in the state $p_0(\bx_0)$ is given by
\begin{equation}
\langle \bq (t)^n\rangle^{\Xi_0}_{p_0}\equiv\int_{\Xi}d\bq \bq^n
G_{p_0}(\bq,t|\bq_0\in\Xi_0), \quad \langle\bq^n\rangle^{\Xi_0}_{p_0}\equiv\int_{\Xi_0}  d\bq_0\bq_0^nP_{p_0}(\bq_0)
\label{Nexpected}
\end{equation}
where we are here only interested in $n=1,2$, whereas the most general tensorial two-point $(0,t)$ (non-aging) correlation
(i.e. covariance) matrix is defined as  
\begin{eqnarray}
\mathbf{C}_{\Xi_0}(t;p_0)\!&\equiv&\!\langle \bq(t)\otimes\bq(0)\rangle^{\Xi_0}_{p_0}-\langle \bq(t)\rangle_{p_0}^{\Xi_0}\otimes\langle\bq\rangle^{\Xi_0}_{p_0}\nonumber\\&=&\!\!\int_{\Xi}\!d\bq\! \int_{\Xi_0}\!d\bq_0 (\bq\otimes\bq_0)
P_{p_0}(\bq,t,\bq_0) - \langle \bq (t)\rangle^{\Xi_0}_{p_0}\otimes \langle\bq\rangle^{\Xi_0}_{p_0},
\label{TNcorr}
\end{eqnarray}
such that $\lim_{t\to\infty}\mathbf{C}(t;p_0)=0,\forall p_0$, where from the
scalar version is in turn obtained by taking the trace
\begin{equation}
C_{\Xi_0}(t;p_0)\equiv\langle \bq(t)\cdot\bq(0)\rangle^{\Xi_0}_{p_0}-\langle \bq(t)\rangle_{p_0}^{\Xi_0}\cdot\langle\bq\rangle^{\Xi_0}_{p_0}=\mathrm{Tr}\mathbf{C}_{\Xi_0}(t;p_0)
\label{Ncorr}
\end{equation}
with the convention $\mathbf{C}_{\Xi_0}(t;P_{\mathrm{eq}})=\langle
\bq(t)\cdot\bq(0)\rangle^{\Xi_0}_{\eeq}-\langle
\bq(t)\rangle_{\eeq}^{\Xi_0}\cdot\langle\bq\rangle^{\Xi_0}_{\eeq}\equiv\mathbf{C}_{\Xi_0}(t)$. We
can equivalently define the time-dependent variance of $\bq(t)$
with $\bq(0)=\bq_0\in\Xi_0$ as
\begin{equation}
\sigma_{\Xi_0}^2(t;p_0)\equiv\langle \bq(t)^2\rangle^{\Xi_0}_{p_0}-(\langle \bq(t)\rangle_{p_0}^{\Xi_0})^2
\label{tvar}
\end{equation}
\noindent \subsection{Spectral theory of projected dynamics}
We now use spectral theory of the Markovian Green's function in
\eql\ref{SGreen}) to analyze the general properties of the
non-Markovian time evolution of the projected lower-dimensional
observable $\bq(t)$. As the initial preparation of the full system
$p_0(\bx_0)$ was found to determine the point-to-point propagation of
the probability density of $\bq$, we
begin by expanding the initial condition of the full system $p_0(\bx_0)$
in the eigenbasis of $\LL$, i.e.
$p_0(\bx_0)=\sum_l|\psi_l^R\rangle\left\langle\psi_l^L|p_0\right\rangle$. The
only assumptions made for $p_0(\bx_0)$ are that it is normalized,
Lebesgue integrable (such that $\left\langle\psi_l^L|p_0\right\rangle$
exists) and locally sufficiently compact to assure that the projection
at time $t=0$ does not project onto an empty set of the observable $\bq_0$.
By further introducing
the elements of the following infinite-dimensional
matrices
\begin{equation}
\Psi_{kl}(\bq)=\leftBk \delta(\BGamma(\bx)-\bq)\rightKl, \quad
\Psi_{kl}(\Xi_0)=\int_{\Xi_0}d\bq\leftBk \delta(\BGamma(\bx)-\bq)\rightKl
\label{element}
\end{equation}
where $\lim_{\Xi_0\to\bq}\Psi_{kl}(\Xi_0)=\Psi_{kl}(\bq)$, we can express $P_{p_0}(\bq,t|\bq_0\in\Xi_0)$
in \eql\ref{projected}) 
as
\begin{equation}
P_{p_0}(\mathbf{q},t,\mathbf{q}_0\in \Xi_0)=\sum_{k}\ee^{-\lambda_kt}\Psi_{0k}(\bq)\sum_l\Psi_{kl}(\Xi_0)\left\langle\psi_l^L|p_0\right\rangle
\label{JnonMarkS}
\end{equation}
and since the preparation of the projected observable is $P_{p_0}(\bq_0\in\Xi_0)=\sum_l\Psi_{0l}(\Xi_0)\left\langle\psi_l^L|p_0\right\rangle$, the
conditional non-Markovian two-point density as
\begin{equation}
G_{p_0}(\mathbf{q},t|\mathbf{q}_0\in \Xi_0)=\frac{\sum_{k}\ee^{-\lambda_kt}\Psi_{0k}(\bq)\sum_l\Psi_{kl}(\Xi_0)\left\langle\psi_l^L|p_0\right\rangle}{\sum_l\Psi_{0l}(\Xi_0)\left\langle\psi_l^L|p_0\right\rangle}.
\label{GnonMarkS}
\end{equation}
For a stationary preparation of the full system,
i.e. $p_0(\bx_0)=P_{\mathrm{eq}}(\bx_0)$, we have that
$\left\langle\psi_l^L|P_{\mathrm{eq}}\right\rangle=\delta_{l,0}$ and
hence $P_{\mathrm{eq}}(\bq\in\Xi_0)=\Psi_{00}(\Xi_0)$ as
well as 
\begin{equation}
  G_{\eeq}(\mathbf{q},t|\mathbf{q}_0\in \Xi_0)=\frac{P_{\mathrm{eq}}(\bq,t,\bq_0\in\Xi_0)}{P_{\mathrm{eq}}(\bq_0\in\Xi_0)}=\frac{\sum_{k}\Psi_{0k}(\bq)\Psi_{k0}(\Xi_0)\ee^{-\lambda_kt}}{\Psi_{00}(\Xi_0)}.
\label{GnonMarkSeq}
\end{equation}
As a result 
\begin{eqnarray}
\langle \bq (t)\rangle^{\Xi_0}_{p_0}&=&\sum_{k}\ee^{-\lambda_kt}\left(\int_\Xi d\bq \bq\Psi_{0k}(\bq)\right)\frac{\sum_l\Psi_{kl}(\Xi_0)\left\langle\psi_l^L|p_0\right\rangle}{\sum_l\Psi_{0l}(\Xi_0)\left\langle\psi_l^L|p_0\right\rangle}\nonumber\\  
\langle \bq (t)\rangle^{\Xi_0}_{\eeq}&=&\sum_{k}\ee^{-\lambda_kt}\left(\int_\Xi d\bq \bq\Psi_{0k}(\bq)\right)\frac{\Psi_{k0}(\Xi_0)}{\Psi_{00}(\Xi_0)}.
\label{NexpectedS}
\end{eqnarray}
Furthermore, we find that
\begin{equation}
\lim_{\Xi_0\to\Xi}\Psi_{kl}(\Xi_0)=\!\int_{\Xi}\!d\bq \leftBk\delta(\BGamma(\bx)-\bq)\rightKl=\leftBk\int_{\Xi}\!d\bq \delta(\BGamma(\bx)-\bq)\rightKl=\langle\psi_k^L|\psi^R_l\rangle=\delta_{k,l},
\label{limit}
\end{equation}
where the order of integration can be exchanged since the delta function in the
distributional sense is smooth (i.e. the limit to a 'true'
delta-function is taken after the integrals) and the domain of
the $\bq$ integration $\Xi$ by definition includes all mappings
$\bq=\BGamma(\bx)$ such that $\int_{\Xi}d\bq
\delta(\BGamma(\bx)-\bq)=1$. As a result
$\lim_{\Xi_0\to\Xi}G_{\eeq}(\mathbf{q},t|\mathbf{q}_0\in
\Xi_0)=\Psi_{00}(\bq)= P_{\eeq}(\bq),\forall t$.  
Using these spectral-theoretic results it follows immediately that the
elements of the general tensorial second moment 
matrix read
\begin{eqnarray}
\langle (\bq(t)\otimes\bq(0))_{ij}\rangle^{\Xi_0}_{p_0}&=&\sum_{k}\ee^{-\lambda_kt}\left(\int_{\Xi_i}dq_i
  q_i\Psi_{0k}(q_i)\right)\sum_l\left\langle\psi_l^L|p_0\right\rangle\left(\int_{\Xi_{0,j}}dq_{0,j}q_{0,j}\Psi_{kl}(q_{0,j})\right)\nonumber\\
\langle (\bq(t)\otimes\bq(0))_{ij}\rangle^{\Xi_0}_{\eeq}&=&\sum_{k}\ee^{-\lambda_kt}\left(\int_{\Xi_i}dq_i
q_i\Psi_{0k}(q_i)\right)\left(\int_{\Xi_{j,0}}dq_{0,j}q_{0,j}\Psi_{k0}(q_{0,j})\right),
\label{TNcorrS}
\end{eqnarray}
which, once plugged into Eq.~(\ref{TNcorr}) together with
Eq.~(\ref{NexpectedS}) and the right member of \eql\ref{Nexpected}),
yield the tensorial correlation (or covariance) matrix $\mathbf{C}(t;p_0)$.
The case treated in the main text, that is, when the projected
coordinate is one-dimensional and the full-system's preparation is
stationary, follows trivially by appropriate simplification of
Eq.~(\ref{projection}) and insertion into \eql\ref{GnonMarkSeq}), which
leads to
\begin{eqnarray}
\!\!\!\!\!\!\!\!\!C(t;\eeq)&\!\equiv&\!\langle q(t)q(0)\rangle^{\Xi_0}_{\eeq}-\langle
q(t)\rangle^{\Xi_0}_\eeq\langle q(0)\rangle_{\eeq}^{\Xi_0}\nonumber\\
&=&\!\sum_{k}\ee^{-\lambda_kt}\!\left(\int_{\Xi}\!\!d
qq\Psi_{0k}(q)\right)\!\!\left(\!\!\int_{\Xi_0}dq_{0}q_{0}\Psi_{k0}(q_{0}) -\frac{\Psi_{k0}(\Xi_0)}{\Psi_{00}(\Xi_0)}\int_{\Xi_0}\!\! dq_0 q_0\Psi_{00}(q_0)\right)\!\!.
\label{EQCorr}
\end{eqnarray}
As we now show in the following section \emph{dynamical
time asymmetry} (i.e. broken time-translation invariance) is inherently
tied to non-Markovian three-point probability density functions of the
projected observable.

\noindent \subsection{Three-point dynamics and breaking of
  time-translation invariance}\label{2B}
In order to describe dynamical time asymmetry we introduce two times, the
so-called ``aging'' (or ``waiting'') time, $t_a$, and the observation time window $\tau=t-t_a$. More
precisely, we consider, as in the previous section, that the \emph{full system} was prepared at $t=0$ in a
general (not necessarily stationary) state $p_0(\bx_0)$, whereby the choice of time origin is dictated by
the initiation of an experiment or the onset of a phenomenon. The
actual observation starts at some later (aging) time $t_a\ge 0$ and is
carried out until a time $t$ and hence has a duration $\tau=t-t_a$. An
example of a non-stationary preparation of a full system would be a temperature quench of a system equilibrated at some
different temperature. We assume, as before, that only the
lower-dimensional observable $\bq(t)$ is observed for all times $t\ge 0$.   

We now define time-delayed, ``aging'' observables. The normalized tensorial aging correlation
matrix is defined as
\begin{eqnarray}
\hat{\mathbf{C}}^{\Xi_0}_{t_a}(\tau;p_0)\equiv\frac{\mathbf{C}^{\Xi_0}_{t_a}(\tau;p_0)}{\mathbf{C}^{\Xi_0}_{t_a}(0;p_0)}=\frac{\langle
  \bq(\tau+t_a)\otimes\bq(t_a)\rangle^{\Xi_0}_{p_0}-\langle
  \bq(\tau+t_a)\rangle^{\Xi_0}_{p_0}\otimes\langle\bq(t_a)\rangle^{\Xi_0}_{p_0}}{\langle
  \bq(t_a)\otimes\bq(t_a)\rangle^{\Xi_0}_{p_0}-\langle
  \bq(t_a)\rangle_{p_0}^{\Xi_0}\otimes\langle\bq(t_a)\rangle^{\Xi_0}_{p_0}}
\label{TaNcorr}
\end{eqnarray}
such that
$\hat{C}_{t_a}(\tau;p_0)\equiv\mathrm{Tr}\hat{\mathbf{C}}_{t_a}(\tau;p_0)$
and for the one-dimensional coordinate starting from a system prepared
in a stationary state that is studied in the main paper  
\begin{equation} 
\hat{C}^{\Xi_0}_{t_a}(\tau,\eeq)\equiv \hat{C}^{\Xi_0}_{t_a}(\tau)\equiv\frac{C^{\Xi_0}_{t_a}(\tau)}{C^{\Xi_0}_{t_a}(0)}=\frac{\langle q(\tau+t_a)q(t_a)\rangle^{\Xi_0}-\langle q(\tau+t_a)\rangle^{\Xi_0}\langle q(t_a)\rangle^{\Xi_0}}{\langle q(t_a)^2\rangle^{\Xi_0}-(\langle q(t_a)\rangle^{\Xi_0})^2}.
 \label{a_corr}
\end{equation}
From the definitions of aging observables in
Eqs.~(\ref{TaNcorr}-\ref{a_corr}) it follows that these are inherently
tied to three-point probability density functions at times $0,t_a,$
and $t_a+\tau$. The full system's dynamics,  corresponding to a
Hamiltonian dynamics coupled to a Markovian heat bath, is Markovian and
time-translation invariant. The
three-point joint density therefore reads
\begin{equation}
 P^{p_0}_{\mathrm{full}}(\bx,t_a+\tau,\bx',t_a,\bx_0)=Q(\bx,t_a+\tau|\bx',t_a)Q(\bx',t|\bx_0,0)p_0(\bx_0).
\label{Mthreej}
\end{equation}
Using the definitions from the previous section and
introducing the shorthand notation
$\mathbf{\mcp}_{\bx,\bx',\bx_0}(\BGamma;\bq,\bq',\bq_0\in\Xi_0)\equiv\mcp_{\bx}(\BGamma;\bq)\mcp_{\bx'}(\BGamma;\bq')\mcp_{\bx_0}(\BGamma;\bq_0\in\Xi_0)$
the \emph{three-point joint density} is defined as
\begin{eqnarray}
  P_{p_0}(\bq,t_a+\tau,\bq',t_a,\bq_0\in\Xi_0)&\equiv&\mathbf{\mcp}_{\bx,\bx',\bx_0}(\BGamma;\bq,\bq',\bq_0\in\Xi_0)P^{p_0}_{\mathrm{full}}(\bx,t_a+\tau,\bx',t_a,\bx_0)\nonumber\\
 &=&\sum_{k,l}\ee^{-\lambda_k\tau-\lambda_l t_a}\Psi_{0k}(\bq)\Psi_{kl}(\bq')\sum_m\Psi_{lm}(\Xi_0)\langle\psi_{m}^L|p_0\rangle.
\label{threej}
\end{eqnarray}
Under the milder (as far as the non-stationarity of $\bq(t)$ is concerned) assumption that the full system at $t=0$ is in equilibrium, that is
$p_0(\bx_0)=\peqo$ as we have assumed in the main text,
$\langle\psi_{m}^L|P_{\eeq}\rangle=\delta_{m,0}$ and Eq.~(\ref{threej})
simplifies to
\begin{eqnarray}
  P_{\eeq}(\bq,t_a+\tau,\bq',t_a,\bq_0\in\Xi_0)&\equiv&\mathbf{\mcp}_{\bx,\bx',\bx_0}(\BGamma;\bq,\bq',\bq_0\in\Xi_0)P^{P_{\eeq}}_{\mathrm{full}}(\bx,t_a+\tau,\bx',t_a,\bx_0)\nonumber\\
 &=&\sum_{k,l}\ee^{-\lambda_k\tau-\lambda_l t_a}\Psi_{0k}(\bq)\Psi_{kl}(\bq')\Psi_{l0}(\Xi_0).
\label{threejEq}
\end{eqnarray}
The corresponding three-point conditional probability densities are in
turn defined by 
\begin{eqnarray}
  \label{Greens3}
 \!\!\!\!G_{p_0}(\bq,t_a+\tau,\bq',t_a|\bq_0\in\Xi_0)&\equiv&\frac{P_{p_0}(\bq,t_a+\tau,\bq',t_a,\bq_0\in\Xi_0)}{P_{p_0}(\bq_0\in\Xi_0)}\nonumber\\
  &=&\frac{\sum_{k,l}\ee^{-\lambda_k\tau-\lambda_l t_a}\Psi_{0k}(\bq)\Psi_{kl}(\bq')\sum_m\Psi_{lm}(\Xi_0)\langle\psi_{m}^L|p_0\rangle}{\sum_l\Psi_{0l}(\Xi_0)\left\langle\psi_l^L|p_0\right\rangle},\\
 \!\!\!\!G_{\eeq}(\bq,t_a+\tau,\bq',t_a|\bq_0\in\Xi_0)&\equiv&\frac{P_{\eeq}(\bq,t_a+\tau,\bq',t_a,\bq_0\in\Xi_0)}{P_{\mathrm{eq}}(\bq_0\in\Xi_0)}\nonumber\\
  &=&\frac{\sum_{k,l}\ee^{-\lambda_k\tau-\lambda_l t_a}\Psi_{0k}(\bq)\Psi_{kl}(\bq')\Psi_{l0}(\Xi_0)}{\Psi_{00}(\Xi_0)}.
\label{Greens3e}
\end{eqnarray}
A broken time-translation invariance is, however, most explicitly visible
by means of  what we will refer to as the two-point conditioned Green's function:  
\begin{eqnarray}
  \label{propa}
 \!\!\!\!\tilde{G}_{p_0}(\bq,t_a+\tau|\bq',t_a,\bq_0\in\Xi_0)&\equiv&\frac{P_{p_0}(\bq,t_a+\tau,\bq',t_a,\bq_0\in\Xi_0)}{P_{p_0}(\bq,t_,\bq_0\in\Xi_0)}\nonumber\\
  &=&\frac{\sum_{k,l}\ee^{-\lambda_k\tau-\lambda_l t_a}\Psi_{0k}(\bq)\Psi_{kl}(\bq')\sum_m\Psi_{lm}(\Xi_0)\langle\psi_{m}^L|p_0\rangle}{\sum_{k}\ee^{-\lambda_kt_a}\Psi_{0k}(\bq')\sum_l\Psi_{kl}(\Xi_0)\left\langle\psi_l^L|p_0\right\rangle},\\
 \!\!\!\!\tilde{G}_{\eeq}(\bq,t_a+\tau|\bq',t_a,\bq_0\in\Xi_0)&\equiv&\frac{P_{\eeq}(\bq,t_a+\tau,\bq',t_a,\bq_0\in\Xi_0)}{P_{\mathrm{eq}}(\bq,t_a,\bq_0\in\Xi_0)}\nonumber\\
  &=&\frac{\sum_{k,l}\ee^{-\lambda_k\tau-\lambda_l t_a}\Psi_{0k}(\bq)\Psi_{kl}(\bq')\Psi_{l0}(\Xi_0)}{\sum_{k}\ee^{-\lambda_kt_a}\Psi_{0k}(\bq')\Psi_{k0}(\Xi_0)}.
\label{propaEq}
\end{eqnarray}
By means of Eqs.~(\ref{Greens3}) and (\ref{Greens3e}) we can now
determine aging expectation values entering \eql\ref{TaNcorr}) and
\eql\ref{a_corr}), which, for a general matrix element $\langle
q_i(\tau+t_a)q_j(t_a)\rangle$ read
\begin{eqnarray}
\langle q_i(\tau+t_a)q_j(t_a)\rangle_{p_0}^{\Xi_0}&=&\int_{\Xi_j} dq_i\int_{\Xi_j}
dq_jq_iq_jG_{p_0}(q_i,t_a+\tau,q_j,t_a|\bq_0\in\Xi_0)\nonumber\\
\langle q_i(\tau+t_a)q_j(t_a)\rangle_{\eeq}^{\Xi_0}&=&\int_{\Xi_i} dq_i\int_{\Xi_j} dq_jq_iq_jG_{\eeq}(q_i,t_a+\tau,q_j,t_a|\bq_0\in\Xi_0)
\label{agingobs}
\end{eqnarray}
The dynamics of the projected observable $\bq(t)$ is typically
referred to as aging if correlation functions like $\hat{\mathbf{C}}_{t_a}(\tau;p_0),\hat{C}_{t_a}(\tau;p_0)$ and/or
  $C_{t_a}(\tau)$ defined in Eqs.~(\ref{TNcorr}-\ref{Ncorr}) depend on
$t_a$. However, the observables in Eq.(\ref{agingobs}) only capture
linear correlations in systems with broken time-translation
invariance, and moreover display a $t_a$-dependence even in Markovian systems which are
time-translation invariant but evolve from a non-stationary initial condition (see Lemma
\ref{nonconc} below).
These correlation functions are therefore by no means conclusive
indicators of broken time-translation invariance. We
therefore propose the \emph{time asymmetry index}, $\AI$ -- a new,
conclusive (albeit not unique) indicator of broken time-translation invariance, which we define as
\begin{equation}
\AI_{\Xi_0}(t_a,\tau)\equiv\hat{\mathcal{D}}_{\bq,\bq'}\left[G_{p_0}(\bq,\tau+t_a,\bq',t_a|\bq_0\in\Xi_0)||G_{p_0}(\bq,\tau|\bq')G_{p_0}(\bq',t_a|\bq_0\in\Xi_0)\right],
\label{AI}  
\end{equation}
where we have introduced the Kullback-Leibler divergence (or relative entropy)
\begin{equation}
  \hat{\mathcal{D}}_{\textbf{y}_1,\textbf{y}_2}[p||q]\equiv
\iint\limits_{{\rm
    supp}\,p}
d\mathbf{y}_1d\mathbf{y}_2p(\mathbf{y}_1,\mathbf{y}_2)\ln\frac{p(\mathbf{y}_1,\mathbf{y}_2)}{q(\mathbf{y}_1,\mathbf{y}_2)},
\label{KL}
\end{equation}
which has the property
$\hat{\mathcal{D}}_{\textbf{y}_1,\textbf{y}_2}[p||q]\ge 0$ with the
equality being true if and only if $p(\mathbf{y}_1,\mathbf{y}_2)$ is
equal to $q(\mathbf{y}_1,\mathbf{y}_2)$ almost everywhere \cite{klarticle}.
The rationale behind this choice is that it is defined to measure
exactly the existence and degree of broken time-translation
invariance and we will use this property in the following section to assert the necessary and
sufficient conditions for the emergence of dynamical time asymmetry. 
We are now in a position to prove the central claims in the
manuscript.

\setcounter{equation}{0}
\setcounter{figure}{0}
\renewcommand{\thefigure}{C\arabic{figure}}
\renewcommand{\theequation}{C\arabic{equation}}
\section{Main theorems with proofs}\label{Oh}
\theoremstyle{definition}
\begin{definition}{Time-translation invariance
    \cite{Noether,Reichl}.}\label{Def1}
The dynamics of the observable $\bq(t)$ resulting from the projection defined in \eql\ref{projection})
of the full Markovian dynamics $\bx_t$ evolving according to
Eq.~(\ref{Langevin}) is said to relax to equilibrium in a 
time-translation invariant manner (i.e. stationary) if and only if  the two-point conditioned Green's
function in Eqs.~(\ref{propa}-\ref{propaEq}) does not depend on $t_a$,
that is
\begin{equation*}
\tilde{G}_{p_0}(\bq,t_a+\tau|\bq',t_a,\bq_0\in\Xi_0)=\tilde{G}_{p_0}(\bq,t'+\tau|\bq',t',\bq_0\in\Xi_0),\forall
\tau,t'>0.
\end{equation*}
\label{inv}
\end{definition}
\theoremstyle{definition}
\begin{definition}{Dynamical time asymmetry.}\label{Def2}
The dynamics of the projected observable
  $\bq(t)$ is said to be dynamically time asymmetric if its relaxation to equilibrium is not
time-translation invariant.
\end{definition}
\theoremstyle{definition}
\begin{definition}{Trivial non-stationarity.}
The dynamics of the projected observable $\bq(t)$ is said to be
trivially non-stationary if the relaxation is time-translation
invariant but evolves from a non-equilibrium initial condition of the
full system, $p_0(\bx_0)\ne P_{\eeq}(\bx_0)$. 
\label{trivial}
\end{definition}  
\begin{theorem}
\label{invariant}  
The dynamics of the observable $\bq(t)$ resulting from the projection defined in \eql\ref{projection})
of the full Markovian dynamics $\bx_t$ evolving according to Eq.~(\ref{Langevin})
is time-translation invariant if and only if at least one of the
following is true:\vspace{0.2cm}\\
(1)  the projected dynamics $\bq(t)$ is Markovian\\
(2)  the full system and projected observable are both prepared in and sampled from
  equilibrium, that is $p_0(\bx_0)=\peq$, $\Xi_0\to\Xi$ such that
  $\lim_{\Xi_0\to\Xi}P_{\eeq}(\bq_0\in\Xi_0)\to 1$.\vspace{0.2cm}\\   
If either of these two assumptions is true
$\AI_{\Xi_0}(t_a,\tau)=0,\forall t_a,\tau>0$.
\end{theorem}

\begin{proof}
We first prove sufficiency. If the projection $\mcp$ is such that
\emph{1.} above holds then
$G_{p_0}(\bq,\tau+t_a,\bq',t_a|\bq_0\in\Xi_0)=G_{p_0}(\bq,\tau|\bq')G_{p_0}(\bq',t_a|\bq_0\in\Xi_0)$
for any $\tau,t_a,\bq,\bq',\bq_0,\Xi_0$ and $p_0(\bx_0)$, such that the
logarithmic term in \eql\ref{KL}) is identically zero everywhere and
hence $\AI_{\Xi_0}(t_a,\tau)=0, \forall t_a,\tau$. Conversely,
if \emph{2.} is true then due to \eql\ref{limit}) we have $G_{\eeq}(\mathbf{q},t|\mathbf{q}_0\in
\Xi)=\Psi_{00}(\bq)= P_{\eeq}(\bq),\forall t$ and according to
\eql\ref{Greens3}) also
$G_{p_0}(\bq,t_a+\tau,\bq',t_a|\bq_0\in\Xi)=G_{p_0}(\bq,\tau|\bq')$,
such that the
logarithmic term in \eql\ref{KL}) is again identically zero everywhere and
hence $\AI_{\Xi_0}(t_a,\tau)=0, \forall t_a,\tau$. This proves
sufficiency. 

To prove necessity we first recall that
$\hat{\mathcal{D}}_{\textbf{y}_1,\textbf{y}_2}[p||q]= 0$ if and only if $p(\mathbf{y}_1,\mathbf{y}_2)$ is
equal to $q(\mathbf{y}_1,\mathbf{y}_2)$ almost everywhere
\cite{klarticle}. In addition, as a result of \eql\ref{element}) and irrespective of the projection (as long
as it does not project onto an empty set) the time evolution of
$G_{p_0}(\bq,\tau+t_a,\bq',t_a|\bq_0\in\Xi_0)$ in \eql\ref{Greens3})
and \eql\ref{Greens3e}) as well as $G_{p_0}(\bq,\tau|\bq')$ and
$G_{p_0}(\bq',t_a|\bq_0\in\Xi_0)$ in \eql\ref{GnonMarkS}) is smooth and
continuous $\forall t_a>0, \tau>0$. Moreover,
$\Psi_{0k}(\bq)\ne\Psi_{lk}(\bq)$ for $l\ne 0$ except for potentially
on a set of $\bq$ with zero measure because of \eql\ref{element}) and
since $\langle \psi^L_0|$ and $\langle \psi^L_l|$ are linearly
independent. Therefore, because
$G_{p_0}(\bq,\tau+t_a,\bq',t_a|\bq_0\in\Xi_0)\ge 0, \forall
\tau,t_a,\bq,\bq',\bq_0,\Xi_0$ and $p_0(\bx_0)$ the Kullback-Leibler
divergence in \eql\ref{KL}) cannot not be zero almost everywhere
except if either one or both of the statements \emph{1.} or \emph{2.} above are
true. This completes the proof of necessity.
\end{proof}

\begin{corollary}
\label{binvariant}    
The dynamics of the projected observable $\bq(t)$ displays a dynamical time asymmetry 
 as soon as the projection $\mcp$ renders it non-Markovian and
it is initially not prepared in, and averaged over, an equilibrium
initial condition, i.e. $P_{p_0}(\bq_0\in\Xi_0)\ne
P_{\eeq}(\bq_0\in\Xi)=1$. If this is true then
$\AI_{\Xi_0}(t_a,\tau)>0$ at least on a dense set of
$t_a$ and $\tau$ with non-zero measure. 
\end{corollary}  
\begin{proof}
The proof follows immediately from a straightforward extension of the proof of Theorem \ref{invariant}.
\end{proof}
\begin{lemma}
\label{nonconc}  
Aging 
correlation functions like $\hat{\mathbf{C}}_{t_a}(\tau;p_0),\hat{C}_{t_a}(\tau;p_0)$ and/or
$C_{t_a}(\tau)$ defined in Eqs.~(\ref{TNcorr}-\ref{Ncorr}) are not
conclusive indicators of the dynamical time asymmetry because they cannot discriminate between
trivial non-stationarity and broken time-translation invariance, that is, they can display
a dependence on $t_a$  even if the relaxation to equilibrium is
time-translation invariant.
\end{lemma}  

\begin{proof}
A simple example suffices to prove this claim. Consider that the
observable $\by_t$ is evolving according to Markov dynamics (for
conditions imposed on $\mcp$ for this to occur please see \cite{Lapolla_f})  with Green's function
$Q(\by,t|\by_0)$. It is not difficult to show that the relaxation to
equilibrium is time-translation invariant. Namely, consider,
the probability density of $\by$ at a time $\tau+t_a$ given that at time $t_a$ the system
was found in a point $\bq'$ whereby it evolved there from an
initial probability density $p_0(\by)$:
\begin{eqnarray}
\tilde{G}_{p_0}(\by,\tau+t_a|\by',t_a,\by_0)&\equiv&\frac{\int_{\Xi_0}d\by_0 Q(\by,\tau+t_a|\by',t_a)Q(\by',t_a|\by_0)p_0(\by_0)}{\int_{\Xi_0}d\by_0
  Q(\by',t_a|\by_0)p_0(\by_0)}\nonumber\\
&=& Q(\by,\tau+t_a|\by',t_a)=Q(\by,\tau|\by'),
\end{eqnarray}
where we allow (redundantly) and under-sampling of $p_0$ by setting $\Xi_0\ne\Xi$. 
Clearly, and expectedly, the relaxation process $\by(t_a)\to \by(t_a
+\tau)$ is  time-translation invariant -- it depends only on $\by'=\by(t_a)$ but
does not depend on how this state was reached.
Analogously to \eql\ref{Greens3}) we also define the
three-point joint probability density of $\by$ evolving from an
initial probability density $p_0(\by)$
\begin{eqnarray}
G_{p_0}(\by,\tau+t_a,\by',t_a|\Xi_0)&\equiv&\frac{\int_{\Xi_0}d\by_0 Q(\by,\tau+t_a|\by',t_a) Q(\by',t_a|\by_0)p_0(\by_0)}{\int_{\Xi_0}d\by_0p_0(\by_0)}\nonumber\\
&=&Q(\by,\tau|\by')
  Q(\by',t_a|\Xi_0),
\label{Markov3}  
\end{eqnarray}
where the propagation
from $\by'=\bq(t_a)\to \bq(t_a +\tau)$ only depends on $\bq(t_a)$ but
not on how this state was reached. It follows immediately that the
time asymmetry index $\AI$ for this process is identically zero (See Theorem~\ref{invariant} and Corollary~\ref{binvariant}).
Nevertheless, because $G_{p_0}(\by,\tau+t_a,\by',t_a|\Xi_0)$
in \eql\ref{Markov3}) depends on the probability that the system is
found at time $t_a$ in $\bq'$ (but not on how it got there) the aging
correlation function obtained from \eql\ref{Markov3}) would display
a dependence on $t_a$ as long as $p_0(\by_0)\ne P_{\eeq}$ and
$\Xi_0\ne\Xi$, which
implies non-stationarity in a trivial sense (i.e. this would equally
well be the case even 
simple Brownian diffusion of a particle in a box, which is manifestly
time-translation invariant). 

Conversely, it is also possible that the relaxation is indeed not
translationally invariant according to Definition~\ref{inv} but a dependence on $t_a$
only arises in the evolution of higher order (even or odd)
correlation functions and is not visible in first order correlations in
Eqs.~(\ref{TaNcorr}-\ref{a_corr}). 
\end{proof}

\begin{observation}{Characteristic scaling of aging correlation functions.}
An interesting and very common observation in the existing
literature on glassy, aging dynamics is an irreducible
structure $C_{t_a}(\tau;p_0)=f(\tau+t_a,t_a)$ (see
e.g. \cite{Herisson,LN,Ritort,Franz}) frequently accompanied by a
characteristic power-law scaling of aging autocorrelation functions
\cite{Franz,Ritort,Jeremy,LN}. A particularly striking observation is the
frequently observed so-called 'full aging' regime where the observation
time window $\tau$ becomes much longer than the aging time $t_a$,
$\tau\gg t_a$, and the following simple scaling emerges
$C_{t_a}(\tau;p_0)\propto t_a/\tau$ \cite{Bouchaud_92,Amir,PNAS,Rodriguez_2003}.
Below we explain the emergence of irreducible
structure as well as power-law-decaying aging autocorrelation
functions incl. the full aging within the context of our spectral-theoretic approach.
\end{observation}

\begin{theorem}{A representation result.}
\label{repr}  
Let $\varphi$ be a positive real number smaller than 1, $0<\varphi<1$,
and the functions $g_1(t):\mathbb{R}^+\to\mathbb{R}$ and $g_2(\tau,t):\mathbb{R}^+\times\mathbb{R}^+\to\mathbb{R}$ be smooth for $t>0$ and
$t,\tau>0$, respectively. A matrix element of aging correlation
functions, $C_{t_a,i,j}(\tau;p_0)$ or $C_{t_a}(\tau;p_0)$, defined in
Eqs.~(\ref{TaNcorr}-\ref{a_corr}) has the following irreducible
structure in the form of a stationary contribution $g_1(t)$ and a
non-stationary contribution $g_2(\tau,t)$:
\begin{equation*}
C_{t_a,i,j}(\tau;p_0) = (1-\varphi)g_1(\tau)+\varphi g_2(\tau,t_a).
\end{equation*}
\end{theorem}

\begin{proof}
We recall the definition of aging expectation values in
\eql\ref{agingobs}) and simply split the double sum in
Eqs.~(\ref{Greens3}) and (\ref{Greens3e}) into two parts
$\sum_{k\ge 0}\sum_{l\ge 0}\to\sum_{k\ge 0;l=0}+\sum_{k\ge
  0}\sum_{l\ge 1}$. The second term in the numerator of
Eqs.~(\ref{TaNcorr}-\ref{a_corr}) is nominally non-stationary
(i.e. depends on both $\tau=t-t_a$ and $t_a$). Collecting terms we obtain the
representation stated in the theorem.
\end{proof}

\theoremstyle{definition}
\begin{definition}{Self-similar dynamics.}
 \label{selfsim} 
  Let us write a general non-aging correlation function in \eql\ref{TNcorr}) as
  \begin{equation*}
\mathbf{C}_{ij}(t;p_0)=\sum_{k>0}w^{ij}_k(\Xi_0;p_0)\ee^{-\lambda_k t}.
  \end{equation*}   
The dynamics of the projected observable $\bq(t)$ is said to be
(transiently) self-similar on a time-scale
$0< t \lesssim \lambda_{k_{\mathrm{min}}}^{-1}$ (see e.g. \cite{Montroll})
if a time-scale change $t\to t\delta$ does not change the relaxation
beyond a renormalization of the weights,
$w^{ij}_k\to\tilde{w}^{ij}_k(\delta)$. That is
if $\lambda_k$
and $w_k$ are not independent, such that $\exists k_{\mathrm{min}}\in \mathbb{Z}^+$ and constants
$(\tau_0,\alpha,\delta,y)\in\mathbb{R}^+$ with $0<\delta<1$ and $y<1$ such that
for $k\in \mathbb{Z}^+>k_{\mathrm{min}}$ we have $\lambda_k=(\delta^{k}\tau_0)^{-1}$ and $w_k=\delta^{-\alpha
  k}y$. Then we have, on the time-scale $0<\tau_0\ll t\lesssim
\lambda_{k_{\mathrm{min}}}^{-1}$
  \begin{eqnarray}
\mathbf{C}_{ij}(t;p_0)&=&
\sum_{0<k<k_{\mathrm{min}}}w^{ij}_k(\Xi_0;p_0)\ee^{-\lambda_k
  t}+\sum_{k\ge k_{\mathrm{min}}}^{\infty}w^{ij}_k(\Xi_0;p_0)\ee^{-\lambda_k
  t}\nonumber\\
&\simeq& \mathrm{const} + y\int_{k_{\mathrm{min}}}^\infty dx \delta^{-\alpha
  x}\ee^{-t\delta^{-x}/\tau_0}\nonumber\\
&=& \mathrm{const}+
\frac{y}{\ln\delta^{-1}}\left(\frac{t}{\tau_0}\right)^{-\alpha}\Gamma\left(\alpha,\lambda_{k_{\mathrm{min}}}t\right)\nonumber\\
&=&\mathrm{const}+\frac{y}{\ln\delta^{-1}}\left[\Gamma(\alpha)\left(\frac{t}{\tau_0}\right)^{-\alpha}
-\mathcal{O}\left(\frac{\delta^{-\alpha k_{\mathrm{min}}}}{\alpha}\right)\right]
\label{ssim}
  \end{eqnarray}   
where $\Gamma(\alpha)$ and $\Gamma(\alpha,z)$ denote the complete and upper
incomplete Gamma functions, respectively, and $\simeq$ stands for
asymptotic equality, that is that the fraction of the left and the
right hand side converges to 1 for  $0<\tau_0\ll t\lesssim
\lambda_{k_{\mathrm{min}}}^{-1}$. $\mathbf{C}_{ij}(t;p_0)$ thus transiently decays asymptotically according to a power-law with exponent $\alpha$.
\end{definition}

\begin{proposition}{Self-similar time asymmetric dynamics.}
\label{similar}
Let us write the non-normalized aging correlation function, i.e. the numerator in
Eqs.~(\ref{TaNcorr}-\ref{a_corr}), compactly as 
\begin{equation}
\mathbf{C}_{t_a,ij}(t;p_0)=\sum_{k>0}\sum_{l>0}\left(w^{ij}_{kl}(\Xi_0;p_0)\ee^{-\lambda_k
  (t-t_a)-\lambda_lt_a}- \tilde{w}^{ij}_{kl}(\Xi_0;p_0)\ee^{-\lambda_kt-\lambda_lt_a}\right).
\label{agingcor}
\end{equation}
We now extend the idea of self-similar scaling in Definition
\ref{selfsim} to aging correlations. Let $C_1,C_2,C_4\in\mathbb{R}^+$, $C_3\in\mathbb{R}$,
$\lambda_{k}=(\delta_1^k\tau_0)^{-1}$, $w_{kl}=\delta_1^{-\alpha
  k}\delta_2^{-\alpha l}y_1$, $\tilde{w}_{kl}=\delta_3^{-\alpha
  k}\delta_2^{-\alpha l}y_1$ for suitably chosen $0<\delta_1,\delta_2,\delta_3<1$,
$y_1<1$, $\tau=t-t_a$ and $(\tau_0,\alpha,k_\mathrm{min})$ as in Definition
\ref{selfsim}. Then we have, for $0<\tau_0\ll t,\tau,t_a\lesssim
\lambda_{k_{\mathrm{min}}}^{-1}$ asymptotically
\begin{eqnarray}
  \mathbf{C}_{t_a,ij}(t;p_0)&\simeq&
  C_1\!+y_1\Gamma(\alpha)\left[\frac{C_2}{\ln\delta_1^{-1}}\left(\frac{\tau}{\tau_0}\right)^{-\alpha}\!\!\!-
    \frac{C_4}{\ln\delta_3^{-1}}\left(\frac{t}{\tau_0}\right)^{-\alpha}\right]\nonumber\\
&+&\frac{y_1\Gamma(\alpha)}{\ln\delta_2^{-1}}\left(\frac{t_a}{\tau_0}\right)^{-\alpha}\left[C_3+\frac{\Gamma(\alpha)}{\ln\delta_1^{-1}}\left(\frac{\tau}{\tau_0}\right)^{-\alpha}-\frac{\Gamma(\alpha)}{\ln\delta_3^{-1}}\left(\frac{t}{\tau_0}\right)^{-\alpha}\right]
  +\!\mathcal{O}(\kappa_{\mathrm{max}})
\label{powr}
\end{eqnarray}
where
$\kappa_{\mathrm{max}}\equiv\displaystyle{\alpha^{-1}\max\left(\frac{C_2\delta_1^{-\beta_{\mathrm{min}}}}{\ln\delta_1^{-1}},\frac{C_4\delta_3^{-\beta_{\mathrm{min}}}}{\ln\delta_3^{-1}},\frac{C_3\delta_2^{-\beta_{\mathrm{min}}}}{\ln\delta_2^{-1}},\frac{\Gamma(\alpha)(\delta_1\delta_2)^{-\beta_{\mathrm{min}}
  }}{\ln\delta_1\ln\delta_2},\frac{\Gamma(\alpha)(\delta_3\delta_2)^{-\beta_{\mathrm{min}}
  }}{\ln\delta_3\ln\delta_2}\right)}$ and  $\beta_{\mathrm{min}}=\alpha
k_{\mathrm{min}}$. On a ``good'' scale of $t_a$, i.e. where
$y_1\Gamma(\alpha)C_3(\tau_0/t_a)^\alpha /\ln\delta^{-1}_1 \simeq{\rm
  const}\equiv B$ is effectively 
constant (that is, varies slowly) with respect to $(t_a/t)^\alpha$
then 
  \begin{equation}
    \mathbf{C}_{t_a,ij}(t;p_0)\simeq C_1+B\left[1+\frac{C_2\ln \delta_2}{C_3\ln\delta_1}\left(\frac{t_a}{\tau}\right)^\alpha-\frac{C_4\ln \delta_3}{C_3\ln\delta_1}\left(\frac{t_a}{\tau+t_a}\right)^\alpha\right]+\mathcal{O}\left(\left[\frac{\tau_0^2}{t_a\tau}\right]^\alpha\right)
  \end{equation} 
  Moreover, when $\tau\gg t_a$ we have the the ``anomalous full aging'' scaling 
  \begin{eqnarray}
    \mathbf{C}_{t_a,ij}(t;p_0)&\simeq& C_1+B\left[1+\frac{C_2\ln
        \delta_2-C_4\ln
        \delta_3}{C_3\ln\delta_1}\left(\frac{t_a}{\tau}\right)^\alpha+\alpha\frac{C_4\ln
        \delta_3}{C_3\ln\delta_1}\left(\frac{t_a}{\tau}\right)^{\alpha+1}\right]+\mathcal{O}\left(\left[\frac{t_a}{\tau}\right]^{\alpha+2}\right)\nonumber\\
        &=& B_1+ B_2 \left(\frac{t_a}{\tau}\right)^\nu
    +\mathcal{O}\left(\left[\frac{t_a}{\tau}\right]^{\alpha+2}\right)
  \label{full_a}  
  \end{eqnarray} 
 where $\nu=\alpha$ if $C_2\ln \delta_2\ne C_4\ln \delta_3$ and
 $\nu=\alpha+1$ otherwise.
\end{proposition}

\begin{proof}
  To prove the proposition we split each of the double sums as
\begin{equation*}
  \sum_{k>0}\sum_{l>0}=\sum_{0<k< k_\mathrm{min}}\sum_{0<l<
    k_\mathrm{min}}+\sum_{k\ge k_\mathrm{min}}\sum_{0<l<
    k_\mathrm{min}}+\sum_{0<k< k_\mathrm{min}}\sum_{l\ge
    k_\mathrm{min}}+\sum_{k\ge k_{\mathrm{min}}}\sum_{l \ge k_{\mathrm{min}}}
\end{equation*}
and note that $\int_{0}^{k_{\mathrm{min}}} dx \delta^{-\alpha
  x}\ee^{-t/(\delta^x\tau_0)}\simeq \mathrm{const}$ for
$t\lesssim\lambda_{k_{\mathrm{min}}}^{-1}$. The rest follows directly from
the computation in Definition
\ref{selfsim} upon rearranging and collecting terms. To the the
``anomalous full aging'' scaling we expand
$(t_a[\tau+t_a])^{\alpha}=(1+\tau/t_a)^{-\alpha}=(t_a/\tau)^\alpha(1-\alpha\tau/t_a+\mathcal{O}((t_a/\tau)^2)$
and collect terms. Upon identifying the constants $B_1$ and $B_2$ we
arrive at Eq.~\eqref{full_a} which completes the proof.
\end{proof}
\begin{remark}
Note that the scaling-from in Definition~\ref{selfsim} arises,
    for example, when the observable corresponds to an internal
    distance within a single macromolecule (such e.g as the Rouse chain) \cite{Rouse_self_sim}
     or within individual protein
    molecules \cite{Granek,proteins_self_sim}, as well as in diffusion on
    fractal objects \cite{fractals}.
\end{remark}

\begin{proposition}{Logarithmic relaxation and 'full aging'.}
\label{full}
Let $\mathbf{C}_{t_a,ij}(t;p_0)$ be written as in \eql\ref{agingcor}) in
Proposition~\ref{similar} and let $\alpha=0$, $\delta_1=\delta_3$ and $C_2=C_4=C\in\mathbb{R}^+$  (i.e. $1/f$ self-similar
scaling with logarithmic relaxation
\cite{Amir,PNAS,Bouchaud_92}; in fact a simple change of
integration variable $x\to \delta^{-x}$ in
Eqs.~(\ref{ssim})-(\ref{powr}) with $\alpha=0$ shows that this case is
mathematically equivalent to the analysis in \cite{Amir,PNAS}). Then we have, for $0<\tau_0\ll t,\tau,t_a\lesssim
\lambda_{k_{\mathrm{min}}}^{-1}$ asymptotically
\begin{eqnarray}
\label{logr}  
  \mathbf{C}_{t_a,ij}(t;p_0)&\simeq&
  C_1\!+\frac{C_2y_1}{\ln\delta_1^{-1}}\ln\!\left(\frac{\tau+t_a}{\tau}\right)+\frac{y_1}{\ln\delta_2^{-1}}\ln\!\left(\frac{1}{\lambda_{k_{\mathrm{min}}}t_a}\right)\!\left[C_3+\ln\!\left(\frac{\tau+t_a}{\tau}\right)\right]
  \!+\!\mathcal{O}(\kappa)
\end{eqnarray}    
where
$\kappa=\max\left(\lambda_{k_{\mathrm{min}}}t,\lambda_{k_{\mathrm{min}}}\tau,\lambda_{k_{\mathrm{min}}}t_a\right)$. Moreover,
in te limit $\tau\gg t_a$ we find
\begin{equation}
 \label{fulls}  
\mathbf{C}_{t_a,ij}(t;p_0)\simeq C_1 + \frac{y_1}{\ln\delta_1^{-1}}\!\left(C_2+\frac{\ln\delta_1^{-1}}{\ln\delta_2^{-1}}\ln\!\left(\frac{1}{\lambda_{k_{\mathrm{min}}}t_a}\right)\right)\frac{t_a}{\tau}+\frac{C_3y_1}{\ln\delta_2^{-1}}\ln\!\left(\frac{1}{\lambda_{k_{\mathrm{min}}}t_a}\right)\!+\!\mathcal{O}(\kappa),
\end{equation}
such that when $\lambda_{k_{\mathrm{min}}}t_a=\mathcal{O}(1)$ we
recover the so-called 'full aging' scaling \cite{Amir,PNAS,Rodriguez_2003}
\begin{equation} 
  \mathbf{C}_{t_a,ij}(t;p_0)\simeq C_1 +
  C_2\frac{t_a}{\tau}
\label{full_c}  
\end{equation} 
\end{proposition}

\begin{proof}
The proof of the proposition is straightforward and follows from
noticing that in the limit $x\ll1$ we have
$\Gamma(0,x)=-\ln(x)+\gamma+\mathcal{O}(x)$. Plugging into the
expression in Proposition~\ref{similar} we find, upon elementary
manipulations, the result in Proposition~\ref{full}. The 'full aging'
scaling is further obtained by Taylor expanding the logarithm to first
order. 
\end{proof}  

\begin{remark}\label{exmpl}
The representation of $C_{t_a,i,j}(\tau;p_0)$ given in Theorem \ref{repr} is indeed frequently observed in
experiments on glassy systems \cite{LN}, while self-similar aging
dynamics with a power law scaling as in Proposition \ref{similar} has
been observed both in glassy systems \cite{Herisson,Ritort,LN} as well as in individual
protein molecules \cite{Jeremy} and, in a similar form, emerges in the case
of phenomenological so-called continuous time random walk models with
diverging waiting times \cite{EliIgor}. Notably, in the specific case
of $1/f$ self-similar dynamics  our analysis also recovers the
well-known, yet puzzling, ``full aging'' limiting scaling of
$C_{t_a,i,j}(\tau;p_0)$ (see \eql\ref{full_c}) 
\cite{Bouchaud_92,Amir,PNAS,Rodriguez_2003}, and in
particular the presence of the logarithmic correction in $t_a$ to the full
aging scaling in \eql\ref{fulls}) may potentially explain the observation that a
perfect $f(t_a/\tau)$ collapse is only observed for a specific ``good''
range values of the aging time $t_a$ \cite{Rodriguez_2003}.  Moreover,
the power-law aging in Eq.~(\ref{full_a}) for $t\gg t_a$ agrees with
the ``renewal aging'' in fractional
dynamics \cite{Ralf} (since $z^\alpha$ is the leading order term of
the expansion of the incomplete Beta function, $B(z,\alpha,1-\alpha)$,
as $z\to 0$). 
These specific results, as
wall as others, therefore emerge as special cases of the framework
presented in this work.
\end{remark}  

\begin{observation}{With respect to dynamical time asymmetry, the under-sampling of equilibrium is equivalent to a temperature quench.}\label{obs2}
Let $P^T_\eeq(\bx)$ be the equilibrium probability density function of
the full system at a temperature $T$ prior to a quench in temperature, 
that is different from the ambient temperature $T_0$, i.e. $T>T_0$. Expanding
$P^T_\eeq(\bx)$ in the eigenbasis of $\LL$ (at the ambient
temperature)  we find
$P^T_\eeq(\bx)=\sum_l|\psi_l^R\rangle\left\langle\psi_l^L|P^T_\eeq\right\rangle$,
where $\lim_{T\to
  T_0}\left\langle\psi_l^L|P^T_\eeq\right\rangle=\delta_{l0}$. Let us
further assume that the observable $\bq$ is fully sampled from
$P^T_\eeq(\bx)$ (like in the case of a supercooled liquid),
i.e. $\Xi_0=\Xi$, such that according to \eql\ref{limit}) we have
$\Psi_{kl}(\Xi)=\delta_{kl}$.
\begin{eqnarray*}
G_{P^T_\eeq}(\bq,t|\bq_0\in
\Xi)&=&\sum_{k}\ee^{-\lambda_kt}\Psi_{0k}(\bq)\left\langle\psi_k^L|P^T_\eeq\right\rangle\\
G_{P^T_\eeq}(\bq,t_a+\tau,\bq',t_a|\bq_0\in\Xi)&=&\sum_{k,l}\ee^{-\lambda_k\tau-\lambda_l
    t_a}\Psi_{0k}(\bq)\Psi_{kl}(\bq')\langle\psi_{l}^L|P^T_\eeq\rangle.
\end{eqnarray*}
since $\Psi_{kk}(\Xi)=\left.\lflat
P^T_\eeq\right\rangle=1$. According to Theorem \ref{invariant}  and
Corollary \ref{binvariant} a temperature quench gives rise
to broken time-translation invariance as longs as the projection renders the dynamics
non-Markovian.
Now consider an system prepared in equilibrium $p_0(\bx)=\peq$ but
with the
projected observable undersampled from said equilibrium, i.e. for a
domain $\bq_0\in\Xi_0\subset \Xi$ such that
$P_{\eeq}(\bq_0\in\Xi_0)\ne 1$. Then (see \eql\ref{GnonMarkSeq}) and \eql\ref{Greens3e}))
\begin{eqnarray*}
G_{\eeq}(\mathbf{q},t|\mathbf{q}_0\in \Xi_0)&=&\sum_{k}\ee^{-\lambda_kt}\Psi_{0k}(\bq)\frac{\Psi_{k0}(\Xi_0)}{\Psi_{00}(\Xi_0)}\\
 G_{\eeq}(\bq,t_a+\tau,\bq',t_a|\bq_0\in\Xi_0) &=& \sum_{k,l}\ee^{-\lambda_k\tau-\lambda_l
    t_a}\Psi_{0k}(\bq)\Psi_{kl}(\bq')\frac{\Psi_{l0}(\Xi_0)}{\Psi_{00}(\Xi_0)},
\end{eqnarray*}
which has a broken time-translation invariance as long as the projection renders the dynamics
non-Markovian according to Theorem \ref{invariant} and Corollary \ref{binvariant}.
Clearly, the only difference between the two non-Markovian time
evolutions is in the factor $\psi_{l}^L|P^T_\eeq\rangle$ versus
$\Psi_{l0}(\Xi_0)/\Psi_{00}(\Xi_0)$, which demonstrates that the
effect of
temperature quench and under-sampling of equilibrium are indeed
(qualitatively) virtually
indistinguishable as stated in the observation.
\end{observation}  

\setcounter{equation}{0}
\setcounter{figure}{0}
\renewcommand{\thefigure}{D\arabic{figure}}
\renewcommand{\theequation}{D\arabic{equation}}
\section{Physical models, experimental and simulation data}\label{systems}
\subsection{Fictitious dynamical time asymmetry in a time-translation invariant system: the Brownian particle in a box}\label{box}
Consider the propagator (i.e. the probability density) of a Browninan particle with diffusion
coefficient $D1$ confined in a
box of unit length $L$, $G(x,t|x_0,0)$ (without loss of generality we
express length in units of $L$ and time in units of $L^2/D$ such that
$x\to x/L$ and $t\to t D/L^2$) with $x\in [0,1]$, evolving according the Fokker-Planck equation
\begin{equation}
 \partial_t G(x,t|x_0,0)=\partial_x^2 G(x,t|x_0,0),\quad \partial_x G|_{x=0}=\partial_x G|_{x=1}=0,
\end{equation}
with initial condition $G(x,0|x_0)=\delta(x-x_0)$. The spectral
expansion of the Green's
function of the problem reads
\begin{equation}
 G(x,t|x_0,0)=\sum_{k=0}^\infty
 \psi_k(x)\psi_k(x_0)\mathrm{e}^{-\lambda_k t},\quad \psi_k(x)=\sqrt{2-\delta_{0,k}}\cos(k \pi x),\, \lambda_k=k^2\pi^2,
\end{equation}
where we note that the problem is self-adjoint and hence
$\psi_k^R(x)=\psi_k^L(x)=\psi_k(x)$. Let us define
\begin{eqnarray}
I_k&\equiv&
2^{-\delta_{k0}}+(1-\delta_{k0})\sqrt{2}(\cos(k\pi)-1)/k^2\pi^2\nonumber\\
J_{k,l}&\equiv&2^{-\delta_{k0}\delta_{l0}}
+I_k\delta_{l0}+I_l\delta_{k0}+\delta_{kl}/2k\pi +2(1-\delta_{kl})(k^2+l^2)[\cos(k\pi)\cos(l\pi)-1]/[(k^2-l^2)\pi]^2.
\end{eqnarray}
Then we have 
\begin{align}
 &\langle x(t) \rangle =\int_0^1 dx G(x,t|x_0,0)x=\sum_{k=0}^\infty
 I_k\psi_k(x_0)\mathrm{e}^{-\lambda_k t}\nonumber\\
 &\langle x(\tau+t_a)x(t_a) \rangle =\int_0^1 dx \int_0^1 dx_1G(x,\tau+t_a|x_1,t_a)G(x_1,t_a|x_0,0)xx_1=\sum_{k=0}^\infty I_k\mathrm{e}^{-\lambda_k (\tau+t_a)}\sum_{l=0}^\infty J_{k,l}\psi_l(x_0)\mathrm{e}^{-\lambda_l t_a},
\end{align}
which enter the definition of the aging correlation function
\begin{equation}
 C_{t_a,x_0}(\tau)=\frac{\langle x(\tau+t_a)x(t_a) \rangle-\langle x(\tau+t_a)\rangle\langle x(t_a) \rangle}{\langle x(t_a)x(t_a) \rangle-\langle x(t_a)\rangle^2}.
 \label{aging corr}
\end{equation}
When the initial distribution is not a point but is sampled from a
flat distribution between $a$ and $b$ (as in the example in the main
text), i.e. uniformly from a domain $\Omega_0=[a,b]$, we instead define
\begin{equation}
 L_k(b,a)= \delta_{k0} + (1-\delta_{k0})\sqrt{2}[\sin(k\pi b)-\sin(k\pi a)]/k\pi(b-a)
\end{equation}
and then
\begin{eqnarray}
  \langle x(t) \rangle_{\Omega_0}&=&\int_0^1 dx \int_{\Omega_0} dx_0 G(x,t|x_0,0)P_0(x_0)x=\sum_{k=0}^\infty I_kL_k(b,a)\mathrm{e}^{-\lambda_k t}\nonumber\\
\langle x(\tau+t_a)x(t_a) \rangle_{\Omega_0}&=&\!\!\int_0^1 dx \int_0^1 dx_1 \int_{\Omega_0} dx_0G(x,\tau+t_a|x_1,t_a)G(x_1,t_a|x_0,0)P_0(x_0)xx_1\nonumber\\&=& \sum_{k=0}^\infty I_k\mathrm{e}^{-\lambda_k (\tau+t_a)}\sum_{l=0}^\infty J_{k,l}L_l(b,a)\mathrm{e}^{-\lambda_l t_a}.
\label{interval}
\end{eqnarray}
Once inserted in Eq.~(\ref{aging corr}) Eq.~(\ref{interval}) deliver
the aging autocorrelation function shown in Fig.~2 in the main text
that displays \emph{fictitious dynamical time asymmetry} (i.e. trivial dependence on $t_a$). In the
meantime, the relaxation dynamics is time-translation invariant
according to Definition~\ref{inv} as a result of Theorem \ref{invariant} (see also
Eq.~(2) in the main text), since it is Markovian and thus satisfies
the Chapman-Kolmogorov semi-group property (see Lemma \ref{nonconc}). The fictitious dynamical time asymmetry is
thus a result of trivial non-stationarity in Definition \ref{trivial}. 

\subsection{Rouse polymer model}\label{Rouse}
 The Rouse polymer chain \cite{Fixman1,wilemski_diffusioncontrolled_1974} is a flexible macromolecule consisting of harmonic springs
 of zero rest-length. The potential energy of the macromolecule with
 $N+1$ point-like units (here referred to as 'beads') with a
 configuration $\bR\in\mathbb{R}^{3(N+1)}\equiv\{\bR_i\}$, where $\bR_i\in\mathbb{R}^3$ is given by
 $U(\bR)=\sum_{i=1}^{N}\frac{3}{\beta b^2}|\bR_{i+1}-\bR_i|^2$,
 where $b$ is the so-called Kuhn length describing the
 size of a chain segment (i.e. the characteristic distance between two
 beads) and will for convenience (and without any loss of generality)
 here be set to $b=\sqrt{3}$. The dynamics is assumed to evolve
 according to overdamped diffusion (with all beads having a equal
 diffusion coefficient $D$)  in a heat bath with zero mean
 Gaussian white noise, i.e. according to the system of coupled It$\mathrm{\hat{o}}$ equations
 \begin{equation}
 \label{Rouse}  
 d\bR_t=-\beta D\underline{\mathbf{M}}\bR_t dt
 +\sqrt{2D}d\hat{\mathbf{W}}_t,
 \end{equation}
 where $\hat{\mathbf{W}}_t$ denotes for a $3(N+1)$-dimensional vector of independent
Wiener processes whose increments have a Gaussian distribution with
zero mean and variance $dt$: $\mathbb{E}[d\hat{W}_{t,i}d\hat{W}_{t',j}]=\delta_{ij}\delta(t-t')dt$.
 $\mathbb{E}[\cdot]$ denotes the expectation over the ensemble of
Wiener increments. The interaction matrix $\underline{\mathbf{M}}$ is
the $3(N+1)\times 3(N+1)$ tridiagonal Rouse super-matrix whose elements are
$\mathbf{M}_{ij} \mathbbm{1}$ (where $\mathbbm{1}$ denotes the
$3\times 3$ unit matrix) and the $(N+1)\times (N+1)$ matrix $\mathbf{M}$ has elements 
$\mathbf{M}_{ii}=(2-1^{\delta_{i1}}-1^{\delta_{iN}})$ and $\mathbf{M}_{ii+1}=\mathbf{M}_{ii-1}=-1$. On  the level of a probability
density function the It$\mathrm{\hat{o}}$ process \eql\ref{Rouse})
corresponds to the $N$-body Fokker-Planck equation, which, introducing
the operator $\bnabla\equiv\{\nabla_i\}$ reads
\begin{equation}
\partial_t P(\bR,t) = D\left(\bnabla^T\bnabla + \beta \bnabla^T\underline{\mathbf{M}}\bR\right)P(\bR,t), 
\label{RouseFPE}  
\end{equation}
which has the structure of \eql\ref{FPE}) and can be decoupled as follows. We first rotate the coordinate
system to normal
coordinates $\bql\in\mathbb{R}^{3(N+1)}=\{\bql_i\},\forall
i\in[0,N]$ \footnote{Note that $\protect \mathbf {Q}_0$ corresponds to the center
  of mass, which does not affect the dynamics of internal
  coordinates.},
i.e. $\bR_i=\bss\bql_i$ ($\bql_i\in\mathbb{R}^3$) and
$\nabla_i=\bss\nabla_{\bql_i}$ with the $(N+1)\times (N+1)$ orthogonal matrix $\bss$,
$\bss^{-1}=\bss^T$, which diagonalizes the Rouse matrix,
$\blam=\bss^{T}\mathbf{M}\bss$, where
\begin{equation}
\label{normalm}
\blam_{ik}=4\sin^2\left(\frac{k\pi}{2(N+1)}\right)\delta_{ik}\equiv\lambda_k\delta_{ik},\quad
\bss_{ik}=\sqrt{\frac{2}{N+1}}\cos\left(\frac{(2i-1)k\pi}{2(N+1)}\right),\forall
k>0,
\end{equation}
and $\bss_{i0}=(N+1)^{-1/2},\forall i$ (which is not required; see footnote).
Introducing the $3(N+1)\times 3(N+1)$ super-matrix $\bssu$, whose elements are
$\bss_{ik}\mathbbm{1}$, the transformation to normal coordinates is
found to decouple the Fokker-Planck
equation \eql\ref{RouseFPE}):
\begin{eqnarray}
  \partial_t P(\bql,t) &=&
  D\left(\bnabla^T\bssu^T\bssu\bnabla +\beta
    \bnabla^T\bssu^T\underline{\mathbf{M}}\bssu\bql \right)P(\bql,t)\nonumber\\
  &=&D\sum_{i=1}^N\left[\partial^2_{\bql_i} + \beta \lambda_i\partial_{\bql_i}\bql_i\right]P(\bql,t),
\label{RouseFPEd}  
\end{eqnarray}
whose structure implies that the solution factorizes
$P(\bql,t)=\prod_{i=1}^NP(\bql_i,t)$ and $P(\bql_i,t)$ is simply the
well-known solution of a 3-dimensional Ornstein-Uhlenbeck process. In
particular, the density of the invariant measure and Green's function
read
\begin{eqnarray}
  \label{OUP}
  P_\eeq(\bql)&=&\prod_{i=1}^N\left(\frac{\lambda_i}{2\pi}\right)^{3/2}\ee^{-\lambda_i\bql_i^2/2}\\
  Q(\bql,t|\bql',t')&=&\prod_{i=1}^N\left(\frac{\lambda_i}{2\pi(1-\ee^{-2\lambda_i(t-t')})}\right)^{3/2}\exp\left(-\frac{\lambda_i(\bql_i-\bql_i'\ee^{-\lambda_i(t-t')})^2}{2(1-\ee^{-2\lambda_i(t-t')})}\right).
\end{eqnarray}
The Gaussian structure of the solution will permit explicit results
not requiring a spectral decomposition of the Fokker-Planck operator
(which, however, is well-known \cite{Fixman1}).

We are here interested in the dynamics of the end-to-end distance of
the polymer, $q(t)\equiv|\bq(t)|=|\bR_{N+1}(t)-\bR_1(t)|$, which would be typically
probed in a single-molecule FRET or optical tweezers experiment. To
make minimal assumptions we
assume a stationary initial preparation of the full system, i.e. $P_0(\bR)=P_\eeq(\bR)$. Since
$q(t)$ at any instance depends on all other degrees of freedom
$\bR_k(t),\forall k\in[2,N]$ its dynamics is strongly
non-Markovian. In normal coordinates this corresponds to
\begin{equation}
\label{coordinate}
q(t)=\sum_{i=1}^N|\mathcal{A}_i\bql_i|=\sum_{i=1}^N\sqrt{\frac{2}{N+1}}\left|\left[\cos\left(\frac{(2N-1)k\pi}{2(N+1)}\right)-\cos\left(\frac{k\pi}{2(N+1)}\right)\right]\bql_i\right|,  
\end{equation}
having defined $\mathcal{A}_i$ in \eql\ref{coordinate}), such that introducing $d\bql\equiv\prod_{i=1}^Nd\bql_i$ the projection
operator \eql\ref{projection}) can be shown to correspond to
\begin{equation}
 \mcp_{\bql}(\BGamma;q) = q^2\int_0^{2\pi}\!\!d\varphi\int_0^\pi \!\! d\theta \int_{\Omega}d\bql\delta\left(\sum_{i=1}^N\mathcal{A}_i\bql_i-\bq(q,\varphi,\theta)\right),
  \label{projR}
\end{equation}
and the non-Markovian conditional two-point probability density is
calculated according to \eql\ref{projected}) and leads, upon a lengthy
but straightforward computation via a Fourier transform
$\mathrm{FT}\{\bq\to\mathbf{v}\}$, i.e.
$\mcp_{\bql}(\BGamma;q)f(\bql)\xrightarrow{\mathrm{FT}}\tilde{f}(\mathbf{v})\xrightarrow{\mathrm{FT}^{-1}}f(\bq)\xrightarrow{\int\!\!
d\varphi\!\int\!\! d\theta}f(q)$,
to
\begin{equation}
G_{\eeq}(q,t|q_0)=\frac{1}{\sqrt{\pi\gamma}}\frac{\phi(t)^{-1}}{\sqrt{1-\phi(t)^2}}\frac{q}{q_0}\exp\left(-\frac{q^2+q_0^2}{4\gamma(1-\phi(t)^2)}+\frac{q_0^2}{4\gamma}\right)\sinh\left(\frac{qq_0\phi(t) }{2\gamma(1-\phi(t)^2)}\right),
\label{Rouse2}  
\end{equation}
where we have defined 
\begin{equation}
\phi(t)=\sum_{i=1}^N\frac{\MA_i^2}{2\lambda_i}\ee^{-\lambda_it},\quad \gamma=\phi(0).
\end{equation}
The (non-aging) autocorrelation function in \eql\ref{Ncorr}),
$C(t)=\langle q(t)q(0)\rangle - \langle q(t)\rangle\langle q(0)\rangle$,
can in turn be shown to be given by
\begin{equation}
C(t)=\frac{4\gamma}{\pi}\left[3+\left(\frac{1}{\phi(t)}+2\right)\arctan\left(\frac{\phi(t)}{\sqrt{1-\phi(t)^2}}\right)\right]
- \frac{16\gamma}{\pi},
\label{RouseCorr}  
\end{equation}
which decays to zero as $t\to\infty$. We now address the three-point
conditional probability density \eql\ref{Greens3e}) and aging autocorrelation
function \eql\ref{a_corr}), which are much more challenging. As such a complex
calculation has, to the best of our knowledge, not been performed
before for any stochastic system, we here present a more detailed derivation.\\ 
\indent We start with \eql\ref{Mthreej}), plug in
Eqs.~(\ref{OUP}) and use \eql\ref{projR}) to first calculate the three-point
joint density of the vectorial counterpart, i.e. $P_{\eeq}(\bq,t_a+\tau,\bq',t_a,\bq_0)$. We now perform a triple Fourier transform
\begin{equation}
  \label{3Fourier}
\tilde{P}_{\eeq}(\bu,t_a+\tau,\bv,t_a,\bw)\equiv\!\int \!\!   d\bq \ee^{-i\mathbf{u}\cdot \bq}\!\int \!\!  d\bq' \ee^{-i\mathbf{v}\cdot \bq'}\!\int \!\!  d\bq_0 \ee^{-i\mathbf{w}\cdot \bq_0}P_{\eeq}(\bq,t_a+\tau,\bq',t_a,\bq_0)
\end{equation}  
and carry out all integrations over $\bql,\bql'$ and $\bql_0$ (a total
of $3N$ integrals each) and introduce the short-hand notation
$S_{t}=\phi(t)$ to find
\begin{equation}
\label{FT3}
\tilde{P}_{\eeq}(\bu,t_a+\tau,\bv,t_a,\bw)=\frac{1}{(2\pi)^9}\exp\left(-\gamma(\bw^2-\bv^2-\bu^2)-2S_{t_a}\bw^T\bv-2S_t\bu^T\bw-2S_\tau\bu^T\bv\right).
\end{equation}
We now invert back all three Fourier transforms and introduce
auxiliary functions $\mX_{\tau,t_a}\equiv
\gamma^3-\gamma(S_{t_a}^2-S_t^2-S^2_\tau)+2S_{t_a}S_tS_\tau$ as well
as $\mY_{\tau,t_a}\equiv \gamma S_t-S_{t_a}S_\tau$ and $\mZ_{\tau,t_a}\equiv
\gamma S_\tau-S_{t}S_{t_a}$ (keeping in mind that $t=\tau+t_a$) to find
\begin{eqnarray}
\label{3vec}  
\!\!\!\!\!P_{\eeq}(\bq,t,\bq',t_a,\bq_0)&=&(4^3\pi^{3}\mX_{\tau,t_a})^{-3/2}\nonumber\\
&\times&\exp\left(-\frac{(\gamma\mX_{\tau,t_a}+\mY_{\tau,t_a})\bq_0^2+(\gamma\mX_{\tau,t_a}+\mZ_{\tau,t_a})\bq'^2+(\gamma^2-S_{t_a}^2)^2\bq^2}{4(\gamma^2-S_{t_a}^2)\mX_{\tau,t_a}}\right)\nonumber\\
&\times&\exp\left(-\frac{(S_{t_a}\mX_{\tau,t_a}+\mY_{\tau,t_a}\mZ_{\tau,t_a})}{2(\gamma^2-S_{t_a}^2)\mX_{\tau,t_a}}\bq'^T\bq_0+\frac{\mY_{\tau,t_a}}{2\mX_{\tau,t_a}}\bq^T\bq_0+\frac{\mZ_{\tau,t_a}}{2\mX_{\tau,t_a}}\bq^T\bq'\right)\!.
\end{eqnarray}
Before we perform the angular integrations, $(qq'q_0)^2\int_0^{2\pi}\!\!d\varphi\int_0^\pi \!\! d\theta
\int_0^{2\pi}\!\!d\varphi'\int_0^\pi \!\! d\theta'
\int_0^{2\pi}\!\!d\varphi_0\int_0^\pi \!\! d\theta_0$, we introduce the
final set of auxiliary functions (i.e. the third in the hierarchy of our
notation):
\begin{eqnarray}
&&\Lambda_{1}^{\tau,t_a}\equiv\frac{S_{t_a}\mX_{\tau,t_a}-\mY_{\tau,t_a}\mZ_{\tau,t_a}}{2(\gamma^2-S_{t_a}^2)\mX_{\tau,t_a}},\quad 
 \Lambda_{2}^{\tau,t_a}\equiv\frac{\mY_{\tau,t_a}}{2\mX_{\tau,t_a}},\quad
 \Lambda_{3}^{\tau,t_a}\equiv\frac{\mZ_{\tau,t_a}}{2\mX_{\tau,t_a}},
 \nonumber\\
 &&\Omega_{\boldsymbol{\Lambda}_{\tau,t_a}}{q\,\,q'\,q_0 \choose a\,\,b\,\,c }\equiv\mathrm{erfi}\left(\frac{ a\Lambda_{1}^{\tau,t_a}
    \Lambda_{3}^{\tau,t_a}q'+b \Lambda_{2}^{\tau,t_a}|
    \Lambda_{1}^{\tau,t_a}q_0+c \Lambda_{3}^{\tau,t_a}q|}{\sqrt{2
      \Lambda_{1}^{\tau,t_a} \Lambda_{2}^{\tau,t_a}
      \Lambda_{3}^{\tau,t_a}}}\right),
  \label{aux3}
\end{eqnarray}
which, after a long and laborious computation leads to the exact result
\begin{eqnarray}
\label{Rouse3}
&&P_\eeq(q,\tau+t_a,q',t_a,q_0)= \displaystyle{\frac{qq'q_0}{16\pi}\left(\frac{\gamma^2-S_{t_a}^2}{[S_{t_a}\mX_{\tau,t_a}-\mY_{\tau,t_a}\mZ_{\tau,t_a}]\mY_{\tau,t_a}\mZ_{\tau,t_a}}\right)^{1/2}}\nonumber\\
&&\times\exp\left(-\frac{S_{t_a}(\gamma^2-S_{t_a}^2)q^2}{4(S_{t_a}\mX_{\tau,t_a}-\mY_{\tau,t_a}\mZ_{\tau,t_a})}-\frac{(\gamma+S_{t_a}\frac{\mZ_{\tau,t_a}}{\mY_{\tau,t_a}})q'^2+(\gamma+S_{t_a}\frac{\mY_{\tau,t_a}}{\mZ_{\tau,t_a}})q_0^2}{4(\gamma^2-S_{t_a}^2)}\right)\nonumber\\
&&\times\Bigg\{\displaystyle{\Omega_{\boldsymbol{\Lambda}_{\tau,t_a}}{q\,\,q'\,q_0
    \choose - + - }-\Omega_{\boldsymbol{\Lambda}_{\tau,t_a}}{q\,\,q'\,q_0
    \choose + + - }+\Omega_{\boldsymbol{\Lambda}_{\tau,t_a}}{q\,\,q'\,q_0
    \choose + - + }+\Omega_{\boldsymbol{\Lambda}_{\tau,t_a}}{q\,\,q'\,q_0
    \choose + + + }}\Bigg\}.
\end{eqnarray}
The conditional three-point density is in turn obtained from
\eql\ref{Rouse3}) by
\begin{equation}
\label{GRouse3}  
G_{\eeq}(q,\tau+t_a,q',t_a|q_0)=P_\eeq(q,\tau+t_a,q',t_a,q_0)/P_\eeq(q_0).
\end{equation}
Having obtained all quantities required for the computation of the
aging correlation function and the
time asymmetry index $\AI$, the remaining integrals
\begin{equation}
C_{t_a}(\tau)=\int_0^{\infty}\!\!\! dq\int_0^{\infty}\!\!\! dq'qq'G_{\eeq}(q,\tau+t_a,q',t_a|q_0)- \!\int_0^{\infty}\!\!\! dq qG_{\eeq}(q,\tau+t_a|q_0)\int_0^{\infty}\!\!\! dq qG_{\eeq}(q,t_a|q_0)
\label{AgCorr}
\end{equation}  
as well as
\begin{equation}
\AI_{q_0}(t_a,\tau)=\int_0^{\infty} dq\int_0^{\infty} dq'G_{\eeq}(q,\tau+t_a,q',t_a|q_0)\log\left(\frac{G_{\eeq}(q,\tau+t_a,q',t_a|q_0)}{G_{\eeq}(q,\tau|q')G_{\eeq}(q',t_a|q_0)}\right)  
\label{AgRouse}
\end{equation}  
are performed using an adaptive Gauss-Kronrod routine
 \cite{boost_quadrature}. The results for a Rouse chain with $N=50$
 beads are presented in
 Fig.~\ref{rousechain_fig} and the corresponding time asymmetry index
 $\AI_{q_0}(t_a,\tau)$ in Fig.~3a in the manuscript. 
 \begin{figure}[ht!!]
 \begin{center}
  \includegraphics[width=0.8\textwidth]{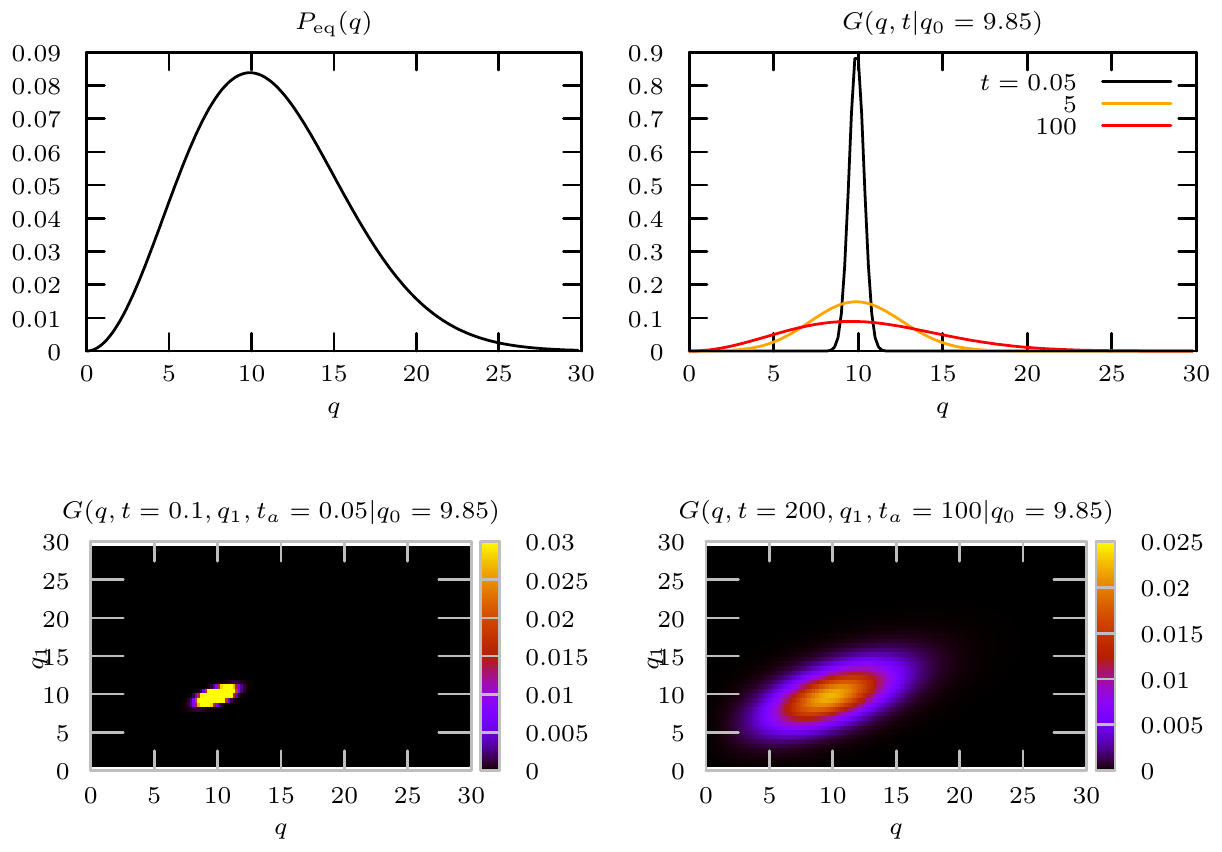}
 \caption{The top left panel shows the density of the invariant
   measure $P_\eeq(q)$, while the top right panel depicts the
   conditional two-point density $G_{\eeq}(q,t|q_0)$ in
   \eql\ref{Rouse2}) for $q_0=9.85$ and three different $t$. The
   bottom panels show the conditional three-point probability density
   function $G_\eeq(q,\tau+t_a,q',t_a|q_0)$ in \eql\ref{GRouse3}) for
   two combinations of $t=\tau+t_a$ and $t_a$.}
 \label{rousechain_fig}
 \end{center}
 \end{figure}

 The density of the invariant measure of the end-to-end distance,
$P_\eeq(q)$ (Fig.~\ref{rousechain_fig}, top left) is concentrated in the regime $0<q<30$ with a maximum at
$q_{\mathrm{peak}}=9.85 \approx 10$. The evolution of the conditional two-point
conditional probability density for an ensemble of trajectories
starting at the typical distance $q_{\mathrm{peak}}$,
$G_{\eeq}(q,t|q_{\mathrm{peak}})$ (Fig.~\ref{rousechain_fig}, top right) evolves smoothly towards $P_\eeq(q)$
with a relaxation time $t_{\mathrm{rel}}=\lambda_1^{-1}\approx
253.4$. Notably, the corresponding three-point density
$G_{\eeq}(q,t,q',t_a|q_{\mathrm{peak}})$ (Fig.~\ref{rousechain_fig}, bottom) shows strong long-time
correlations in the evolution of $q(t)$, e.g. even for aging times
$t_a=100$ (which are already of the order of, but still smaller than,
$t_{\mathrm{rel}}$) the value of $q$ at time $t=2t_a$ is strongly
correlated with its value $q'$ at $t_a$ (Fig.~\ref{rousechain_fig},
bottom right). This long-lasting correlations, which are the result of
the projection of the full $3(N+1)$-dimensional dynamics of the
polymer onto a single distance coordinate $q(t)$, are responsible for
the dynamical time asymmetry.

Note that the relaxation time scales quadratically with the length of
the chain, in.e. $t_{\mathrm{rel}}\propto N^2$ and therefore the
dynamical time asymmetry extends, for long polymers $N\gtrsim 10^4$
over many orders in time. However, for such long chains the
computation or $\AI$ at short $t_a,\tau$ becomes numerically
unstable. In Fig.~\ref{size} we demonstrate the quadratic growth of
relaxation time-scales displaying dynamical time asymmetry with
increasing $N$.  
 \begin{figure}[ht!!]
 \begin{center}
  \includegraphics[width=0.98\textwidth]{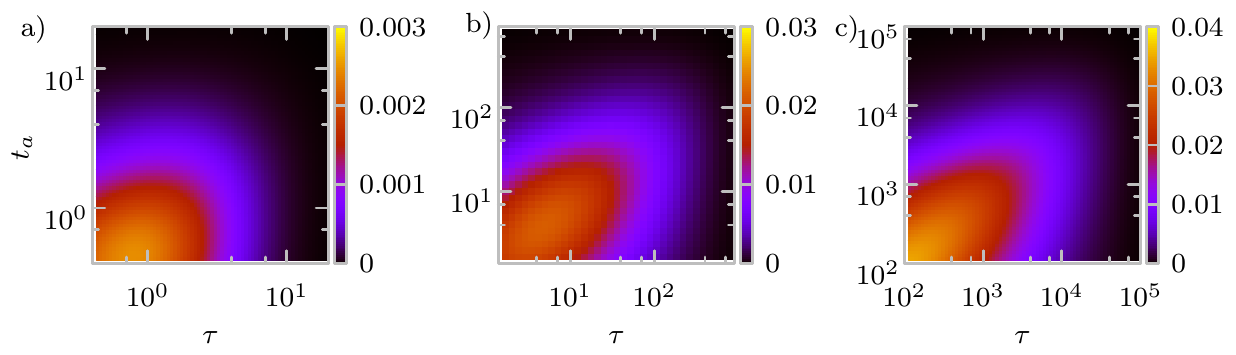}
 \caption{$\AI(t_a,\tau)$ for a) $N=10$, b) $N=10^2$, and c) $N=10^3$
   portraying a growing time-scale of dynamical time asymmetry.}
 \label{size}
 \end{center}
 \end{figure}

\subsection{Single file diffusion}\label{sfile}
The single file model refers to the overdamped Brownian motion of a
system of $N$ particles with hard core exclusion interactions, which
for simplicity (and because the finite-size scenario is obtained by a
simple re-scaling of space) we assume to be point-like and confined to
an interval of unit length $L=1$
\cite{lizana_single-file_2008,lapolla_unfolding_2018,Lapolla_CPC}. We express
length in units of $L$ and time
in units of $\tau=D/L^2$, where $D$ corresponds to the diffusion
coefficient which is assumed to be equal for all particles. The state of the
system is completely described with the vector of particle positions $\bx=(x_1,\ldots,x_N)$.
We are
interested in tagged-particle dynamics and therefore our projected observable
corresponds to the position of the $i$-th particle, $q(t)=x_i(t)$. The
full system's dynamics is driven solely by entropic driving forces
because the potential energy is strictly zero. In turn, the free
energy landscape (i.e. the potential of the mean force acting on the
tagged particle) corresponds to the entropic landscape, whereas the
potential energy hypersurface is perfectly flat.

The Fokker-Planck equation for the Green's function with initial condition $Q(\bx,t=0|\bx_0)$ describing the dynamics of $N$ reads
\begin{equation}
(\partial_t- \sum_i\partial^2_{x_i})Q(\bx,t|\bx_0)=0,\quad Q(\bx,t=0|\bx_0)=\prod_{i=1}^N\delta(x_{i,0}-x_i),
\end{equation}
which is solved under $N-1$ non-crossing boundary conditions
\begin{equation}
\lim_{x_{i+1}\to
  x_i}(\partial_{x_{0,i+1}}-\partial_{x_{0,i}})Q(\mathbf{x},t|\mathbf{x}_0)=0,
\quad \forall i.
\end{equation}
The system is exactly solvable with the coordinate Bethe ansatz, which
yields explicit results for the spectral expansion of
$Q(\mathbf{x},t|\mathbf{x}_0)$ \cite{lapolla_unfolding_2018,Lapolla_CPC},
i.e. $Q(\mathbf{x},t|\mathbf{x}_0)=\sum_{\mathbf{k}}\psi_{\mathbf{k}}^R(\bx)\psi_{\mathbf{k}}^L(\bx_0)\ee^{-\lambda_{\mathbf{k}}t}$
according to \eql\ref{SGreen}), where we introduced the $N$-tuple $\mathbf{k}=(k_1,\ldots,k_N),
k_i\in\mathbb{N},\forall i$. Expressions for the eigenfunctions
$\psi_{\mathbf{k}}^R(\bx),\psi_{\mathbf{k}}^L(\bx)$ are given in
\cite{lapolla_unfolding_2018,Lapolla_CPC} and the eignevalues corresponding to
$\lambda_{\mathbf{k}}=\sum_{i=1}^N k_i \pi^2$. Note that the
relaxation time $t_{\rm rel}=1/\lambda_\mathbf{1}$ once re-scaled to
natural units in terms of the collision time (i.e. $t_{\rm
  col}=(L/N)^2/D$ scales as $\propto N^2$. 

The projection operator is turn defined by \eql\ref{projection}) with $\delta(x_i-q)$, which according
to \eql\ref{element}) yields, upon some tedious algebra,
$G_\eeq(q,t|q_0\in\Xi_0)$ in \eql\ref{GnonMarkSeq})  and $G_\eeq(q,t,q',t_a|q_0\in\Xi_0)$ in \eql\ref{Greens3e}) with matrix elements
\begin{equation}
\Psi_{\mathbf{kl}}(x)=\frac{\mathbf{m_{l}}!}{N_L!N_R!}\sum_{\{n_i\}}T_j(x)\prod_{i=1}^{j-1}
L_i(x)\prod_{i=j+1}^{N} R_i(x)
\label{sfel}
\end{equation}
where $\mathbf{m_{l}}!=\prod_im_{k_i}$ is the multiplicity of the
eigenstate with $m_{k_i}$ corresponds to the number of times a
particular value of $k_1$ appears in the tuple and $N_L,N_R$ are the
number of particles to the left and right from the tagged particle,
respectively. The sum $\sum_{\{n_i\}}$ is over all permutations of the
elements of the $N$-tuple $\mathbf{k}$. For the equilibrium density we
find $P_\eeq(q)=\frac{N!}{N_L!N_R!}q^{N_L}(1-q)^{N_R}$. In \eql\ref{sfel}) we
have defined the auxiliary functions
\begin{eqnarray}
 T_j(x)&=&\begin{cases}
          1 \quad \lambda_k=\lambda_l=0\\
          \displaystyle{\sqrt{2}\cos(\lambda_{k/l}\pi x)} \quad \lambda_{k}=0 \, \text{or} \, \lambda_l=0\\
          \displaystyle{2\cos(\lambda_{k}\pi x)\cos(\lambda_{l}\pi x)} \quad \text{otherwise}\\
 \end{cases}\\
\end{eqnarray}
\begin{eqnarray}
L_j(x)&=&\begin{cases}
          x \quad \lambda_k=\lambda_l=0\\
         \displaystyle{\sqrt{2}\frac{\sin(\lambda_{k/l}\pi x)}{\lambda_{k/l}\pi} \quad \lambda_{k}=0} \, \text{or} \, \lambda_l=0\\
          \displaystyle{x+\frac{\sin(2\lambda_k\pi x)}{2\lambda_k\pi x}} \quad \lambda_k=\lambda_l\\
          \displaystyle{2\frac{\lambda_k\cos(\lambda_l\pi x)\sin(\lambda_k \pi x)-\lambda_l\cos(\lambda_k\pi x)\sin(\lambda_l \pi x)}{\pi(\lambda_k^2-\lambda_l^2)}} \quad \text{otherwise}\\
         \end{cases}\\
  R_j(x)&=&\begin{cases}
          \displaystyle{1-x} \quad \lambda_k=\lambda_l=0\\
          \displaystyle{-\sqrt{2}\frac{\sin(\lambda_{k/l}\pi x)}{\lambda_{k/l}\pi}} \quad \lambda_{k}=0 \, \text{or} \, \lambda_l=0\\
          \displaystyle{1-x-\frac{\sin(2\lambda_k\pi x)}{2\lambda_k\pi x}} \quad \lambda_k=\lambda_l\\
          \displaystyle{2\frac{-\lambda_k\cos(\lambda_l\pi x)\sin(\lambda_k \pi x)+\lambda_l\cos(\lambda_k\pi x)\sin(\lambda_l \pi x)}{\pi(\lambda_k^2-\lambda_l^2)}} \quad \text{otherwise}\\
         \end{cases}
\end{eqnarray}
The autocorrelation functions $C(t)$ and $C_{t_a}(\tau)$ can now be
calculated using Eqs.~(\ref{EQCorr}) and (\ref{a_corr}),
respectively, where trivially $\langle q(t)\rangle=(N_l+1)/(N+1)$ and $\langle q(0)^2\rangle=(N_L+2)(N_L+1)/[(N+2)(N+1)]$.  As it was impossible to carry out this final step
analytically, we carried out the integrals in Eqs.~(\ref{EQCorr}) and
(\ref{a_corr}) depicted in Fig.~3b in the main text  numerically according to  the trapezoidal rule. 
 \begin{figure}[ht!!]
 \begin{center}
  \includegraphics[width=0.8\textwidth]{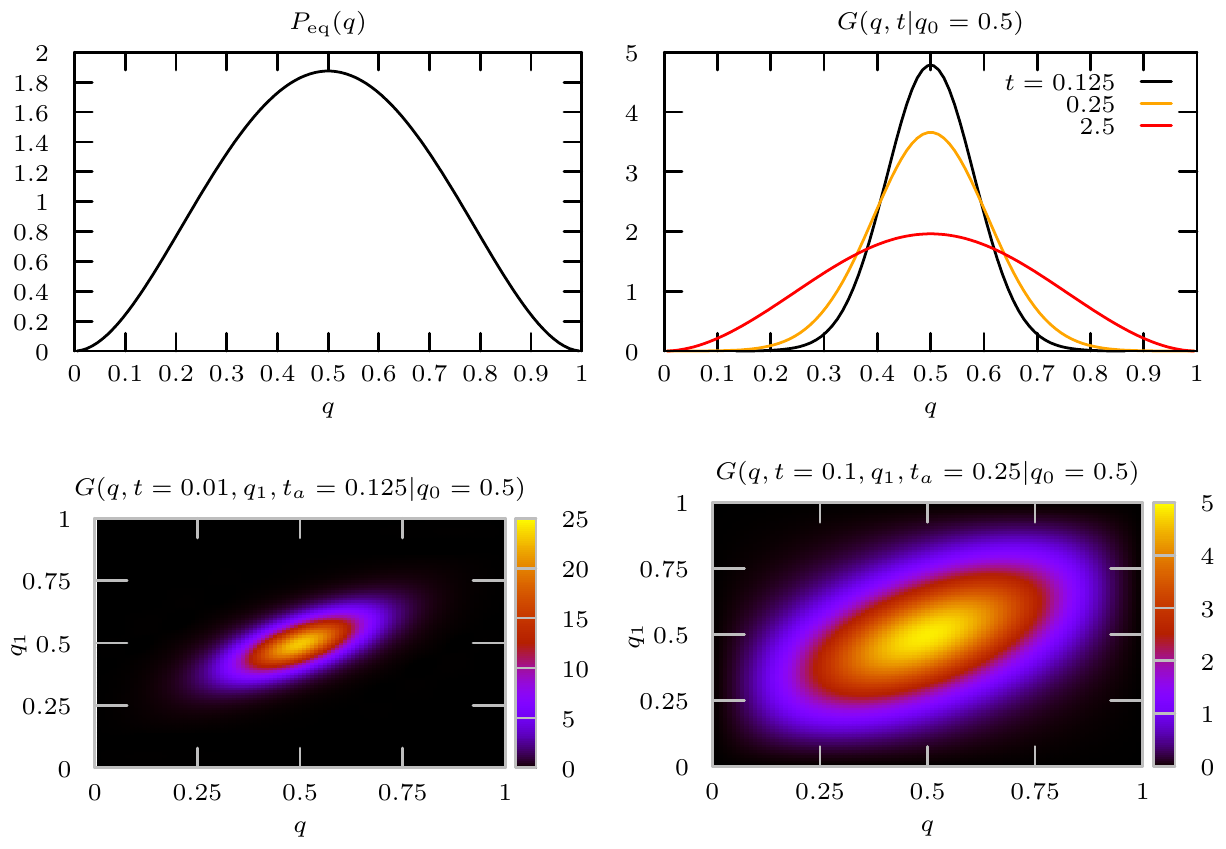}
  \caption{\textbf{Single file.} The top left panel depicts the density of the
    equilibrium measure $P_\eeq(q)$ of the third particle in a single
    file of five particles, $q(t)=x_3(t)$. The top right panel
    shows a two-point conditional probability density function
    $G(q,t|q_0)$ for different values of $t$, where
    $q_0=0.5$ (i.e. the maximum of the $P_\eeq(q)$). The bottom panels depict the
    three-point density $G(q,t,q',t_a|q_0)$ at
    different $\tau$ and $t_a$ evolving from the same initial
    condition. To produce $G(q,t,q',t_a|q_0)$ the spectral expansion \eql\ref{SGreen}) was truncated at maximum Bethe eigenvalue $225\pi^2$.}
  \label{single_fig}
  \end{center}
 \end{figure}\\
The computation of the time asymmetry index $\AI$ in \eql\ref{AI}), which was as well performed using the trapezoidal rule on a grid of
100 points, is extremely challenging even for moderate values of $N$. By repeating the integration using a smaller grid of 50
points we double-checked that the integration routine converged to a
sufficient degree. The results for a single file of $N=5$ particles
tagging the third particle are presented in Fig.~3b in the main
text and in Fig.~\ref{single_fig}.

The density of the invariant measure of the tagged central particle
(Fig.~\ref{single_fig}, top left) peaks in the center of the unit
box, $q_{\mathrm{peak}}=0.5$, and decays towards the borders due to the entropic repulsion with
the neighbors. The evolution of the conditional two-point
conditional probability density for an ensemble of trajectories
starting at the typical distance $q_{\mathrm{peak}}$,
$G_{\eeq}(q,t|q_{\mathrm{peak}})$ (Fig.~\ref{single_fig}, top right) evolves smoothly towards $P_\eeq(q)$
with a relaxation time $t_{\mathrm{rel}}=\lambda_1^{-1}\approx 2.5$.

Similar to the end-to-end distance of the Rouse polymer the
corresponding three-point density of the tagged particle,
$G_{\eeq}(q,t,q',t_a|q_{\mathrm{peak}})$ (Fig.~\ref{single_fig}, bottom) shows strong long-time
correlations in the evolution of $q(t)$, e.g. even for aging times
$t_a=0.125$ (which are already of the order of, but still smaller than,
$t_{\mathrm{rel}}$) the value of $q$ at time $t=2t_a$ is strongly
correlated with its value $q'$ at $t_a$ (Fig.~\ref{rousechain_fig},
bottom right). This long-lasting correlation reflects prolonged entropic
bottlenecks (i.e. 'traffic-jams'), which require collective
rearrangements of many particles and thus decorrelate slowly, giving
rise to strong memory effects and dynamical time asymmetry (see Fig.~3b in the main text) as soon as $q_0$ is initially not
sampled from $P_\eeq(q_0)$.

 \subsection{Analysis of experimental and simulation data}\label{experimental}
We now consider the time series of a
low-dimensional projected observable sampled at discrete time
steps that is frequently encountered in the analysis of experimental data. Therefore we first translate all definitions in
Sec.~\ref{lower} to discrete time series and explain in detail how to
carry out the complete analysis of aging dynamics for such systems. To
ease the application of these new concepts we also provide a C++
routine \textsf{TSymmetryFinder} that will be made available on GitHub.

 We consider a discrete time series of length $N$ --  in our case a
 one-dimensional physical observable $q(t_i)$ -- sampled at constant time
 intervals with spacing $t_{i+1}-t_i=\Delta t,\forall i$. We pre-process the data by
 evaluating the mean value over the time series
 $\overline{q}=N^{-1}\sum_{i=1}^Nq(t_i)$ and then center the data by
 subtracting the mean value, $q(t_i)\to q(t_i)-\overline{q}$. 
In the first stage we determine the (non-aging) autocorrelation
 function
 \begin{equation}
 C(\tau=n_\tau \Delta t)=\frac{1}{(N-n_\tau ) \Delta
   t}\sum_{i=1}^{N-n_\tau}q(t_{i+n_\tau})q(t_i)\Delta t,\qquad
 n_\tau\ll N
 \end{equation}
and determine the relaxation time as
$t_r:\displaystyle{\min_{t_i}}[\mathcal{C}(t_i)/\mathcal{C}(0)<\epsilon]$,
where we choose $\epsilon=0.05$. All data are
henceforth analyzed such that $n_{\mathrm{max}}\equiv
n_{\tau_{max}}\approx  n_r=t_r/(\Delta t)$ in order to assure sufficient
sampling when evaluating sliding averages.  

Next we determine the equilibrium probability density function  and two-point conditional probability density (Eq.~(\ref{projected})) as a
histogram taken over the data. We introduce bins $\mathcal{B}_i$
centered at $q_i$ with a width $\delta q$ and define
the characteristic function of a bin
$\mathbbm{1}_{\mathcal{B}_i}[q(t_i)]=1$ if $q_i -\delta q/2\le
  q(t_i)<q_i+\delta q/2$ and zero otherwise and let
  $\mathbbm{1}_{\Xi_0}[q]$ be the indicator function of the initial condition. Let us further define
$n_l(n_\tau)=\displaystyle{\max_i q(t_i)\in\Xi_0}:n_l+n_\tau\le
  N$. The equilibrium probability density function is then determined
  as
\begin{equation}
P_{\eeq}(q_i)=(N\delta q)^{-1}\sum_{i=1}^N\mathbbm{1}_{\mathcal{B}_i}[q(t_i)]
\end{equation}  
and the 2-point conditional probability density as
\begin{eqnarray}
  G(q_i,n_\tau|q_0\in\Xi_0)&=&\frac{\sum_{i=1}^{n_l(n_\tau)}\mathbbm{1}_{\mathcal{B}_i}[q(t_{i+n_\tau})]\mathbbm{1}_{\Xi_0}[q(t_i)]}{\delta
    q\sum_{i=1}^{n_l(n_\tau)}\mathbbm{1}_{\Xi_0}[q(t_i)]},\\
  G(q_i,n_\tau|q_j)&=&\frac{\sum_{i=1}^{n_l(n_\tau)}\mathbbm{1}_{\mathcal{B}_i}[q(t_{i+n_\tau})]\mathbbm{1}_{\mathcal{B}_j}[q(t_i)]}{\delta
    q\sum_{i=1}^{n_l(n_\tau)}\mathbbm{1}_{\Xi_0}[q(t_i)]},
\end{eqnarray}
The three-point conditional density \eql\ref{Greens3e}) is defined for
$n_{\tau'}\le n_{\tau}$ analogously as
\begin{equation}
G(q_i,n_\tau,q_j,n_{\tau'}|q_0\in\Xi_0)=\frac{\sum_{i=1}^{n_l(n_\tau)}\mathbbm{1}_{\mathcal{B}_i}[q(t_{i+n_\tau})]\mathbbm{1}_{\mathcal{B}_j}[q(t_{i+n_{\tau'}})]\mathbbm{1}_{\Xi_0}[q(t_i)]}{\delta
  q^2\sum_{i=1}^{n_l(n_\tau)}\mathbbm{1}_{\Xi_0}[q(t_i)]}.
\end{equation}
Introducing $\delta n=n_\tau-n_a$ the time-asymmetry index is in turn determined as a double sum
\begin{equation}
  \AI_{\Xi_0}(t_a,\tau)=\delta
  q^2\sum_{i,j}G(q_i,\delta
  n+n_a,q_j,n_a|q_0\in\Xi_0)\log\frac{G(q_i,\delta
    n+n_a,q_j,n_a|q_0\in\Xi_0)}{G(q_i,\delta n+n_a|q_j)G(q_i,n_a|q_0\in\Xi_0)},
\end{equation}
whereas the normalized aging correlation function (\eql\ref{a_corr}) is determined
according to
\begin{eqnarray}
  \hat{C}_{t_a}(\tau)&=&\frac{\langle q(t_{\delta n + n_a}) q(t_{n_a})\rangle-\langle q(t_{\delta n + n_a})\rangle\langle q(t_{n_a})\rangle}{\langle q^2(t_{n_a})\rangle-\langle q(t_{n_a})\rangle^2},\\
\langle q(t_{\delta n + n_a})
    q(t_{n_a})\rangle &\equiv& \sum_{i,j}q_iq_jG(q_i,\delta
    n+n_a,q_j,n_a|q_0\in\Xi_0),\\
\langle q(t_i)\rangle &\equiv& \sum_{i,j}q_iG(q_i,n_i|q_0\in\Xi_0),    
\end{eqnarray}
where $\tau=\delta n\Delta t$ and $t_a=n_a\Delta t$ and we note that
by construction (i.e. due to the centering of data) $\langle
q(t_{i})\rangle=0, \forall i$. 100 bins in
$q(t_{\delta n + n_a})$ and 100 bins in
$q(t_{n_a})$ were used for each combination of $\tau$ and $t_a$ to determine $G(q_i,\delta
  n+n_a,q_j,n_a|q_0\in\Xi_0),G(q_i,\delta n+n_a|q_j)$ and
  $G(q_i,n_a|q_0\in\Xi_0)$ and in turn
  $\AI_{\Xi_0}(t_a,\tau)$.
 
 \subsubsection{DNA-hairpin}\label{DNA}
 Dual optical tweezers data of the DNA hairpin were kindly provided by
 the Woodside group \cite{Woodside2} in the form of a
 constant trap measurements of the DNA hairpin 30R50T4, sampled at
 $400$ kHz, for trap stiffness $0.63 pN/nm$ in one optical trap and
 $1.1 pN/nm$ in the other. The time series was $2.75\cdot 10^4
 \mathrm{ms}$ long. The normalized aging correlation function
 $C_{t_a}(\tau)$ and dynamical time asymmetry index $\AI$ are depicted in
 Figs.~3c and 4b in the manuscript. Here, we additionally present in
 Fig.~\ref{hairpin_fig}, for illustrative purposes and for the sake of
 completeness, the density of the invariant (equilibrium) measure
 $P_\eeq(q)$ and exemplary two-point conditional probability
 $G(q,t|q_0\in\Xi_0)$ and the three-point conditional density $G(q,\tau+t_a,q',t_a|q_0\in\Xi_0)$,
 respectively, for various $t$. The histograms in the
 relative deviations $q(t_i)=q_{\mathrm{raw}}(t_i)-\overline{q}$ were
 determined by binning the interval from -25 nm to +15 nm into 100 bins
 and $\overline{q}=3.47$ nm. These probability density functions are shown
 in Fig.~\ref{hairpin_fig}.
 \begin{figure}[ht!!]
 \begin{center}
  \includegraphics[width=0.8\textwidth]{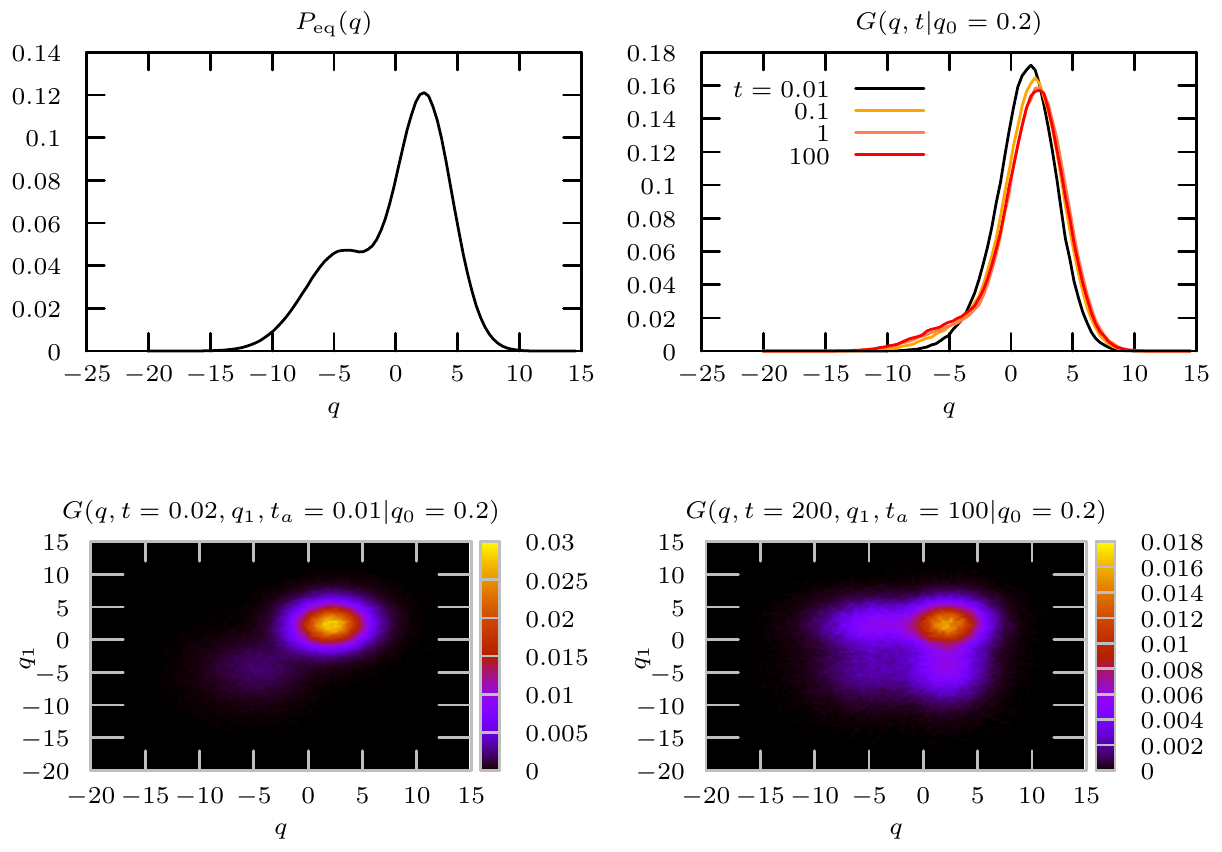}
  \caption{\textbf{DNA Hairpin.} In the top left panel depicts the density of the
    equilibrium measure $P_\eeq(q)$
    of the centered time series
    $q(t_i)=q_{\mathrm{raw}}(t_i)-\overline{q}$, while the top right
    shows a two-point conditional probability density function
    $G(q,t|q_0\in\Xi_0)$ for different values of $t$, where
    $\Xi_0=[0.2-0.2,0.2+0.2]$ nm. The bottom panels depict the
    three-point density $G(q,t,q',t_a|q_0\in\Xi_0)$ at
    different $\tau$ and $t_a$ evolving from the same initial condition.}
  \label{hairpin_fig}
  \end{center}
 \end{figure}
 \begin{figure}[ht!!]
  \begin{center}
   \includegraphics[width=0.8\textwidth]{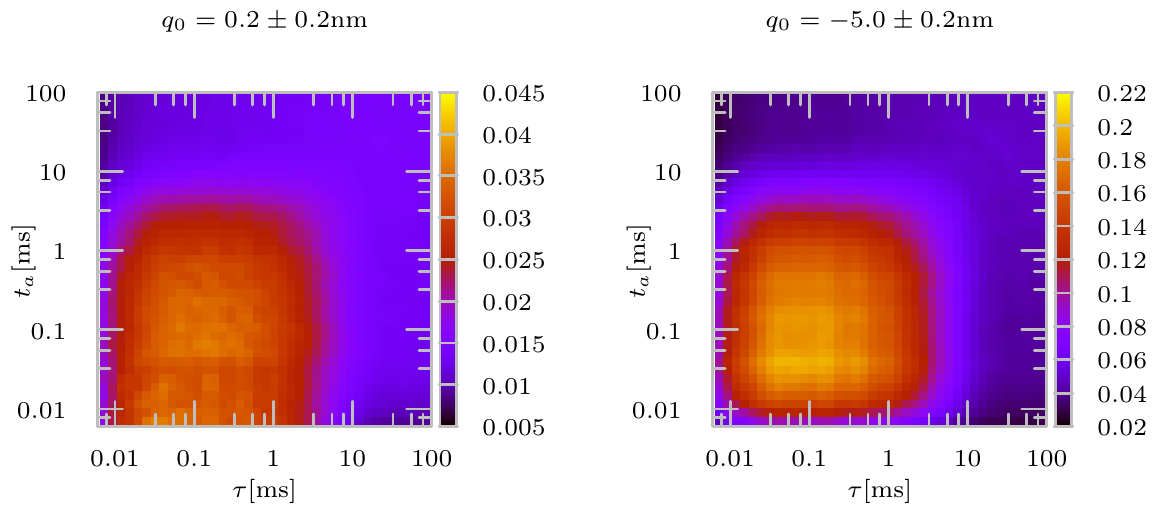}
   \caption{\textbf{DNA Hairpin, second example.} The time asymmetry index for the DNA hairpin data as in
     Fig.~\ref{hairpin_fig} but with the initial condition
     $\Xi_0=[-5-0.2,-5+0.2]$ nm.}
   \label{hairpinaddkl_fig}
  \end{center}
 \end{figure}

 The density of the invariant measure of the extension of the hairpin
 is bimodal, reflecting the existence of two long-lived conformational states
 (Fig.~\ref{hairpinaddkl_fig}, top left). 
The evolution of the conditional two-point
conditional probability density for an ensemble of trajectories
starting at the typical distance $q_{\mathrm{peak}}$,
$G_{\eeq}(q,t|q_{\mathrm{peak}})$ (Fig.~\ref{hairpinaddkl_fig}, top right) evolves smoothly towards $P_\eeq(q)$
with a relaxation time $t_{\mathrm{rel}}=\lambda_1^{-1}\approx 15 ms$,
and nicely depicts the onset of conformational transitions (see red
line). 

Another striking feature of hairpin dynamics is seen in the
corresponding three-point density,
$G_{\eeq}(q,t,q',t_a|q_{\mathrm{peak}})$ (Fig.~\ref{hairpinaddkl_fig},
bottom), which depicts, alongside the linear correlations along and near
the diagonal $q=q'$ that were also present in the Rouse polymer and
tagged-particle diffusion in a single-file, prominent non-linear correlations
(see off-diagonal peaks). This readily reveals that temporal
correlations in the motion persist beyond the time-scale of
conformational transitions, that is, the hairpin relaxation dynamics post transition remembers
the configurations prior to the transition even on time-scales of
$\gtrsim 200$ ms, which reflects a very long range of broken
Markovianity.  However, a comparison with the corresponding time asymmetry
index in Fig.~3c in the main text shows that at aging times $t_a=100$
ms the dynamics is already time-translation invariant. This is a nice and clear practical demonstration of the
important conceptual difference
between the notion of relaxation with a broken time-translation
invariance and memory effects in time-translation
invariant relaxation of a low-dimensional physical
observable.

In order to demonstrate the robustness of these
observations with respect to specific the initial condition $q_0=1(0)$
(as long as $p_0(q_0)\ne P_\eeq(q_0)$ that is) we also present in Fig.~\ref{hairpinaddkl_fig} the results for a different set of initial conditions.
 The results in Fig.~\ref{hairpinaddkl_fig} show qualitatively the same
 features and are fully consistent with the
 statements in the manuscript.

 Finally, we asses the statistical uncertainty of determining
 $\AI$ from the experimental time-series. We do so by
 performing the analysis on an ensemble of trajectories obtained by randomly
 removing 10 (from the total of 50, i.e. $20$\%\ of the data) trajectories and averaging over 20
 repetitions of a data-set created in this manner. We quantify the
 statistical error by determining the local standard deviation of
 $\AI$, i.e. $\sigma_{\AI}=\sqrt{N_i^{-1}\sum_{i=1}^{N_i}\AI_i^2-(N_i^{-1}\sum_{i=1}^{N_i}\AI_i)^2}$
 where $N_i$=20.
 The results are
 shown in Fig.~\ref{err_an} and depict a local error that is smaller
 than 1\%. 
 \begin{figure*}[ht!!]
  \begin{center}
   \includegraphics[width=0.8\textwidth]{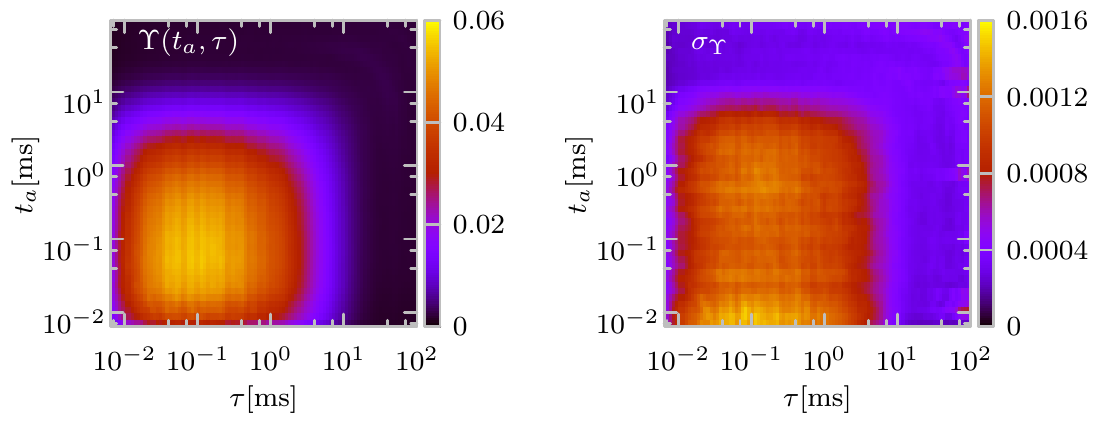}
   \caption{\textbf{Statistical error in the DNA-hairpin analysis.} a)
   average $\AI(\tau,t_a)$ determined from an ensemble with forced
   under-sampling (i.e. by omiting 20\% of the data); b) the standard
   deviation of the local $\AI(\tau,t_a)$.}
   \label{err_an}
  \end{center}
 \end{figure*}
 
  \subsubsection{Yeast 3-phosphoglycerate kinase (PGK)}\label{PGK}
  Atomistic Molecular Dynamics (MD) simulation of yeast PGK were
  carried out by Hu et al. \cite{Jeremy}, starting from the
  PDB structure 3PGK with a duration of $1.71\cdot 10^5
  \mathrm{ps}$. The observable $q(t_i)$ refers here to the distance
  between the center of mass of the N-terminal domain (residues 1-185)
  and the center of mass of C-terminal domain (residues 200-389).  In
  Fig.~\ref{3PGK_fig} we depict the density of the invariant (equilibrium) measure
 $P_\eeq(q)$ and exemplary two-point conditional probability
 $G(q,t|q_0\in\Xi_0)$ and the three-point conditional density $G(q,\tau+t_a,q',t_a|q_0\in\Xi_0)$,
 respectively, for various $t$. The histograms in the
 relative deviations $q(t_i)=q_{\mathrm{raw}}(t_i)-\overline{q}$ were
 determined by binning the interval from -0.2 nm to +0.3 nm into 100 bins
 and $\overline{q}=0.67$ nm.\\
  \begin{figure}[ht!!]
   \begin{center}
    \includegraphics[width=0.9\textwidth]{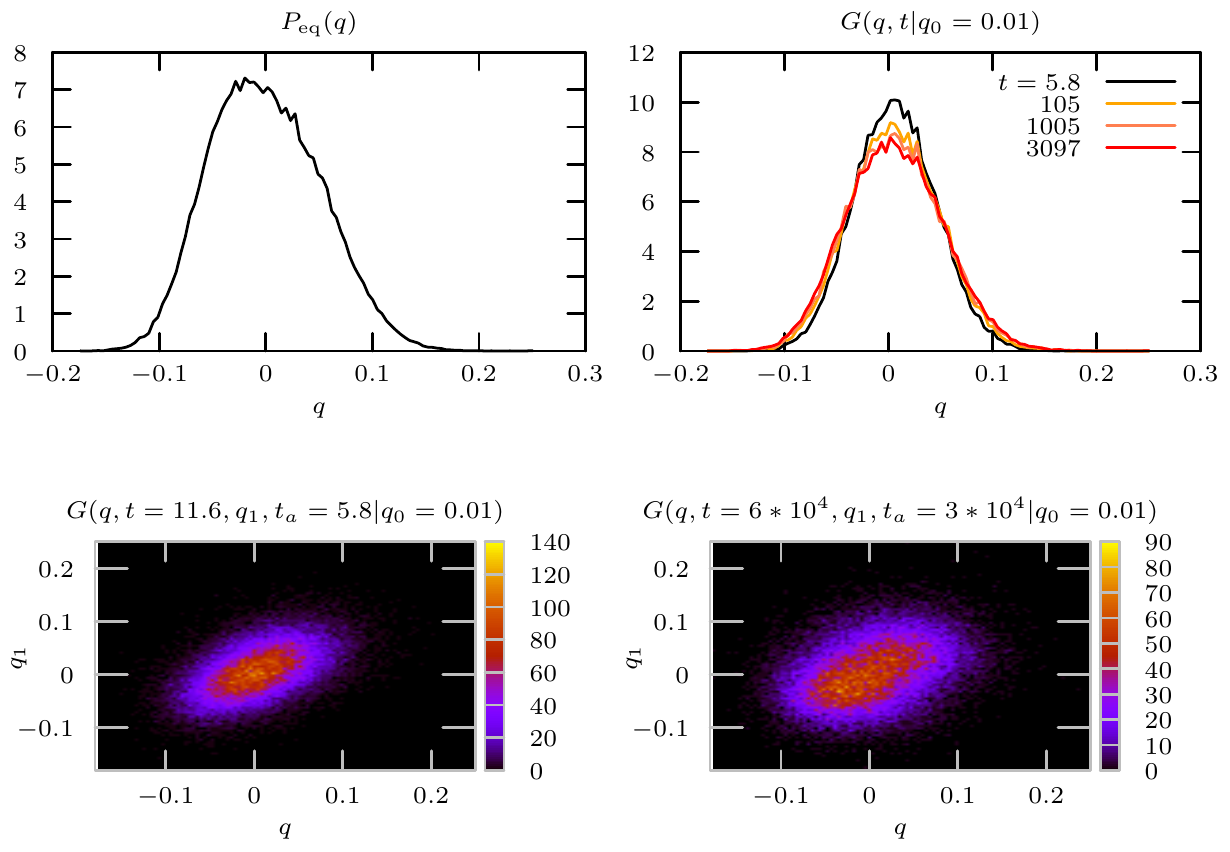}
    \caption{\textbf{PGK.} The top left panel depicts the density of the
    equilibrium measure $P_\eeq(q)$
    of the centered time series
    $q(t_i)=q_{\mathrm{raw}}(t_i)-\overline{q}$, while the top right
    shows a two-point conditional probability density function
    $G(q,t|q_0\in\Xi_0)$ for different values of $t$, where
    $\Xi_0=[0.01-0.005,0.01+0.005]$ nm. The bottom panels depict the
    three-point density $G(q,t,q',t_a|q_0\in\Xi_0)$ at
    different $\tau$ and $t_a$ evolving from the same initial condition.}
    \label{3PGK_fig}
   \end{center}
  \end{figure}
The density of the invariant measure $P_\eeq(q)$
(Fig.~\ref{3PGK_fig}, top left) is unimodal and effects of a poorer
statistics are readily discernable through the roughness of the
curve. The evolution of the conditional two-point
conditional probability density for an ensemble of trajectories
starting at the typical distance $q_{\mathrm{peak}}$,
$G_{\eeq}(q,t|q_{\mathrm{peak}})$ is shown in Fig.~\ref{3PGK_fig} (top
right panel) and reveals that the the dynamics along $q$ is strongly
localized (i.e. $G_{\eeq}(q,t|q_{\mathrm{peak}}$ barely changes between
$t=100$ ps and $t=3000$). Note that PGK did not relax within the duration of the trajectory.  The
corresponding three-point density,
$G_{\eeq}(q,t,q',t_a|q_{\mathrm{peak}})$ (Fig.~\ref{3PGK_fig},
bottom), shows that the observable almost does not relax at all within
$6\times 10^4$ ps (compare left and right panel). As in the case of
the Rouse polymer $G_{\eeq}(q,t,q',t_a|q_{\mathrm{peak}})$ shows
strong and long-lasting correlations between positions, and comparison
with the dynamical time asymmetry index in Fig.~3d in the main text reveals that the
dynamics has a strongly time-translation invariance. These findings corroborate the original
analysis of Hu et al. \cite{Jeremy} who observed aging effects. 

Similar to the DNA hairpin we also present in Fig.~\ref{3pgkaddkl_fig} the
results for the dynamical time asymmetry  index for a different set of initial conditions and two different
choices of the observable $q(t)$ for PGK, which demonstrate the
robustness of the results. 
   \begin{figure}[ht!!]
  \begin{center}
   \includegraphics[width=0.9\textwidth]{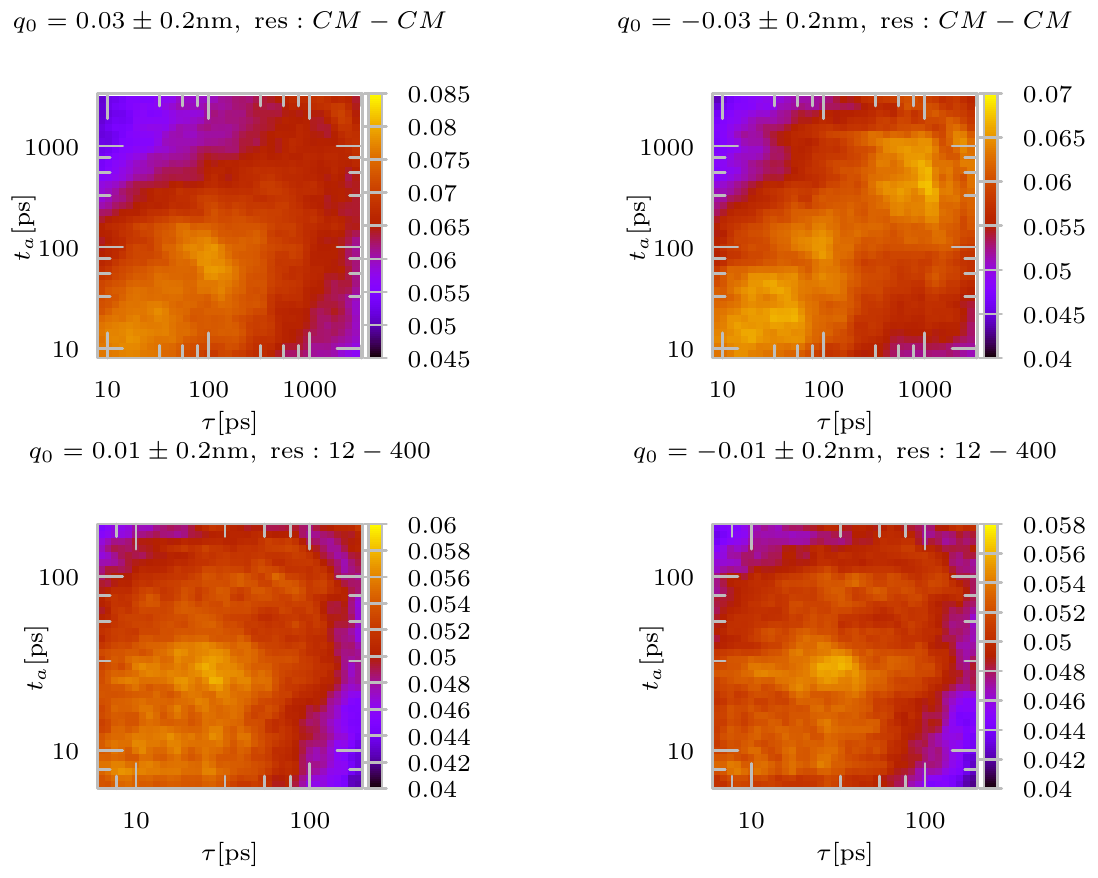}
   \caption{\textbf{PGK, second example.} Top: dynamical time asymmetry  index for yeast PGK when $q(t)$ corresponds to
     distance between the center of masses of the N- and C- terminal
     domains (respectively residues 1-185 and 200-389) for a pair
     different initial conditions. Bottom: dynamical time asymmetry  index for yeast PGK when $q(t)$ corresponds to the distance between two specific residues (the 12th and the 400th) for two different initial conditions.}
   \label{3pgkaddkl_fig}
  \end{center}
 \end{figure}

PGK  obviously does not relax within the duration of the trajectory and
more generally it is conceivable that larger, complex proteins do not relax at all during
their life-time \cite{Jeremy}, which makes them virtually 'forever aging'
\cite{SRalf}, which may have important consequences for their
biological function. Such aging effects on function were
observabed in single-enzyme turnover statistics
\cite{Xie1,Xie2,Xie3} and have so-far been
rationalized only with ad-hoc phenomenological models
\cite{Xie2,Xie3}. The present theoretical framework
provides a unifying mechanistic understanding dynamics with broken
time-translation invariance in soft and
biological matter and will pave the way for deeper and more systematic
investigations of the potential biological relevance of memory and dynamical time asymmetry for
enzymatic catalysis. 

\clearpage
\newpage
\twocolumngrid
\bibliography{bib_main}
\end{document}